\documentclass[acmsmall]{acmart}

\usepackage[utf8]{inputenc}

\usepackage{amsmath,amssymb}
\usepackage{calligra}
\let\savebigtimes\bigtimes
\let\bigtimes\relax
\usepackage{mathabx}
\let\bigtimes\savebigtimes
\usepackage{mathabx}
\usepackage{xspace}
\usepackage{subfig,caption}
\usepackage{todonotes}
\usepackage{paralist}
\usepackage{tikz,tikz-qtree}
\usetikzlibrary{shadows,mindmap,trees,shapes,positioning,arrows,decorations,backgrounds,positioning,fit,automata,matrix}
\usetikzlibrary{decorations.pathreplacing}
\usetikzlibrary{decorations.pathmorphing, patterns,shapes}
\usetikzlibrary{shapes}
\usetikzlibrary{arrows,automata}
\usepackage{cancel}

\usepackage{booktabs}
\usepackage{tabularx}

\definecolor{SkyBlue}{rgb}{0.53, 0.81, 0.92}
\definecolor{MediumSpringGreen}{rgb}{0.0, 0.98, 0.6}
\definecolor{RoyalBlue}{rgb}{0.25, 0.41, 0.88}
\definecolor{SpringGreen}{rgb}{0.0, 1.0, 0.5}
\definecolor{TurquoiseBlue}{rgb}{0.0, 1.0, 0.94}
\definecolor{Yellow}{rgb}{1.0, 1.0, 0.0}
\definecolor{Crimson}{rgb}{0.86, 0.08, 0.24}
\definecolor{LimeGreen}{rgb}{0.2, 0.8, 0.2}
\definecolor{PaleTurquoise}{rgb}{0.28, 0.82, 0.8}
\definecolor{Green}{rgb}{0.0,0.5,0.0}

\newcommand{\ie}{i.e.\xspace}
\newcommand{\Ie}{I.e.\xspace}

\renewcommand{\epsilon}{\varepsilon}
\renewcommand{\phi}{\varphi}

\newcommand\textbfit[1]{{\bf\em #1}}
\newcommand{\defin}[1]{\textbfit{\boldmath #1}}

\newcommand{\cone}[1]{\mathrm{Cone}(#1)}
\newcommand{\qini}{q_{\mathrm{ini}}}

\newcommand{\col}{\mathrm{Col}}
\newcommand{\colors}{C}

\newcommand{\Lqual}[1]{L^{\mathrm{Acc}}_{\mathrm{=1}}(#1)}
\newcommand{\prefix}{\sqsubseteq}
\newcommand{\prefixstrict}{\sqsubset}

\newcommand{\Eloise}{\'Elo\"ise\xspace}
\newcommand{\Abelard}{Ab\'elard\xspace}

\newcommand{\Nature}{Nature\xspace}
\newcommand{\Ei}{\mathbf{E}}
\newcommand{\Ai}{\mathbf{A}}

\newcommand{\Ni}{\mathbf{N}}
\renewcommand{\Ei}{\mathrm{E}}
\renewcommand{\Ai}{\mathrm{A}}

\renewcommand{\Ni}{\mathrm{N}}

\newcommand{\strat}{\phi}

\newcommand{\arena}{\mathcal{G}}
\newcommand{\game}{\mathbb{G}}
\newcommand{\play}{\lambda}
\newcommand{\arenaNR}{\widehat{\mathcal{G}}}
\newcommand{\gameNR}{\widehat{\mathbb{G}}}
\newcommand{\playNR}{\widehat{\lambda}}
\newcommand{\phiNR}{\widehat{\phi}}
\newcommand{\Plays}[1]{\mathrm{Plays}}
\newcommand{\Traces}[1]{\mathrm{Traces}}

\newcommand{\Value}[2]{\mathrm{Val}_{#2}(#1)}

\newcommand{\branch}{\pi} 
\newcommand{\lbg}{\mu}
\newcommand{\run}{\rho}

\newcommand{\Evei}{\Ei}

\newcommand{\outcomes}[3]{\mathrm{Outcomes}_{#1}^{#2,#3}}

\newcommand{\mesure}[3]{\lbg_{#1}^{#2,#3}}

\newcommand{\Dom}{\mathrm{Dom}}

\newcommand{\GValue}[1]{\mathrm{Val}(#1)}

\newcommand{\LRejAtMostCount}[1]{L^{\mathrm{Rej}}_{\leq \mathrm{Count}}(#1)}

\newcommand{\LAccUnc}[1]{L^{\mathrm{Acc}}_{Uncount}(#1)}

\newcommand{\VE}{V_{\Ei}}
\newcommand{\VA}{V_{\Ai}}

\newcommand{\VN}{V_{\Ni}}
\newcommand{\WC}{\Omega}
\newcommand{\WCNR}{\widehat{\Omega}}

\newcommand{\eg}{\emph{e.g.}\xspace}
\newcommand{\resp}{resp. \xspace} 

\newcommand{\LLarge}[1]{L^{\mathrm{Acc}}_{\mathrm{Large}}(#1)}

\newcommand{\CL}[1]{CardLeak(#1)}
\newcommand{\LVal}[1]{LeakVal(#1)}
\newcommand{\cardinal}[1]{Card(#1)}
\newcommand{\NRmapF}{\tau}
\newcommand{\NRmap}[1]{\NRmapF(#1)}

\newcommand{\treeGame}[3]{T_{#1}^{#2,#3}}

\newcommand{\graphT}{\widetilde{{G}}}
\newcommand{\VT}{\widetilde{V}}
\newcommand{\ET}{\widetilde{E}}
\newcommand{\arenaT}{\widetilde{\mathcal{G}}}
\newcommand{\gameT}{\widetilde{\mathbb{G}}}
\newcommand{\playT}{\widetilde{\lambda}}
\newcommand{\phiT}{\widetilde{\phi}}
\newcommand{\WCT}{\widetilde{\WC}}
\newcommand{\GammaT}{\widetilde{\Gamma}}
\newcommand{\DeltaT}{\widetilde{\Delta}}
\newcommand{\simT}{\approx}

\newcommand{\graphL}{\widehat{G}}
\newcommand{\VL}{\widehat{V}}
\newcommand{\EL}{\widehat{E}}
\newcommand{\arenaL}{\widehat{\mathcal{G}}}
\newcommand{\gameL}{\widehat{\mathbb{G}}}

\newcommand{\WCL}{\widehat{\WC}}

\newcommand{\encod}[1]{[\![ #1 ]\!]}

\newcommand{\gameP}{\mathbb{P}}
\newcommand{\gamePNR}{\widehat{\mathbb{P}}}

\newcommand{\VF}{\widecheck{V}}
\newcommand{\EF}{\widecheck{E}}

\newcommand{\playF}{\widecheck{\lambda}}
\newcommand{\phiF}{\widecheck{\phi}}

\newcommand{\graphFk}{\widecheck{G}_k}

\newcommand{\arenaFk}{\widecheck{\mathcal{G}}_k}
\newcommand{\gameFk}{\widecheck{\mathbb{G}}_k}
\newcommand{\playFk}{\widecheck{\lambda}}
\newcommand{\phiFk}{\widecheck{\phi}}
\newcommand{\WCFk}{\widecheck{\WC}_k}

\setcopyright{acmcopyright}
\acmJournal{TOCL}
\acmYear{2020}\acmVolume{21}\acmNumber{3}\acmArticle{21}\acmMonth{2}
\acmDOI{10.1145/3377137}

\begin{document}

\theoremstyle{acmdefinition}
\newtheorem{remark}[theorem]{Remark}

\title{How Good Is a Strategy in a Game With Nature?}

\author{Arnaud Carayol}
\affiliation{%
  \institution{CNRS, LIGM (Université Paris Est \& CNRS)}
  \streetaddress{5 boulevard Descartes — Champs sur Marne}
  \city{Marne-la-Vallée Cedex 2}
  \postcode{77454}
  \country{France}
}
\email{Arnaud.Carayol@univ-mlv.fr}

\author{Olivier Serre}
\affiliation{%
  \institution{Université de Paris, IRIF, CNRS}
  \streetaddress{Bâtiment Sophie Germain, Case courrier 7014, 
8 Place Aurélie Nemours}
  \city{Paris Cedex 13}
  \postcode{75205}
  \country{France}
}
\email{Olivier.Serre@cnrs.fr}

\begin{abstract}
We consider games with two antagonistic players —~\Eloise (modelling a program) and \Abelard (modelling a byzantine environment)~— and a third, unpredictable and uncontrollable player, that we call \Nature. Motivated by the fact that the usual probabilistic semantics  very quickly leads to undecidability when considering either infinite game graphs or imperfect-information, we propose  two alternative semantics that leads to decidability where the probabilistic one fails: one based on counting and one based on topology. 
\end{abstract}



\maketitle

\section{Introduction}

An important problem in computer science is the specification and the verification of systems allowing \emph{non-deterministic} behaviours. A non-deterministic behaviour can appear in several distinct contexts: 
\begin{enumerate}[(i)]
\item controllable behaviours (typically arising when the \emph{program} is not fully specified, permitting to later restrict it); 
\item uncontrollable possibly byzantine behaviours  (typically arising from interactions of the program with its \emph{environment}, \eg a user); 
\item uncontrollable unpredictable behaviours (usually arising from \emph{nature} often modelled by randomisation).
\end{enumerate} 

Here we do an explicit distinction between the environment and nature: while we cannot assume that a user will not be malicious, the situation with nature is different as we can accept a negligible set of bad behaviours which implicitly means that they are very unlikely to appear.
On top of this, one may also want to allow \emph{imperfect-information} (typically arising when the protagonists —~the program, the environment and nature~— share some public variables but also have their own private variables) and/or \emph{infinite} state systems (\eg arising when modelling recursive procedures).

As the above mentioned features are omnipresent in nowadays systems, their specification has already deserved a lot of attention and there are several robust abstract mathematical models for them. There are also work on the specification side (\ie on how to express a desirable behaviour of the system) and on the decidability of the fundamental question of whether a given specification is met by a given system. {Unfortunately, whenever one combines any kind of non-determinism (interaction of the program with both nature and an uncontrollable environment) with one of the two others (either imperfect-information or an infinite number of states) it directly leads to undecidability. Moreover, if ones to recover decidability while still considering infinite arenas strong unnatural restrictions are needed.} 

\emph{Two-player stochastic games on graphs} are a natural way to model such systems. 
In a nutshell, a stochastic game is defined thanks to a directed graph whose vertices have been partitioned among two antagonistic players —~\Eloise (modelling the program) and \Abelard (modelling the byzantine environment)~— and a third, unpredictable and uncontrollable player, that we call \Nature. The play starts with a token on a fixed initial vertex $v_0$ of the graph that is later moved by the players (the player owning the vertex where the token is, chooses a neighbour to which the token is moved to, and so on forever)  leading to an infinite path in the game graph. We are interested in \emph{zero-sum} games, \ie we consider a winning condition  $\WC$ consisting of a subset of plays and we say that a play is winning for \Eloise if it belongs to $\WC$ and otherwise it is winning for \Abelard. A game $\game$ is such a graph together with a winning condition.

In the previous model, \Nature usually comes with a probabilistic semantics (as in the seminal work of Condon~\cite{Condon92}), \ie any vertex controlled by \Nature is associated with a probability distribution over its neighbours and this probability distribution is used to pick the next move when the token is on the corresponding vertex. The central concept is the one of a strategy, which maps any prefix of a play to the next vertex to move the token to. Once  a strategy $\phi_\Ei$ for \Eloise and a strategy $\phi_\Ai$  for \Abelard have been fixed, the set of all possible plays in the game where the players respect their strategies can be equipped with a probability measure $\mu_{v_0}^{\phi_\Ei,\phi_\Ai}$, and one can therefore define the \defin{value} of the game as  ($\phi_\Ei$ and $\phi_\Ai$ range over \Eloise and \Abelard strategies respectively)
$$\GValue{\game} = \sup_{\phi_\Ei} \inf_{\phi_\Ai} \{\mu_{v_0}^{\phi_\Ei,\phi_\Ai}(\WC) \}$$
Then, the following questions are of special interest.
\begin{enumerate}
\item “\emph{Decide whether the value of the game is larger than some given threshold $\eta$}” and its qualitative weakening “\emph{Decide if the value is equal to $1$}”.
\item “\emph{When exists, compute an optimal strategy for \Eloise}” where an optimal strategy $\phi_\Ei$ for \Eloise is one such that $\GValue{\game} = \inf_{\phi_\Ai} \{\mu_{v_0}^{\phi_\Ei,\phi_\Ai}(\WC) \}$ (note that such a strategy may not exist even if the graph is finite).
\end{enumerate}

If the game is played on a finite graph and the winning condition is $\omega$-regular, all those questions  can be answered and algorithms are known and their complexities, depending on the winning condition, range from P to PSPACE  (see \eg \cite{ChatterjeePHD} for an overview). 

Unfortunately the landscape drastically changes as soon as one either considers infinite game graphs and/or imperfect-information (\ie instead of knowing the exact state of the system, each player only knows that it belongs to some equivalence class). In particular we have the following undecidability (somehow minimal) results:
\begin{itemize}
\item If the game graph is a pushdown graph, then even if \Abelard is not part of the game, the qualitative analysis of reachability games is undecidable~\cite{EtessamiY05}.
\item If \Eloise has imperfect-information then, even if the graph is finite and \Abelard is not part of the game, {almost-sure winning is undecidable for co-Büchi games~\cite{BGB12}}.
\end{itemize}

In this work, we propose two alternative semantics that lead to decidable problems where the previous probabilistic approach fails. The main idea is to evaluate (for fixed strategies of \Eloise and \Abelard) how {“small” the set of resulting losing plays for \Eloise is}. 

Our first setting is based on \emph{counting}. In order to evaluate how good a situation is for \Eloise  (\ie using some strategy $\phi_\Ei$ against a strategy $\phi_\Ai$ of \Abelard)  we simply count how many losing plays there are: the fewer the better. Of special interest are those strategies for which, against any strategy of \Abelard, the number of losing plays is at most countable. The idea of counting can be traced back to the work in \cite{BeauquierNN91,BN95} on automata with cardinality constraints. There is also work on the logical side with decidable results but that do not lead to efficient algorithms~\cite{BKR10}.

Our second setting is based on \emph{topology}. In order to evaluate how good a situation is for \Eloise  (\ie using some strategy $\phi_\Ei$ against a strategy $\phi_\Ai$ of \Abelard)  we use a topological notion of “bigness“/“smallness“ given by the concept of large/meager set. The idea of using topology was considered previously in the context of \emph{finite} Markov chains \cite{VolzerV12} and \emph{finite} Markov decision processes \cite{AsarinCV10}.

The approach we follow in this paper is to provide reductions to games that do not involve Nature. More precisely, a typical result will be to provide a transformation of a game involving \Nature into a new game that no longer involves \Nature and that is such that the algorithmic question considered on the original game reduces to another question on the new (two-player) game. In particular when the latter is decidable it implies decidability of the original problem. In order to be as general as possible we try to impose as few restrictions as possible on the underlying graph of the game (typically we allow infinite graphs) and on the winning condition (many results are obtained for Borel conditions): this permits to obtain decidability results for a wide range of games.

The paper starts with definitions of basic objects in Section~\ref{section:preliminaries} while Section~\ref{section:perfect} introduces the different settings we consider in this paper. In Section~\ref{section:perfect-decision} we focus on perfect-information games and we provide reductions for both the cardinality setting and the topological setting; algorithmic consequences as well as consequences for automata on infinite trees are then discussed in Section~\ref{section:consequences-perfect}. We then turn to the imperfect-information setting in Section~\ref{section:imperfect} and discuss consequences of our results in Section~\ref{section:consequences-imperfect}. Finally we summarize our results and propose some perspectives in Section~\ref{section:conclusion}.

\section{Preliminaries}\label{section:preliminaries}

We now introduce basic concepts that will be used all along the paper.

\subsection{Sets}

Let $X$ be a set, we denote by $\cardinal{X}$ its cardinal. In this work, we will only need to consider\footnote{This is a consequence of the structure of the sets we consider (see Proposition~\ref{prop:continuumHypothesis}) and has nothing to do with the continuum hypothesis.} 
 finite cardinals, $\aleph_0$ (the cardinality of the natural numbers) or $2^{\aleph_0}$ (the cardinality of the real numbers). A set is \defin{countable} if its cardinal is smaller or equal than $\aleph_0$ (equivalently, the set is either finite or in bijection with the natural numbers).

If $S_1,\dots,S_k$ are sets we denote by $S=S_1\uplus\cdots\uplus S_k$ the fact that they form a partition of $S$, \ie $S=S_1\cup\cdots\cup S_k$ and $S_i\cap S_j=\emptyset$ for every $i\neq j$.

Let $S$ be a set, $\sim$ be an equivalence relation on $S$ and $s$ be some element in $S$. Then, we denote by $[s]_{/_\sim}=\{s'\mid s\sim s'\}$ the equivalence class of $s$ for relation $\sim$ and by $S_{/_\sim}$ the set of equivalence classes of $\sim$ on elements of $S$.

\subsection{Words}

Let $A$ be a (possibly infinite) set seen here as an \defin{alphabet}. We denote by $A^*$ the set of finite words over the alphabet $A$ and by $A^\omega$ the set of infinite words over the alphabet $A$. If $u$ is a word we denote by $|u|\in\mathbb{N}\cup\{\omega\}$ its length. We denote by $\epsilon$ the empty word and we let $A^+=A^*\setminus\{\epsilon\}$. 

If $u\in A^*$ and $v\in A^*\cup A^\omega$ we denote by $u\cdot v$ {(or simply $uv$)} the (possibly infinite) word obtained by concatenating $u$ and $v$. A word $u\in A^*$ is a \defin{prefix} of a word $w\in A^*\cup A^\omega$ if there exists some $v\in A^*\cup A^\omega$ such that $w=u\cdot v$, and we denote this situation by $u\prefix w$; moreover if $u\neq w$ we say that $u$ is a \emph{strict} prefix (denoted by $u\prefixstrict w$). 
 A set $S\subseteq A^*$ is \defin{prefix-closed} if for all $u\in S$ and $v\prefix  u$ one has $v\in S$.

 Let $(u_i)_{i\geq 0}$ be a sequence of finite words in $A^*$ such that for all $i\geq 0$ one has $u_i\prefix u_{i+1}$ and for infinitely many $i\geq 0$ one has $u_i\prefixstrict u_{i+1}$. 
We define its \defin{limit} $u_\infty \in  A^\omega$ as the unique infinite word such that for all $i \geq 0$, $u_i\prefixstrict u_\infty$.
Equivalently, $u_\infty=a_1a_2a_3\cdots\in A^\omega$ where for all $k\geq 1$, $a_k$ is the $k$-th letter of any $u_i$ such that $|u_i|\geq k$.

\subsection{Trees}

{In this paper we consider various notions of trees that we introduce now (see Figure~\ref{figure:basic-tree} for some illustrations).} Let $D$ be a (countable) set of directions; a \defin{$D$-tree} (or simply a \defin{tree} when $D$ is clear) is a prefix-closed subset of $D^*$. A $D$-tree is \emph{complete} if it equals $D^*$; it is \emph{binary} if $\cardinal{D}=2$ (and in general one identifies $D$ with $\{0,1\}$). 

For a given tree $T$, we refer to any element $u\in T$ as a \defin{node}; if $T=\{0,1\}^*$ is the complete binary tree, we refer to $u\cdot0$ (\resp $u\cdot 1$) as the left (\resp right) son of $u$. The node $\epsilon$ is called the \emph{root}.

In the sequel we implicitly assume that the trees we consider do not contain leaves, \ie for every node $u\in T$ there is some direction $d\in D$ such that $ud\in T$.

{An (infinite)} \defin{branch} in a $D$-tree $T$ is an infinite word $\branch\in D^\omega$ such that there is an increasing  (for the prefix ordering) sequence of nodes $(u_i)_{i\geq 0}$ whose limit is $\branch$. A node $u$ belongs  to a branch $\branch$ whenever $u\prefixstrict \branch$. Branches in the complete $D$-tree exactly coincide with $D^\omega$. For a node $u\in T$, the \defin{cone} ${\mathrm{Cone}_{T}(u)}$ is defined as the set of branches of $T$ passing through $u$ (\ie $\mathrm{Cone}_{T}(u)=\{ \branch \mid \branch\;\textrm{branch of $T$ and}\;\; u\prefixstrict \branch\}$).

Let $A$ be a (countable) alphabet; an \defin{$A$-labelled tree} $t$ is a total function $t:Dom\rightarrow A$ where $Dom$ is a tree. For a node $u\in Dom$ we call $t(u)$ the \emph{label} of $u$; and for a branch $\branch=\branch_0 \branch_1\cdots$ {of the tree $Dom$}, we call $t(\branch_0)t(\branch_0\branch_1)t(\branch_0\branch_1\branch_2)\cdots\in A^\omega$ the label of $\branch$.
For a node $u\in Dom$ we let $t[u]$ be the \defin{subtree} rooted at $u$, \ie $t[u]:\Dom'\rightarrow A$ with $\Dom'=\{v\mid u\cdot v\in\Dom\}$ and $t[u](v) = t(uv)$. Finally we call an \defin{$A$-labelled $D$-tree} an $A$-labelled tree whose domain is a $D$-tree.

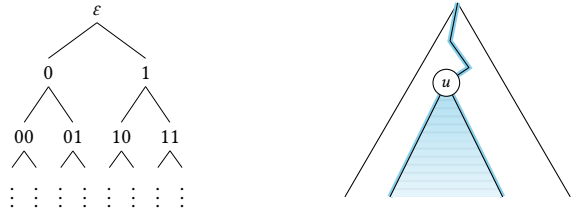
\begin{figure}
\centering
\begin{tikzpicture}
\begin{scope}[scale=.8,transform shape]
	\node(arbre) at (-6,-1.7){\Tree 
[.{$\epsilon$ }
	[.$0$ 
		[.$00$  [ .$\vdots$ ] [ .$\vdots$ ]] [.$01$ [ .$\vdots$ ] [ .$\vdots$ ]]
	] 
	[.$1$
		[.$10$  [ .$\vdots$ ] [ .$\vdots$ ]] [.$11$  [ .$\vdots$ ] [ .$\vdots$ ] ]
	]
 ]
};	
\end{scope}
\begin{scope}[scale=.5]
\draw (0,0) -- (-3,-5.1) ;
\draw (0,0) -- (3,-5.1) ;

\draw[line width=2pt,color=SkyBlue] (0,0) -- (-0.2,-1) -- (0.3,-1.7)--(-0.3,-2.1) ;

\draw[line width=2pt,color=SkyBlue] (-0.3,-2.1) --  (-1.8,-5.1);
\draw[line width=2pt,color=SkyBlue] (-0.3,-2.1) -- (1.2,-5.1);
\fill[bottom color=SkyBlue!100!black!20,
top color=SkyBlue!100!black!70] (-0.3,-2.1) --  (-1.8,-5.1) -- (1.2,-5.1) -- cycle;

\node[draw,fill=white,circle,scale=.6] (u) at (-0.3,-2.1) {\large $u$};
\draw (0,0) -- (-0.2,-1) -- (0.3,-1.7)--(u) ;
\draw (u)  -- (-1.8,-5.1);
\draw (u) -- (1.2,-5.1);
	
\end{scope}
\end{tikzpicture}
\caption{On the left first levels of the complete binary tree; on the right the cone $\cone{u}$}\label{figure:basic-tree}
\end{figure}

\subsection{Graphs}
A (directed) \defin{graph} $G$ is a pair $(V,E)$ where $V$ is a {countable} set of \defin{vertices} and $E\subseteq V\times V$ is a set of \defin{edges}. For a vertex $v$, {we denote by  $E(v)$ the set of its successors $\{v'\mid (v,v')\in E\}$} and in the rest of the paper (hence, this is implicit from now on), we only consider graphs that have no dead-end, \ie such that $E(v)\neq\emptyset$ for all $v$.

\section{Perfect-information Games With Nature: Main Definitions and Concepts}
\label{section:perfect}

\subsection{Definitions}

In this paper, we are interested in games involving two antagonistic players —~\Eloise and \Abelard~— together with a third \emph{uncontrollable and unpredictable} player called \Nature. An \defin{arena} is a tuple $\arena=(G,\VE,\VA,\VN)$ where $G=(V,E)$ is a graph and $V=\VE\uplus \VA\uplus \VN$  is a partition of the vertices among the three players. We say that a vertex $v$ is owned by \Eloise (\resp by \Abelard, \resp by \Nature) if $v\in\VE$ (\resp $v\in\VA$, \resp $v\in\VN$).

\Eloise, \Abelard and \Nature play in $\arena$ by moving a pebble along edges. A
\defin{play} from an initial vertex $v_0$ proceeds as
follows: the player owning $v_0$ moves the pebble to a vertex $v_1\in E(v_0)$. Then, the player owning $v_1$ chooses a successor
$v_2\in E(v_1)$ and so on forever. As we assumed that there is no dead-end, a play is an
infinite word $v_0v_1v_2\cdots \in V^\omega$ such that for all $i \geq 0$, one has $v_{i+1}\in E(v_i)$. A \defin{partial play} is a prefix of a play, \ie it is a finite word $v_0v_1\cdots v_\ell \in V^*$  such that for all $0\leq i<\ell$, one has $v_{i+1}\in E(v_i)$.

\begin{example}\label{example:game-running1}
Consider the arena depicted in Figure~\ref{fig:example-game}, where we adopt the following convention: a vertex owned by \Eloise (\resp \Abelard, \resp \Nature) is depicted by a circle (\resp a square, \resp a diamond). The underlying graph is $G=(V,E)$ where $V=\{E_i,A_i,N_i\mid i\geq 0\}$ and $$E=\{(A_i,A_{i+1}),(A_i,N_i),(N_{i+1},N_{i}),(N_i,E_i),(N_i,N_i),(E_{i+1},E_{i}),(E_i,E_i)\mid i\geq 0\}$$	
Note that this graph is an example of a pushdown graph~\cite{MullerS85} (\ie it can be presented as the transition graph of a pushdown automaton).

The partition of the vertices among the player is given by $V_\Ei=\{E_i\mid i\geq 0\}$, $V_\Ai=\{A_i\mid i\geq 0\}$ and $V_\Ni=\{N_i\mid i\geq 0\}$.

The following sequence $\lambda$ is an example of a play in that arena $$\lambda=A_0A_1A_2N_2N_2N_2N_1E_1E_1E_1E_0E_0E_0\cdots$$
\end{example}

\begin{figure}[htb]
\begin{center}
\begin{tikzpicture}[>=stealth',thick,scale=1,transform shape]
\tikzstyle{Abelard}=[draw]
\tikzstyle{Eloise}=[draw,circle]
\tikzstyle{Nature}=[draw,diamond,scale = .65,font=\Large]
\tikzstyle{AbelardH}=[]
\tikzstyle{EloiseH}=[circle]
\tikzstyle{NatureH}=[diamond,scale = .65,font=\Large]
\tikzstyle{C0}=[fill=MediumSpringGreen]
\tikzstyle{C1}=[fill=Yellow]
\tikzstyle{C2}=[fill=Cyan]
\tikzset{every loop/.style={min distance=10mm,looseness=10}}
\tikzstyle{loopleft}=[in=150,out=210]
\tikzstyle{loopright}=[in=-30,out=30]
\node[Abelard,C0] (A1) at (1,0) {$A_0$};
\node[Abelard,C0] (A2) at (2.5,0) {$A_1$};
\node[Abelard,C0] (A3) at (4,0) {$A_2$};
\node[Abelard,C0] (A4) at (5.5,0) {$A_3$};
\node[AbelardH] (A5) at (7,0) {};
\node[Nature,C1] (N1) at (1,-1.5) {$N_0$};
\node[Nature,C1] (N2) at (2.5,-1.5) {$N_1$};
\node[Nature,C1] (N3) at (4,-1.5) {$N_2$};
\node[Nature,C1] (N4) at (5.5,-1.5) {$N_3$};
\node[NatureH] (N5) at (7,-1.5) {};
\node[Eloise,C0] (E1) at (1,-3) {$E_0$};
\node[Eloise,C1] (E2) at (2.5,-3) {$E_1$};
\node[Eloise,C1] (E3) at (4,-3) {$E_2$};
\node[Eloise,C1] (E4) at (5.5,-3) {$E_3$};
\node[EloiseH] (E5) at (7,-3) {};

\path[->] (A1) edge (N1);\path[->] (A1) edge (A2);
\path[->] (A2) edge (N2);\path[->] (A2) edge (A3);
\path[->] (A3) edge (N3);\path[->] (A2) edge (A3);
\path[->] (A4) edge (N4);\path[->] (A3) edge (A4);
\path[->,dotted] (A4) edge (A5);

\path[->] (N1) edge  [in=110,out=160, loop] ();\path[->] (N1) edge (E1);
\path[->] (N2) edge  [in=110,out=160, loop] ();\path[->] (N2) edge (N1);\path[->] (N2) edge (E2);
\path[->] (N3) edge  [in=110,out=160, loop] ();\path[->] (N3) edge (N2);\path[->] (N3) edge (E3);
\path[->] (N4) edge  [in=110,out=160, loop] ();\path[->] (N4) edge (N3);\path[->] (N4) edge (E4);
\path[->,dotted] (N5) edge (N4);

\path[->] (E1) edge  [loop below, loop] ();
\path[->] (E2) edge  [loop below, loop] ();\path[->] (E2) edge (E1);
\path[->] (E3) edge  [loop below, loop] ();\path[->] (E3) edge (E2);
\path[->] (E4) edge  [loop below, loop] ();\path[->] (E4) edge (E3);
\path[->,dotted] (E5) edge (E4);

\end{tikzpicture}
\end{center}
\caption{Example of an infinite arena}\label{fig:example-game}
\end{figure}
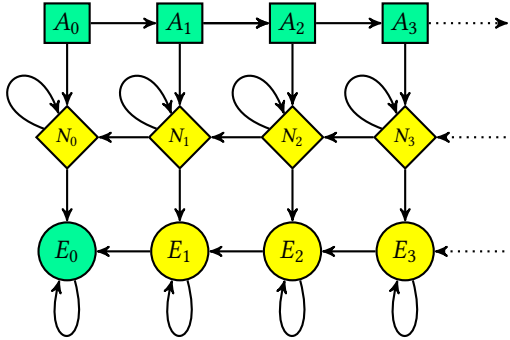

A \defin{strategy} for \Eloise is a function $\phi_\Ei:V^*\VE\rightarrow V$ assigning, to every partial play
ending in some vertex $v\in \VE$, a vertex $v'\in E(v)$. Strategies for \Abelard are defined likewise, \ie as functions $\phi_\Ai:V^*\VA\rightarrow V$. A strategy $\phi$ is \defin{positional} if for any two partial
plays $\pi$ and $\pi'$ ending in the same vertex, we have $\phi(\pi)=\phi(\pi')$. When it is clear from the context that a strategy of \Eloise (\resp \Abelard) is positional we will define it as a function from $V_\Ei$ (\resp $V_\Ai$) to $V$. A strategy for \Eloise $\strat_\Ei$ is a \defin{finite-memory} strategy if it can be implemented by a finite-memory machine that sequentially reads the vertices visited during the play; formally we require that there is a finite set $M$, an element $m_0\in M$, a function $\mathrm{Up}: M\times V\rightarrow M$ and a function $\mathrm{Move}: M\times V\rightarrow V$ such that $\strat_\Ei(v_0\cdots v_i) = \mathrm{Move}(m_i,v_i)$ for every partial play $v_0\cdots v_i$ ending in a vertex $v_i\in V_\Ei$ where we inductively define $m_i$ by letting $m_{k+1} = \mathrm{Up}(m_k,v_{k+1})$ for every $k\geq 0$. Finite-memory strategies for \Abelard are defined likewise. The size of the memory used by such a strategy is defined as the size of $M$. Note that positional strategies correspond to the special case of finite-memory strategy where $M=\{m_0\}$ is a singleton.

In a given play $\lambda=v_0v_1\cdots$ we say that \Eloise (\resp \Abelard) \emph{respects a strategy} $\phi$ if whenever $v_i\in \VE$ (\resp $v_i\in \VA$) one has $v_{i+1} = \phi(v_0\cdots v_i)$.

With an initial vertex $v_0$ and a pair of strategies $(\phi_\Ei,\phi_\Ai)$, we associate the set $\outcomes{v_0}{\phi_\Ei}{\phi_\Ai}$ of possible plays where each player respects his strategy, \ie $\lambda\in \outcomes{v_0}{\phi_\Ei}{\phi_\Ai}$ if and only if $\lambda$ is a play starting from $v_0$ where \Eloise respects $\phi_\Ei$ and \Abelard respects $\phi_\Ai$. 
In the classical setting where \Nature is not involved (\emph{i.e.,} $\VN=\emptyset$), when the strategies of \Eloise and \Abelard are fixed there is only one possible play, \ie $\outcomes{v_0}{\phi_\Ei}{\phi_\Ai}$ is a singleton. The presence of \Nature induces a branching structure: indeed, $\outcomes{v_0}{\phi_\Ei}{\phi_\Ai}$ is the set of branches of the $V$-tree $\treeGame{v_0}{\phi_\Ei}{\phi_\Ai}$ consisting of those partial plays where each  player respects his strategy\footnote{We make here a slight abuse: indeed, as we take as the root the trivial partial play $v_0$ and not $\epsilon$, it breaks the definition of a tree as being a prefix closed set.}.

\begin{example}\label{example:game-running2}
	Consider again the arena from Example~\ref{example:game-running1} (depicted in Figure~\ref{fig:example-game}) and define the following positional strategies $\phi_\Ei$ and $\phi_\Ai$ for \Eloise and \Abelard.
	
	\begin{itemize}
		\item For every $i\geq 0$, $\phi_\Ei(E_i) = E_{i-1}$ if $i>0$ and $\phi_\Ei(E_0)=E_0$.
		\item For every $i\geq 0$, $\phi_\Ai(A_i) = A_{i+1}$ if $i<2$ and $\phi_\Ai(A_i)=N_i$ if $i\geq 2$.
	\end{itemize}
	
	Then the set $\outcomes{A_0}{\phi_\Ei}{\phi_\Ai}$ consists of the plays described by the following $\omega$-regular expression
	
	$$ A_0A_1A_2N_2(N_2^\omega + (N_2^+(E_2E_1E_0^\omega+N_1^\omega+N_1^+(E_1E_0^\omega+(N_0^\omega+N_0^+E_0^\omega)))))$$
	\ie a play in $\outcomes{A_0}{\phi_\Ei}{\phi_\Ai}$ starts by moving the token to $N_2$ and can either get trap forever in some $N_i$ with $i\geq 2$ or eventually reaches some $E_i$ with $i\geq 2$ from where it goes to $E_0$ and stays there forever.
	
\end{example}

A \defin{winning condition} is a subset $\WC\subseteq V^\omega$ and a \defin{game} is a tuple $\game=(\arena,\WC,v_0)$ consisting of an arena, a winning condition and an initial vertex $v_0$.
In this paper, we only consider winning conditions that are \defin{Borel} sets, \ie that belong to the $\sigma$-algebra defined from the basic open sets of the form $KV^\omega$ with $K\subseteq V^*$.

A well known popular example of Borel winning conditions are the \defin{parity} conditions. Let $\col:V\rightarrow \colors$ be a colouring function assigning to every vertex a colour in a \emph{finite} set $C\subset \mathbb{N}$.  Then one defines $\WC_\col$ to be the set of all plays where the smallest infinitely often repeated colour is even, \ie 
$$\WC_\col=\{v_0v_1v_2\cdots\in V^\omega\mid \liminf(\col(v_i))_i\text{ is even}\} $$

\defin{Büchi} (\resp \defin{co-Büchi}) conditions are those parity conditions where $\colors=\{0,1\}$ (\resp $\colors=\{1,2\}$); it requires for a play to be winning to go infinitely (\resp only finitely) often through vertices coloured by $0$ (\resp $1$) and in general it is defined by a set of final (\resp forbidden) vertices: those of colour $0$  (\resp $1$).

A more general class of winning conditions are so-called $\omega$-regular conditions. Such a condition $\WC_{\tau,L}$ is defined thanks to a mapping $\tau:V\rightarrow A$ where $A$ is a \emph{finite} alphabet, and an $\omega$-regular language $L$ over the alphabet $A$ (see \eg \cite{PP04} for definitions of $\omega$-regular languages). Then one simply lets:
$$\WC_{\tau,L}=\{v_0v_1v_2\cdots\in V^\omega\mid 
\tau(v_0)\tau(v_1)\tau(v_2)\cdots\in L\} $$

A play $\lambda$ from $v_0$ is \defin{won} by \Eloise if and only if $\lambda\in \WC$; otherwise $\lambda$ is won by \Abelard. 

\begin{example}\label{example:game-running3}
Consider again the arena from Example~\ref{example:game-running1} and \ref{example:game-running2} depicted in Figure~\ref{fig:example-game}. Consider	 the Büchi condition defined by letting the final vertices (depicted in green in the picture) be the set $\{A_i\mid i\geq 0\}\cup\{E_0\}$. Then the play $\lambda=A_0A_1A_2N_2N_2N_2N_1E_1E_1E_1E_0E_0E_0\cdots$ is won by \Eloise.
\end{example}

A strategy $\phi_\Ei$ is (surely) \defin{winning} for \Eloise in $\game$ if for any strategy $\phi_\Ai$ of \Abelard one has $\outcomes{v_0}{\phi_\Ei}{\phi_\Ai}\subseteq \Omega$, \ie she wins regardless of the choices of \Abelard and \Nature. Symmetrically, a strategy $\phi_\Ai$ is (surely) \defin{winning} for \Abelard in $\game$ if for any strategy $\phi_\Ei$ of \Eloise one has $\outcomes{v_0}{\phi_\Ei}{\phi_\Ai}\cap \Omega=\emptyset$.

{As the winning condition is Borel,} it is a well known result ---~Martin's determinacy Theorem \cite{Martin75}~--- that whenever $\VN=\emptyset$ the game is \emph{determined}, \ie  either \Eloise or \Abelard has a winning strategy. Due to \Nature, it is easily seen that in many situations neither \Eloise nor \Abelard has a winning strategy.
 For instance, consider the B\"uchi game depicted in Figure~\ref{fig:nondeterminiedGame} where all vertices belong to Nature and where the final vertex is $1$. The strategies for \Eloise and \Abelard are both the trivial function with empty domain $\phi_\emptyset$  and the set $\outcomes{1}{\phi_\emptyset}{\phi_\emptyset} = 1\cdot \{1,2\}^\omega$ contains plays that are winning for \Eloise (e.g. $1^\omega$) as well as plays that are winning for \Abelard (e.g. $1 2^\omega$).

\begin{figure}[htb]
\begin{center}
\begin{tikzpicture}[>=stealth',thick,scale=1,transform shape]
\tikzstyle{Abelard}=[draw]
\tikzstyle{Nature}=[draw,diamond]
\tikzstyle{Buchi}=[fill=MediumSpringGreen]
\tikzset{every loop/.style={min distance=10mm,looseness=10}}
\tikzstyle{loopleft}=[in=150,out=210]
\tikzstyle{loopright}=[in=-30,out=30]
\node[Nature,Buchi] (N1) at (0,0) {$1$};
\node[Nature] (N2) at (2,0) {$2$};

\path[->] (N1) edge  [loopleft, loop] ();
\path[->] (N2) edge  [loopright, loop] ();
\path[->,bend right] (N1) edge (N2);
\path[->,bend right] (N2) edge (N1);
\end{tikzpicture}
\end{center}
\caption{A non-determined Büchi game where Nature plays alone.}\label{fig:nondeterminiedGame}
\end{figure}
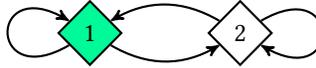

Another example of this situation is given by the Büchi game from Example~\ref{example:game-running3} where it is easily observed that neither \Eloise nor \Abelard has a winning strategy.

 One way of solving this situation, \ie to still evaluate how good a strategy/game is for \Eloise, is to equip \Nature with a probabilistic semantics, leading to the well-known concept of \emph{stochastic games} that we briefly recall in the next section, the main focus of the present paper being to propose alternative semantics (the cardinality one and the topological one) that lead to decidable problems where the previous probabilistic approach fails.

\subsection{The Probabilistic Setting}
\label{subsection:proba}

We now briefly recall the concept of \defin{stochastic games}~\cite{Shapley53,Condon92} (see also \cite{ChatterjeePHD} for an  overview of the field and formal details on the objects below) which consists of equipping the games with Nature with a probabilistic semantics. 
In a nutshell, any vertex in $\VN$ comes with a probability distribution over its neighbours and then, for a fixed tuple $(v_0,\phi_\Ei,\phi_\Ai)$, these probabilities are used to define a $\sigma$-algebra (taking as cones the sets of plays sharing a common finite prefix) and a probability measure $\mu_{v_0}^{\phi_\Ei,\phi_\Ai}$ on $\outcomes{v_0}{\phi_\Ei}{\phi_\Ai}$. In particular, this permits to associate with any pair $(\phi_\Ei,\phi_\Ai)$ a real in $[0,1]$ defined as the probability of the (mesurable) subset $\outcomes{v_0}{\phi_\Ei}{\phi_\Ai}\cap\Omega$ in the previous space. Of special interest is the \defin{value} of a given strategy $\phi_\Ei$ of \Eloise, that estimates how good $\phi_\Ei$ is for her:
$$\Value{\phi_\Ei}{\game}=\inf \{\mu_{v_0}^{\phi_\Ei,\phi_\Ai}(\WC) \mid \phi_\Ai\text{ \Abelard strategy}\}$$
Finally, the value of the game is defined by taking the supremum of the values of \Eloise's strategies:
$$\GValue{\game}=\sup \{\Value{\phi_\Ei}{\game} \mid \phi_\Ei\text{ \Eloise strategy}\}$$
A strategy $\phi_\Ei$ is optimal when $\Value{\phi_\Ei}{\game}= \GValue{\game}$ and it is almost surely winning when \hbox{$\Value{\phi_\Ei}{\game} = 1$}.
 
 A deep result due to Martin~\cite{Martin98} establishes the following determinacy result ($\phi_\Ei$ and $\phi_\Ai$ range over strategies of \Eloise and \Abelard respectively):

$$ \sup_{\phi_\Ei} \inf_{\phi_\Ai} \{\mu_{v_0}^{\phi_\Ei,\phi_\Ai}(\WC) \}  = 1- \inf_{\phi_\Ai}  \sup_{\phi_\Ei}  \{\mu_{v_0}^{\phi_\Ei,\phi_\Ai}(\WC) \}$$

\begin{example}\label{example:game-stoc}
Consider the arena depicted in Figure~\ref{fig:example-game-stochastic}. It is essentially a variant of the game from examples~\ref{example:game-running1}--\ref{example:game-running3} where we replaced the self-loop on $E_0$ by an edge from $E_0$ to $A_0$ and where we associate, with every node $N_i$ owned by \Nature, the following probability distribution $d_i$ on its neighbours:
\begin{itemize}
	\item $d_0(E_0) = d_0(N_0) = 1/2$;
	\item if $i>0$, $d_i(E_i)=d_i(N_i)=d_i(N_{i-1})=1/3$.
\end{itemize}
Then it is easily seen that the positional strategy $\phi_\Ei$ for \Eloise defined by letting $\phi_\Ei(E_i)=E_{i-1}$ if $i>0$ and $\phi_\Ei(E_0)=A_0$ is almost surely winning. Indeed, for a fixed strategy $\phi_\Ai$ the plays $\lambda$ in $\outcomes{A_0}{\phi_\Ei}{\phi_\Ai}$ that are losing for \Eloise are included in the countable union (over all integer $i$ and all integer $k$) of the finite sets of plays that get trap in $N_i$ forever after entering in it after having previously visited exactly $k$-times vertex $E_0$: hence, this set has measure $0$.
\end{example}

\begin{figure}[htb]
\begin{center}
\begin{tikzpicture}[>=stealth',thick,scale=1,transform shape]
\tikzstyle{Abelard}=[draw]
\tikzstyle{Eloise}=[draw,circle]
\tikzstyle{Nature}=[draw,diamond,scale = .65,font=\Large]
\tikzstyle{AbelardH}=[]
\tikzstyle{EloiseH}=[circle]
\tikzstyle{NatureH}=[diamond,scale = .65,font=\Large]
\tikzstyle{C0}=[fill=MediumSpringGreen]
\tikzstyle{C1}=[fill=Yellow]
\tikzstyle{C2}=[fill=Cyan]
\tikzset{every loop/.style={min distance=10mm,looseness=10}}
\tikzstyle{loopleft}=[in=150,out=210]
\tikzstyle{loopright}=[in=-30,out=30]
\node[Abelard,C0] (A1) at (1,0.5) {$A_0$};
\node[Abelard,C0] (A2) at (3,0.5) {$A_1$};
\node[Abelard,C0] (A3) at (5,0.5) {$A_2$};
\node[Abelard,C0] (A4) at (7,0.5) {$A_3$};
\node[AbelardH] (A5) at (9,0.5) {};
\node[Nature,C1] (N1) at (1,-1.5) {$N_0$};
\node[Nature,C1] (N2) at (3,-1.5) {$N_1$};
\node[Nature,C1] (N3) at (5,-1.5) {$N_2$};
\node[Nature,C1] (N4) at (7,-1.5) {$N_3$};
\node[NatureH] (N5) at (9,-1.5) {};
\node[Eloise,C0] (E1) at (1,-3.5) {$E_0$};
\node[Eloise,C1] (E2) at (3,-3.5) {$E_1$};
\node[Eloise,C1] (E3) at (5,-3.5) {$E_2$};
\node[Eloise,C1] (E4) at (7,-3.5) {$E_3$};
\node[EloiseH] (E5) at (9,-3.5) {};

\path[->] (A1) edge (N1);\path[->] (A1) edge (A2);
\path[->] (A2) edge (N2);\path[->] (A2) edge (A3);
\path[->] (A3) edge (N3);\path[->] (A2) edge (A3);
\path[->] (A4) edge (N4);\path[->] (A3) edge (A4);
\path[->,dotted] (A4) edge (A5);

\path[->] (N1) edge  [in=110,out=160, loop] node [left,near start] {\scriptsize $1/2$} ();\path[->] (N1) edge node [right,near start] {\scriptsize $1/2$}(E1);
\path[->] (N2) edge  [in=110,out=160, loop]node [above left,near end] {\scriptsize $1/3$} (); ();\path[->] (N2) edge node [below,near start] {\scriptsize $1/3$}(N1);\path[->] (N2) edge node [right,near start] {\scriptsize $1/3$}(E2);
\path[->] (N3) edge  [in=110,out=160, loop]node [above left,near end] {\scriptsize $1/3$} (); ();\path[->] (N3) edge  node [below,near start] {\scriptsize $1/3$} (N2);\path[->] (N3) edge node [right,near start] {\scriptsize $1/3$}(E3);
\path[->] (N4) edge  [in=110,out=160, loop]node [above left,near end] {\scriptsize $1/3$} (); ();\path[->] (N4) edge node [below,near start] {\scriptsize $1/3$} (N3);\path[->] (N4) edge node [right,near start] {\scriptsize $1/3$}(E4);
\path[->,dotted] (N5) edge (N4);

\path[->,in=180,out=180,looseness=1.5] (E1) edge (A1);
\path[->] (E2) edge  [loop below, loop] ();\path[->] (E2) edge (E1);
\path[->] (E3) edge  [loop below, loop] ();\path[->] (E3) edge (E2);
\path[->] (E4) edge  [loop below, loop] ();\path[->] (E4) edge (E3);
\path[->,dotted] (E5) edge (E4);

\end{tikzpicture}
\end{center}
\caption{Example of a stochastic game}\label{fig:example-game-stochastic}
\end{figure}
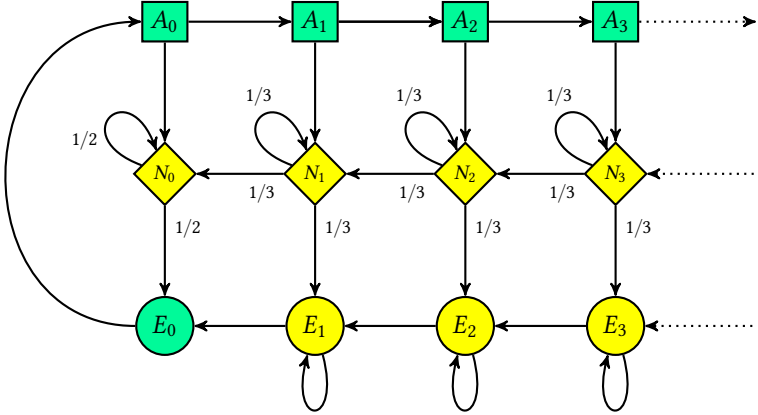

\subsection{The Cardinality Setting}\label{subsection:cardinality}

We now propose a change of perspective based on counting: in order to evaluate how good a situation is for \Eloise  (\ie using some strategy $\phi_\Ei$ against a strategy $\phi_\Ai$ of \Abelard)  we simply count how many losing plays there are; the fewer they are the better the situation is. 

As a preliminary illustration we revisit Example~\ref{example:game-stoc} .

\begin{example}\label{example:game-count}
 Consider again the stochastic game from Example~\ref{example:game-stoc} (depicted in Figure~\ref{fig:example-game-stochastic}) together with the positional strategy $\phi_\Ei$ of \Eloise. Now forget about the stochastic view of \Nature and think of it as being simply non-deterministic. Let $\phi_\Ai$ be an arbitrary strategy of \Abelard: as remarked in Example~\ref{example:game-stoc} the set of losing plays in $\outcomes{A_0}{\phi_\Ei}{\phi_\Ai}$ is countable. On the other hand the set $\outcomes{A_0}{\phi_\Ei}{\phi_\Ai}$ is easily seen to be uncountable. Therefore, one can consider that the set of losing play for \Eloise when using the strategy $\phi_\Ei$ is somehow negligible with respect to the set of all plays.
\end{example}

We first note the following proposition~\cite{Alexandrov1916} that characterises the cardinals of the Borel subsets of an arbitrary set $\outcomes{v_0}{\phi_\Ei}{\phi_\Ai}$.

{
\begin{proposition}
\label{prop:continuumHypothesis}
For any arena, any initial vertex, any pair of strategies $(\phi_\Ei,\phi_\Ai)$ and any Borel subset $S\subseteq \outcomes{v_0}{\phi_\Ei}{\phi_\Ai}$, one has $\cardinal{S}\in\mathbb{N}\cup\{\aleph_0,2^{\aleph_0}\}$.
\end{proposition}
}

\begin{proof}
As  $\outcomes{v_0}{\phi_\Ei}{\phi_\Ai}$ is the set of branches of a tree whose set of directions is {countable} (see for instance Theorem 3.11 in \cite{Kechris}), it is a Polish space with the standard basis
$ \{ \mathrm{Cone}_{\treeGame{v_0}{\phi_\Ei}{\phi_\Ai}}(v) \mid v \in \treeGame{v_0}{\phi_\Ei}{\phi_\Ai} \}.$
By  \cite[Theorem 13.6]{Kechris}, any Borel subset $S$  is either countable or has cardinality $2^{\aleph_{0}}$. 
\end{proof}

We define the cardinality leaking of an \Eloise's strategy as a measure of its quality.

\begin{definition}[Cardinality Leaking of a Strategy]\label{def:cardleaking}
Let $\game=(\arena,\WC,v_0)$ be a game and let $\phi_\Ei$ be a strategy of \Eloise. The \defin{cardinality leaking} of $\phi_\Ei$ is the cardinal $\CL{\phi_\Ei}$ defined by
$$\CL{\phi_\Ei} = \sup\{\cardinal{\outcomes{v_0}{\phi_\Ei}{\phi_\Ai}\setminus\WC}\mid \phi_{\Ai}\text{ strategy of \Abelard}\}$$

Proposition~\ref{prop:continuumHypothesis} implies that $\CL{\phi_\Ei}\in\mathbb{N}\cup\{\aleph_0,2^{\aleph_0}\}$.
\end{definition}

The goal of \Eloise is to minimise the number of losing plays, hence leading the following concept.

\begin{definition}[Leaking Value of a Game]\label{def:cardleakGame}
Let $\game=(\arena,\WC,v_0)$ be a game. The \defin{leaking value} of $\game$  is the cardinal $\LVal{\game}$ defined by 
$$\LVal{\game} = \inf\{\CL{\phi_{\Ei}}\mid \phi_{\Ei}\text{ strategy of \Eloise}\}$$
Thanks to Proposition~\ref{prop:continuumHypothesis} it follows that $\LVal{\game}\in\mathbb{N}\cup\{\aleph_0,2^{\aleph_0}\}$.
\end{definition}

In the reminder of this article, we consider that a strategy is good from the cardinality point of view if its cardinality leaking is countable.
From a modelisation point of view, we agree that this notion can be questionnable. In particular it only makes sense if for all strategy $\phi_{\Ei}$ and $\phi_{\Ai}$ of \Eloise and \Abelard respectively,
the set of outcomes is uncountable. A sufficient condition to ensure this last 
property is that all vertices of \Nature have at least two successors and 
that every play visits infinitely many vertices of \Nature.
A stronger requirement that we will also consider is to look for strategy with a fixed finite cardinality leaking.

\begin{remark}\label{remark:supNotMaxCardleaking}
One can wonder whether the $\sup$ in the definition of $\CL{\phi_\Ei}$ can be replaced by a $\max$, \ie whether, against any fixed strategy of \Eloise, \Abelard has always an “optimal” counter strategy.

Actually this is not possible as exemplified by the Büchi game depicted in Figure~\ref{fig:supNotMaxCardleaking} where  coloured vertices ($v_A$ and $v_W$) are the final ones ~— with $v_{A}$ as initial vertex. 
\begin{figure}
\begin{center}
\begin{tikzpicture}[>=stealth',thick,scale=1,transform shape]
\tikzstyle{Abelard}=[draw]
\tikzstyle{Eloise}=[draw,circle]
\tikzstyle{Nature}=[draw,diamond,scale = .65,font=\Large]
\tikzstyle{Buchi}=[fill=MediumSpringGreen]
\tikzset{every loop/.style={min distance=15mm,looseness=10}}
\tikzstyle{loopleft}=[in=150,out=210]
\tikzstyle{loopright}=[in=-30,out=30]
\tikzstyle{loopbelow}=[in=-120,out=-60]
\tikzstyle{loopabove}=[in=120,out=60]
\node[Abelard,Buchi,minimum height=.7cm] (vA) at (0,-0.7) {$v_{A}$};
\node[Nature] (vN) at (1.5,0) {$v_{N}$};
\node[Eloise] (vL) at (0,.5) {$v_{L}$};
\node[Eloise] (vE) at (3,0) {$v_{E}$};
\node[Eloise,Buchi] (vW) at (4.5,0) {$v_{W}$};

\path[->] (vA) edge  [loopleft, loop] ();
\path[->] (vW) edge  [loopright, loop] ();
\path[->] (vL) edge  [loopleft, loop] ();
\path[->] (vA) edge (vN);
\path[->] (vN) edge (vL);
\path[->,bend right] (vN) edge (vE);
\path[->,bend right] (vE) edge (vN);
\path[->] (vE) edge (vW);
\end{tikzpicture}
\end{center}
\caption{Arena of Remark~\ref{remark:supNotMaxCardleaking}.}\label{fig:supNotMaxCardleaking}
\end{figure}
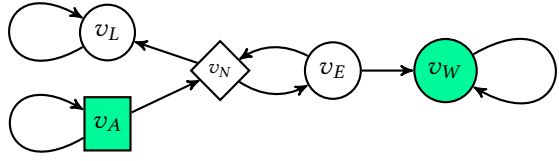

Consider the strategy $\phi_{\Ei}$ of \Eloise consisting in a partial play ending in $v_{E}$ (in $v_{L}$ and $v_{W}$ \Eloise has a single choice) to go to $v_{N}$ if there are less occurrences of $v_{E}$ than of $v_{A}$ in the partial play and to go to $v_{W}$ otherwise. Clearly for any strategy $\phi_{\Ai}$ of \Abelard,  {$\cardinal{\outcomes{v_0}{\phi_\Ei}{\phi_\Ai}\setminus\WC}$ is finite. However for any $k \geq 0$, \Abelard can ensure that there are  $k$ plays lost by \Eloise by looping $(k-1)$ times on the vertex $v_{A}$ before going to $v_{N}$.}
\end{remark}

{As cardinals are well-ordered,  \Eloise always has an “optimal” strategy for the leaking value criterion (\ie we can replace the $\inf$ by a $\min$ in Definition~\ref{def:cardleakGame}).}
\begin{proposition}
Let $\game=(\arena,\WC,v_0)$ be a game. There is a strategy $\phi_\Ei$ of \Eloise such that $\LVal{\game} = \CL{\phi_{\Ei}}$.
\end{proposition}

{
\begin{remark}
It is natural to wonder if, for a strategy $\phi_\Ei$ of \Eloise such that $\CL{\phi_{\Ei}}=\LVal{\game}$, there exists a strategy $\phi_\Ai$ of \Abelard which reaches $\LVal{\game}$, \ie such that $\cardinal{\outcomes{v_0}{\phi_\Ei}{\phi_\Ai}\setminus\WC}=\LVal{\game}$. By considering for instance the reachability game $\game$ depicted in Figure~\ref{fig:counter-example-abelard-optimal}, we will see that such a strategy for \Abelard may not necessarily exists. This game is played between \Abelard and \Nature and hence the empty strategy $\phi_{\Ei}$ for \Eloise is optimal, \ie $ \CL{\phi_{\Ei}}=\LVal{\game}$. In this game, a strategy for \Abelard is entirely characterised by the first index $n$ (if it exists) such that \Abelard moves from $A_n$ to the game $\game_n$ (that is such that exactly $n$ plays are losing for \Eloise in it). Hence for all $n \geq 1$, we denote by $\phi^n_{\Ai}$ the strategy of \Abelard consisting in moving from $A_i$ to $A_{i+1}$ for all $i<n$ and going to $\game_n$ on $A_n$ and by $\phi^\infty_{\Ai}$ the strategy in which \Abelard always moves from $A_i$ to $A_{i+1}$. As there are 
exactly $n$ losing plays for \Eloise in $\game_n$, we have $\cardinal{\outcomes{v_0}{\phi_\Ei}{\phi_\Ai^n}\setminus\WC}=n$ and $\cardinal{\outcomes{v_0}{\phi_\Ei}{\phi_\Ai^\infty}\setminus\WC}=1$. It follows that the $\LVal{\game}= \aleph_0$ which cannot be reached by any strategy of \Abelard.
\end{remark}}

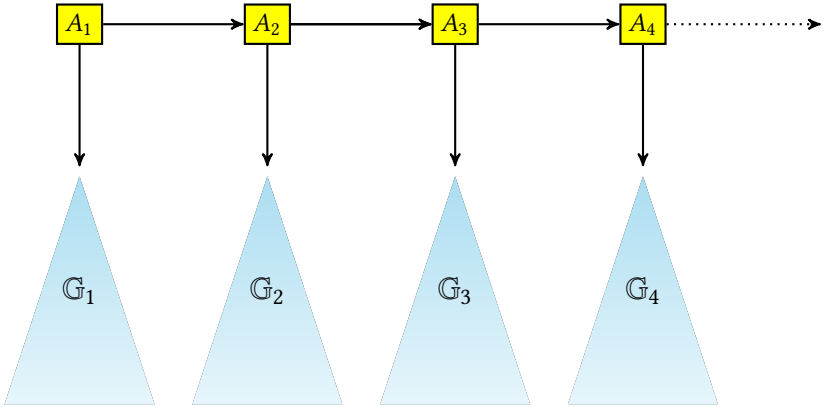
\begin{figure}[htb]
\begin{center}
\begin{tikzpicture}[>=stealth',thick,scale=1,transform shape]
\tikzstyle{Abelard}=[draw]
\tikzstyle{Eloise}=[draw,circle]
\tikzstyle{Nature}=[draw,diamond,scale = .65,font=\Large]
\tikzstyle{AbelardH}=[]
\tikzstyle{EloiseH}=[circle]
\tikzstyle{NatureH}=[diamond,scale = .65,font=\Large]
\tikzstyle{C0}=[fill=MediumSpringGreen]
\tikzstyle{C1}=[fill=Yellow]
\tikzstyle{C2}=[fill=Cyan]
\tikzset{every loop/.style={min distance=10mm,looseness=10}}
\tikzstyle{loopleft}=[in=150,out=210]
\tikzstyle{loopright}=[in=-30,out=30]
\node[Abelard,C1] (A1) at (1,0.5) {$A_1$};
\node[Abelard,C1] (A2) at (3.5,0.5) {$A_2$};
\node[Abelard,C1] (A3) at (6,0.5) {$A_3$};
\node[Abelard,C1] (A4) at (8.5,0.5) {$A_4$};
\node[AbelardH] (A5) at (11,0.5) {};

\node[] (N1) at (1,-1.5) {};
\node[] (N2) at (3.5,-1.5) {};
\node[] (N3) at (6,-1.5) {};
\node[] (N4) at (8.5,-1.5) {};

\path[->] (A1) edge (N1);\path[->] (A1) edge (A2);
\path[->] (A2) edge (N2);\path[->] (A2) edge (A3);
\path[->] (A3) edge (N3);\path[->] (A2) edge (A3);
\path[->] (A4) edge (N4);\path[->] (A3) edge (A4);
\path[->,dotted] (A4) edge (A5);

\fill[bottom color=SkyBlue!100!black!20,
top color=SkyBlue!100!black!70] (1,-1.5) --  (0,-4.5) -- (2,-4.5) -- cycle;
\node[font=\Large] (G1) at (1,-3) {$\game_1$};
\fill[bottom color=SkyBlue!100!black!20,
top color=SkyBlue!100!black!70] (3.5,-1.5) --  (2.5,-4.5) -- (4.5,-4.5) -- cycle;
\node[font=\Large] (G1) at (3.5,-3) {$\game_2$};
\fill[bottom color=SkyBlue!100!black!20,
top color=SkyBlue!100!black!70] (6,-1.5) --  (5,-4.5) -- (7,-4.5) -- cycle;
\node[font=\Large] (G1) at (6,-3) {$\game_3$};
\fill[bottom color=SkyBlue!100!black!20,
top color=SkyBlue!100!black!70] (8.5,-1.5) --  (7.5,-4.5) -- (9.5,-4.5) -- cycle;
\node[font=\Large] (G1) at (8.5,-3) {$\game_4$};
\end{tikzpicture}
\end{center}
\caption{Example of a reachability game played between \Abelard and \Nature. For all $n \geq 0$, the game $\game_n$ is only played by \Nature and is such that exactly $n$ plays are losing for \Eloise.}\label{fig:counter-example-abelard-optimal}
\end{figure}

\subsection{The Topological Setting}\label{subsection:topological}

{A notion of topological “bigness” and “smallness” is given by  \emph{large} and \emph{meager} sets respectively (see \cite{Graedel08,VolzerV12} for a survey of the notion).  From the modelisation point of view, the intuition is that meager sets (the complements of  large sets) are somehow negligible. In \cite{VolzerV12}, the authors give weight to this idea by showing that, for regular trees (\ie those trees obtained by unfolding finite graphs), the set of branches satisfying an $\omega$-regular condition is large if and only if it has probability 1 (in the sense of Section~\ref{subsection:proba}). However they also show that in general, even for the Büchi condition and when the tree is the unfolding of a pushdown graph, this is no longer true (see \cite[p. 27]{VolzerV12}).}

Let $t$ be a $D$-tree for some set $D$ of directions. Then its set of branches can be seen as a topological space by taking as basic open sets the {set} of cones. A set of branches $B\subseteq D^\omega$ is \defin{nowhere dense} if for all node $u\in t$, there exists  another node $v\in t$ such that $u\prefix v$ and such that $v$ does not belong to any branch in $B$.
A set of branches is \defin{meager} if it is the countable union of nowhere dense sets. Finally it is \defin{large} if it is the complement of a meager set.

{
A natural topological criterion to consider that a strategy  $\phi_\Ei$ for \Eloise is good against a strategy $\phi_\Ai$ of \Abelard is that the set of plays lost by \Eloise is meager in the tree $\treeGame{v_0}{\phi_\Ei}{\phi_\Ai}$.}

\begin{definition}[Topologically-Good Strategies]\label{def:topoGood}
Let $\game=(\arena,\WC,v_0)$ be a game and let $\phi_\Ei$ be a strategy of \Eloise. We say that $\phi_\Ei$ is \defin{topologically-good} if and only if for any strategy $\phi_\Ai$ of \Abelard the set  $\outcomes{v_0}{\phi_\Ei}{\phi_\Ai}\setminus \WC$ of losing plays for \Eloise is meager in the tree $\treeGame{v_0}{\phi_\Ei}{\phi_\Ai}$; or equivalently
 the set  $\outcomes{v_0}{\phi_\Ei}{\phi_\Ai}\cap\WC$ of plays won by \Eloise is large.
\end{definition}

\begin{example}
\label{example:game-topo}
 Consider again the stochastic game from Example~\ref{example:game-stoc} (depicted in Figure~\ref{fig:example-game-stochastic}) together with the positional strategy $\phi_\Ei$ of \Eloise (also presented in Example~\ref{example:game-stoc}). Now forget about the stochastic view of \Nature and think of it as being simply non-deterministic. We claim that the strategy $\phi_\Ei$ is topologically-good. 
 
 Indeed, let $\phi_\Ai$ be an arbitrary strategy of \Abelard: as already remarked in Example~\ref{example:game-stoc} the set of plays in $\outcomes{A_0}{\phi_\Ei}{\phi_\Ai}\setminus\WC$ is included in the countable union over all $i,k\geq 0$ of the plays $\Lambda_{i,k}$ where a play  belongs to $\Lambda_{i,k}$ if it gets trap forever in node $N_i$ from round $k$, \ie $\Lambda_{i,k}=V^{k-1}N_i^\omega$. Now we remark that $\outcomes{A_0}{\phi_\Ei}{\phi_\Ai}\cap \Lambda_{i,k}$ is nowhere dense in the tree $\treeGame{v_0}{\phi_\Ei}{\phi_\Ai}$, as any partial play $\lambda$ can be extended to another partial play $\lambda'$ so that any extension of the latter as an infinite play is not trap forever in $N_i$ from round $k$ (it can be trapped forever, but later), \ie for any node $u_\lambda$ in $\treeGame{v_0}{\phi_\Ei}{\phi_\Ai}$ there exists another node $u_{\lambda'}$ such that $u_{\lambda'}\prefix u_\lambda$ and $u_{\lambda'}$ does not belong to any branch in $\Lambda_{i,k}$. Hence, it means that the set $\outcomes{A_0}{\phi_\Ei}{\phi_\Ai}\setminus\WC$ is a countable union of nowhere dense sets, equivalently it is meager. Therefore, we conclude that $\phi_\Ei$ is topologically-good. 
\end{example}

We now recall a useful notion, Banach-Mazur games, to reason on meager sets. Banach-Mazur theorem gives a game characterisation of large and meager sets of branches (see for instance \cite{Oxtoby71,Kechris,Graedel08}). The \defin{Banach-Mazur game} on a tree $t$,  is a two-player game where \Abelard and \Eloise choose alternatively a node in the tree, forming a branch: \Abelard chooses first a node and then \Eloise chooses a descendant of the previous node and \Abelard chooses a descendant of the previous node and so on forever. In this game it is \emph{always \Abelard that starts} a play. 

Formally a \emph{play} is an infinite sequence $u_1,u_2,\ldots$ of words in $D^+$ such that for all $i$ one has $u_1 u_2\cdots u_i\in t$, and the branch associated with this play is $u_1u_2\cdots$. A \emph{strategy} for \Eloise is a mapping $\phi : (D^+)^+ \rightarrow D^+$ that takes as input a finite sequence of words, and outputs a word. A play $u_1,u_2,\ldots$ \emph{respects} $\phi$ if for all $i\geq 1$, $u_{2i}= \phi(u_1,\ldots,u_{2i-1})$. We define $Outcomes(\phi)$ as the set of plays that respect $\phi$ and $\mathcal{B}(\phi)$ as the set of branches associated with the plays in $Outcomes(\phi)$. 

The Banach-Mazur theorem {(see \footnote{In \cite{Graedel08} the players of the Banach-Mazur game are called $0$ and $1$ and Player $0$ corresponds to \Abelard while player $1$ corresponds to \Eloise. Hence, when using a statement from \cite{Graedel08} for our setting one has to keep this in mind as well as the fact that one must replace the winning condition by its complement (hence, replacing “meager” by “large”).}
 \eg \cite[Theorem~4]{Graedel08})} states that a set of branches $B$ is large if and only if there exists a strategy $\phi$ for \Eloise such that $\mathcal{B}(\phi)\subseteq B$.  {Hence, if one thinks of $B$ as a winning condition for \Eloise (\ie she wins a play if and only if it belongs to $B$), it means that those sets $B$ for which she has a winning strategy are exactly the large ones.}

Furthermore a folk result {(see \eg \cite[Theorem~9]{Graedel08})} about Banach-Mazur games states that {when $B$ is Borel}\footnote{{This statement holds as soon as the Banach-Mazur games are determined and hence, in particular for Borel sets.}} one can look only at “simple” strategies, defined as follows. A \defin{decomposition-invariant strategy} is a mapping $f:t \rightarrow D^+$ and we associate with $f$ the strategy $\phi_f$ defined by $\phi_f(u_1,\ldots,u_{k})=f(u_1\cdots u_k)$. Finally, we define $Outcomes(f) = Outcomes(\phi_f)$ and $\mathcal{B}(f)=\mathcal{B}(\phi_f)$. The folk result states that for any {Borel} set of branches $B$,  there exists a strategy $\phi$ such that $Outcomes(\phi)\subseteq B$  if and only if there exists a decomposition-invariant strategy $f$ such that $\mathcal{B}(f)\subseteq B$.

\begin{example}
Consider the game in Example~\ref{example:game-topo} (depicted in Figure~\ref{fig:example-game-stochastic}), fix again the same strategy $\phi_\Ei$ for \Eloise and define as a strategy for \Abelard a strategy where when the token is in some vertex in $V_\Ai$ \Abelard	moves it (by successive moves) to $A_k$ where $k$ denotes the number of visits to vertex $A_0$ from the beginning of the play and from $A_k$ moves it down to $N_k$. 

As $\phi_\Ei$ is topologically-winning, it means that the set $B=\outcomes{v_0}{\phi_\Ei}{\phi_\Ai}\setminus\WC$ is large in $\treeGame{v_0}{\phi_\Ei}{\phi_\Ai}$. 
We illustrate the concept of Banach-Mazur game by defining a winning decomposition-invariant strategy $f$ for \Eloise in the Banach-Mazur game on $\treeGame{v_0}{\phi_\Ei}{\phi_\Ai}$. For that it suffices to let $f(\lambda)=\lambda'$ where $\lambda'$ is some arbitrary partial play extending $\lambda$ by a path ending in node $E_0$ (such a path always exists). Then, it is straightforward to verify that $\mathcal{B}(f)\subseteq B$.
\end{example}

\subsection{The Tree-Language Setting}

We now propose a last setting, that in some cases permits to capture the three previously defined ones. This setting only makes sense  when the arena $\arena=(G,\VE,\VA,\VN)$ comes with a mapping $\col: V\rightarrow\colors$ where  $V=\VE\uplus \VA\uplus \VN$ denotes the set of vertices in the arena and $\colors$ is a \emph{finite} set. Fix a subset $\mathcal{L}$ of $\colors$-labeled $V$-trees. 

For a given initial vertex $v_0$ and a pair $(\phi_\Ei,\phi_\Ai)$ of strategies for \Eloise and \Abelard, we can map the $V$-tree $\treeGame{v_0}{\phi_\Ei}{\phi_\Ai}$ to a $\colors$-labelled $V$-tree where each node $v_0v_1\cdots v_k\in \treeGame{v_0}{\phi_\Ei}{\phi_\Ai}$ is labelled by $\col(v_k)$. In the sequel we overload notation $\treeGame{v_0}{\phi_\Ei}{\phi_\Ai}$ to designate this tree.

Now we say that a strategy $\phi_\Ei$ is \defin{$\mathcal{L}$-good} if and only if for every strategy $\phi_\Ai$ of \Abelard the $\colors$-labelled tree $\treeGame{v_0}{\phi_\Ei}{\phi_\Ai}$ belongs to $\mathcal{L}$.

Let $K$ be a subset of $\colors^\omega$ and let $\WC_K$ be the winning condition defined by letting $$\WC_K=\{v_0v_1v_2\cdots \mid \col(v_1)\col(v_2)\col(v_3)\in K\}$$
Then the following trivially holds.

\begin{lemma}\label{lemma:settingsVSLGood}
	Let $\colors$ be a finite set and let $K\subseteq C^\omega$. Let $\arena=(G,\VE,\VA,\VN)$ be an arena with vertices $V$ and let $\col:V\rightarrow\colors$. Let $v_0\in V$ be some initial vertex and let $\game$ be the game $\game=(\arena,v_0,\WC_K)$. Then the following holds.
	\begin{enumerate}
		\item \Eloise almost surely wins $\game$ if and only if she has an $\mathcal{L}_{\mathrm{Stoc}}$-good strategy when playing in arena $\arena$ starting from $v_0$ where $\mathcal{L}_{\mathrm{Stoc}}$ is the set of $\colors$-labelled $V$-trees such that almost all branches\footnote{More formally, the set of branches that are labelled by a sequence in $K$ has measure $1$ for the Lebesgue measure obtained from the Carathéodory extension theorem when defining the measure of a cone, \ie a set of branches sharing a common finite prefix $v_0\cdots v_k$ as the product $\prod_{0\leq i< k\mid v_i\in \VN}d_{v_i}(v_{i+1})$ where $d_{v_i}$ denotes the probability distribution over the neighbours of a vertex $v_i\in \VN$.} are labelled by a sequence in $K$.
		\item The leaking value of $\game$ is countable (\resp smaller than some fixed $k$) if and only if \Eloise has an $\mathcal{L}_{\mathrm{Card}}$-good strategy when playing in arena $\arena$ starting from $v_0$ where $\mathcal{L}_{\mathrm{Card}}$ is the set of $\colors$-labelled $V$-trees such that all branches but countably many (\resp but $k$) are labelled by a sequence in $K$.
		\item \Eloise has a topologically-good strategy in $\game$ if and only she has an $\mathcal{L}_{\mathrm{Topo}}$-good strategy when playing in arena $\arena$ starting from $v_0$ where $\mathcal{L}_{\mathrm{Topo}}$ is the set of $\colors$-labelled $V$-trees such that the subset of branches labelled by a sequence in $K$ is large.
	\end{enumerate}
\end{lemma}

Let $\colors$ be a finite set. Then a $\colors$-labelled tree can be seen as a relational structure (see \cite{libkin_04_elements} for basic concepts on relational structures) whose universe is the set of nodes of the tree and whose relations consist of a unary predicate for every element $c$ in $\colors$ (that holds in every node labelled by $c$) and a binary predicate for the parent/son relation. 

Let $G=(V,E)$ be a graph, $\colors$ be a finite set, $\col: V\rightarrow\colors$ be a mapping, $\arena=(G,\VE,\VA,\VN)$ be an arena and $v_0$ be an initial vertex. We define the \defin{unfolding} of $\arena$ from $v_0$ as the $\colors$-labelled $V$-tree whose set of nodes is the set $$\{v_1\cdots v_k\mid v_0v_1\cdots v_k \text{ is a partial play in arena }\arena\}$$ and where the root $\epsilon$ is labelled by $\col(v_0)$ and any other node $v_1\cdots v_k$ is labelled by $\col(v_k)$. We see it as a relational structure as explained above with three extra unary predicate (on for each player) $p_\Ei$, $p_\Ai$ and $p_\Ni$ such that $p_\Ei$ (\resp $p_\Ai$, \resp $p_\Ni$) holds in a node $u$ if and only if the last vertex of $v_0\cdot u$ belongs to $\VE$ (\resp $\VA$, \resp $\VN$).

One can wonder whether existence of $\mathcal{L}$-good strategies can be decided for special classes of languages $\mathcal{L}$. The most natural one are those definable in monadic second order logic (MSO). As this is the only place in this paper where we make use of logic, we refer the reader to \cite{Thomas97} for formal definitions and classical results regarding MSO logic. The following result was remarked by Pawe{\l} Parys \cite{ParysPrivate2016} and it permits to derive decidability for several classes of arenas (see Corollary~\ref{cor:Parys-parity} below).

\begin{theorem}[\cite{ParysPrivate2016}]\label{theo:parys}
	Let $G=(V,E)$ be a graph, let $\colors$ be a finite set, let $\col: V\rightarrow\colors$ be a mapping, let $\arena=(G,\VE,\VA,\VN)$ be an arena and let $v_0$ be an initial vertex. Let $\mathcal{L}$ be an MSO-definable set of $\colors$-labelled $V$-trees. Then there exists an MSO formula $\Phi_{\mathcal{L}\mathrm{-good}}$ such that $\Phi_{\mathcal{L}\mathrm{-good}}$ holds on the unfolding of $\arena$ from $v_0$ if and only if \Eloise has an $\mathcal{L}$-good strategy from $v_0$.
\end{theorem}

\begin{proof}
	Let $\Phi_{\mathcal{L}}$ be an MSO formula defining $\mathcal{L}$ (\ie it holds in a tree if and only if the tree belongs to $\mathcal{L}$). Call $T_{\mathcal{G}}$ the unfolding of the arena $\mathcal{G}$ from $v_0$. The formula $\Phi_{\mathcal{L}\mathrm{-good}}$ existentially quantifies a strategy for \Eloise, then universally quantifies a strategy of \Abelard and then relativise the formula $\Phi_{\mathcal{L}}$ to the subtree induced by the respective strategies. 
	
More formally, in order to quantify over a strategy —~say for \Eloise~— one quantifies a set of nodes $X$ such that for every node $u$ owned by \Eloise (\ie such that $p_\Ei(u)$ holds) exactly one son of $u$ belongs to $X$ (it corresponds to the image by the strategy of the partial play $v_0\cdot u$ associated with the node $u$): call $\mathrm{Valid_\Ei}(X)$ an MSO-formula checking that a set of nodes $X$ satisfies the previous requirement. Symmetrically one defines a formula $\mathrm{Valid_\Ai}(Y)$ to check that a set $Y$ correctly encodes a strategy of \Abelard. 
	
	Next, we define $\mathrm{Reach}(X,Y,Z)$ as a formula that holds if and only if $Z$ is the set of vertices reachable from the root by following the strategies encoded by $X$ and $Y$. The formula $\mathrm{Reach}(X,Y,Z)$ simply states that $Z$ is the smallest set that contains the root, and such that for every node $u$ in $Z$ if it satisfies $p_\Ei$ (\resp $p_\Ai$) then there is exactly one son of $u$ in $Z$ and it also belongs to $X$ (\resp $Y$), and if it satisfies $p_\Ni$ then all its successors belongs to $Z$. Hence, $Z$ consists of all nodes in $\treeGame{v_0}{\phi_X}{\phi_Y}$ (where $\phi_X$ and $\phi_Y$ are the strategies associated with $X$ and $Y$ respectively).
	
	Now, let $\Phi^{\mathrm{rel}}_{\mathcal{L}}(Z)$ be the formula obtained from $\Phi_{\mathcal{L}}$ by relativising to $Z$, \ie by guarding every quantification to nodes in $Z$, \ie every $\exists x\, \Psi(x)$ in $\Phi_{\mathcal{L}}$ is replaced by $\exists x\, (x\in Z\wedge \Psi(x))$ and every $\exists X\, \Psi(X)$ in $\Phi_{\mathcal{L}}$ is replaced by $\exists X\, (X\subseteq Z\wedge \Psi(X))$.
	
	Then the formula $\Phi_{\mathcal{L}\mathrm{-good}}$ is simply defined as 
	$$
		\Phi_{\mathcal{L}\mathrm{-good}} = \exists X 
			[\mathrm{Valid_\Ei}(X) \wedge \forall Y \forall Z\ (\mathrm{Valid_\Ai}(Y)\wedge \mathrm{Reach}(X,Y,Z)) \Rightarrow \Phi^{\mathrm{rel}}_{\mathcal{L}}(Z)
			]
	$$
\end{proof}

Thanks to Courcelle-Walukiewicz theorem~\cite{CourcelleW98} stating that every MSO-definable property on the unfolding of a structure is an MSO-definable property on the structure itself, we can lift Theorem~\ref{theo:parys}, where we see an arena as a relational structure with predicates for the edge relations, the image of the $\col$ function, three predicates reflecting whether a vertex belongs to $\VE$, $\VA$ or $\VN$ and a last predicate to distinguish the initial vertex $v_0$.

\begin{corollary}\label{cor:Parys}
		Let $G=(V,E)$ be a graph, let $\colors$ be a finite set, let $\col: V\rightarrow\colors$ be a mapping, let $\arena=(G,\VE,\VA,\VN)$ be an arena and let $v_0$ be an initial vertex. Let $\mathcal{L}$ be an MSO-definable set of $\colors$-labelled $V$-trees. Then there exists an MSO formula $\Phi_{\mathcal{L}\mathrm{-good}}$ such that $\Phi_{\mathcal{L}\mathrm{-good}}$ holds on $\arena$ if and only if \Eloise has an $\mathcal{L}$-good strategy from $v_0$.
\end{corollary}

One can now wonder in which cases Corollary~\ref{cor:Parys} can be combined with Lemma~\ref{lemma:settingsVSLGood}. A natural candidate for the criterion $K\subseteq \colors^\omega$ on branches is the parity condition\footnote{{We could consider more general $\omega$-regular conditions but they reduce to a parity condition via product of a game with a finite graph.}} (in fact it is the only reasonable option for $\mathcal{L}$ sets from Lemma~\ref{lemma:settingsVSLGood} to be MSO-definable). We also restrict here to arenas with finite out-degree. 

First note that it is known from \cite[Theorem~21]{CHS14a} that the language of $\{a,b\}$-binary trees such that almost every branch contains a node label by $a$ is not MSO-definable\footnote{{We refer the reader to \cite[Theorem~21]{CHS14a} for a complete proof of this statement but here are the key arguments. Call $\mathcal{L}_a$ the language of $\{a,b\}$-binary trees such that almost every branch contains a node label by $a$. One first argues that for any regular tree $t$, if there is no cone in $t$ whose branches only contain the letter $b$, then this tree $t$ belongs to $\mathcal{L}_a$. Then, one let $\mathcal{L}$ be the set of trees that does not belong to $\mathcal{L}_a$ but does not have a cone whose branches only contain $b$. Hence, $\mathcal{L}$ does not contain a regular tree and, assuming by contradiction that $\mathcal{L}_a$ is MSO-definable, basic properties of MSO-definable languages imply that $\mathcal{L}$ is empty. Finally, one build a (non-regular) tree $t_0$ in $\mathcal{L}$ which leads a contradiction with the MSO-definability of $\mathcal{L}_a$: the rough idea to construct $t_0$ is to pick, in every cone, a node that gets labelled by $a$ and to chose that node deep enough to ensure that the set of branches containing a node labelled by $a$ has measure strictly smaller than $1$.}}. Hence, it follows that on can design a reachability game on a finite graph such that the associated language $\mathcal{L}_{\mathrm{Stoc}}$ by Lemma~\ref{lemma:settingsVSLGood} point (1) is not MSO-definable and therefore for that game Corollary~\ref{cor:Parys} is useless to decide existence of an almost-surely winning strategies. For the cardinality and the topological settings the situation is much better.

\begin{corollary}\label{cor:Parys-parity}
		Let $\game=(\arena,v_0,\WC)$ be a parity game played on an arena of finite out-degree. Then for each of the following three problems one can construct a formula $\Phi_\WC$ so that the problem reduces to decide wether formula $\Phi_\WC$ holds on $\arena$ (\resp on the unfolding of $\arena$).
	\begin{enumerate}
		\item Decide whether the leaking value of $\game$ is countable.
		\item Decide whether the leaking value of $\game$ is smaller than some fixed $k$.
		\item Decide whether \Eloise has a topologically-good strategy in $\game$.
	\end{enumerate}
\end{corollary}

\begin{proof}
This is obtained by combining Lemma~\ref{lemma:settingsVSLGood} and Theorem~\ref{theo:parys} (or Corollary~\ref{cor:Parys} if one wants the statement on the arena and not on the unfolding)	together with the fact that the languages $\mathcal{L}_{\mathrm{Card}}$ and $\mathcal{L}_{\mathrm{Topo}}$ from Lemma~\ref{lemma:settingsVSLGood} are $\omega$-regular (equivalently, MSO-definable) when considering parity conditions (see \cite{CarayolS17} for a proof of this result).
\end{proof}

\begin{remark}\label{rk:Parys}
One directly obtains from Corollary~\ref{cor:Parys-parity} decidability over various classes of arenas that enjoy MSO-decidability (or whose unfolding does): finite arenas, pushdown arenas~\cite{Walukiewicz01} or even CPDA arenas \cite{HMOS08}. We will discuss later these results (and how to obtain them differently) in Section~\ref{subsection:Games-on-Infinite-Arenas-perfect} .
\end{remark}

In the case where the graph is finite one can also safely restrict the set of strategies for \Eloise and \Abelard to finite-memory strategies. More precisely,

\begin{corollary}\label{cor:Parys-finite-graph-strat-finite-memory}
		Let $G=(V,E)$ be a finite graph, let $\colors$ be a finite set, let $\col: V\rightarrow\colors$ be a mapping, let $\arena=(G,\VE,\VA,\VN)$ be an arena and let $v_0$ be an initial vertex. Let $\mathcal{L}$ be an MSO-definable set of $\colors$-labelled $V$-trees. Then the following are equivalent.
		\begin{itemize}
			\item 	\Eloise has an $\mathcal{L}$-good strategy from $v_0$.
			\item 	\Eloise has an $\mathcal{L}$-good finite-memory strategy from $v_0$.
			\item 	\Eloise has {a finite memory} strategy $\strat_\Ei$ such that for every finite-memory strategy $\strat_\Ai$ of \Abelard, $\treeGame{v_0}{\strat_\Ei}{\strat_\Ai}$ belongs to $\mathcal{L}$.
		\end{itemize}
\end{corollary}

{
\begin{proof}
	From the proof of Theorem~\ref{theo:parys} we get that an \Eloise's $\mathcal{L}$-good strategy can be defined in MSO, and as $G$ is finite it can be implemented by a finite transducer hence, be finite-memory. Now, it remains to prove that one can without loss of generality restrict \Abelard's strategies to be finite-memory. For that, we use the same argument. Fix a \emph{finite-memory} strategy $\strat_\Ei$ of \Eloise: the set of \Abelard's strategies $\strat_\Ai$ such that the set of losing plays for \Eloise is MSO definable in a synchronised product of the arena together with a transducer implementing strategy $\strat_\Ei$, \ie it is MSO definable on a fixed finite graph. Hence, if this set is non-empty it contains a finite-memory strategy.
\end{proof}
}

Note that the results from the previous two corollary are somehow not very satisfactory. Indeed, they only apply to $\omega$-regular winning conditions and moreover the computational complexity may be very costly (due to the fact that one works with MSO logic whose decidability is {tower-exponential in the number of quantifier alternations which in our case is quite high}). For those reasons we consider alternative approaches in the next Section. {This permits to significantly reduce the complexity for decidable instances (\eg in the setting of pushdown arenas) and also to tackle, for the cardinality setting, winning conditions not captured by MSO logic, \ie beyond $\omega$-regular ones.}

\section{Perfect-information Games With Nature: Decision Problems}
\label{section:perfect-decision}

We first consider the following two problems regarding the leaking value of a game: \begin{inparaenum}
 \item Is the leaking value is at most $\aleph_0$? \item For some given $k\in\mathbb{N}$, is the leaking value is smaller or equal than $k$? 	
 \end{inparaenum}
 
 For both questions, we make no assumption on the game itself, namely we do not restrict the class of arenas neither the winning conditions. As we are working in such a general setting, we do not focus on decidability but rather on finding reductions to questions on games \emph{without} \Nature. More precisely, for both problems we provide a transformation of the arena and of the winning condition such that the problem reduces to the existence of a winning strategy for \Eloise in a game without \Nature played on the new arena and equipped with the new winning condition. This occupies Section~\ref{section:perfectLeakingCountable} and Section~\ref{section:perfectLeakingFinite}.

Finally, in Section~\ref{section:perfectTopo}, we follow the same approach but for the existence of topologically good strategies. However, we need to restrict our attention to games where \Abelard is not playing.

Of course, for special classes of arenas and of winning conditions, those reductions implies decidability and various important consequences that we discuss in detail in Section~\ref{section:consequences-perfect}.

\subsection{A Game to Decide If the Leaking Value Is at Most $\aleph_{0}$}\label{section:perfectLeakingCountable}

Our goal in this section is to design a technique to decide if the leaking value is at most $\aleph_{0}$ in a given two-player game
with Nature for {an arbitrary} Borel winning condition. 

Fix a graph $G=(V,E)$, an arena $\arena=(G,\VE,\VA,\VN)$ and a game $\game=(\arena,v_0,\WC)$ where $\WC$ is a \emph{Borel} winning condition. 
 We design a two-player perfect-information game \emph{without Nature} $\gameL=(\arenaL,v_0,\WCL)$ such that \Eloise wins $\gameL$ if and only if $\LVal{\game} \leq \aleph_{0}$.

Intuitively in the game $\gameL$, every vertex $v$ of Nature is replaced by
a gadget (see Figure~\ref{figure:gadget-leaking-perfect}) in which \Eloise announces a successor $w$ of $v$ (\ie some $w \in E(v)$) that she wants to \emph{avoid} and then \Abelard chooses a successor of $v$. If he picks $w$ we say that he \emph{disobeys} \Eloise otherwise he \emph{obeys} her. In vertices of  \Eloise and \Abelard, the game $\gameL$ works the same as the game $\game$.
The winning condition $\WCL$ for \Eloise is either that the play (without the gadget nodes) belongs to $\WC$ or that \Abelard does not obeys \Eloise infinitely often (\ie after some point, \Abelard always disobeys \Eloise). Remark that, this is in particular the case if, after some point, no vertex corresponding to a vertex of \Nature is encountered.

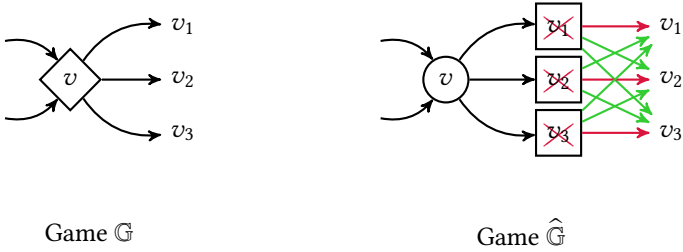
\begin{figure}[htb]
	\centering
	\begin{tikzpicture}[>=stealth',thick,scale=1,transform shape]
\tikzstyle{Abelard}=[draw, minimum size=.6cm]
\tikzstyle{Eloise}=[draw,circle, minimum size=.6cm]
\tikzstyle{Nature}=[draw,diamond]
\tikzstyle{Buchi}=[fill=MediumSpringGreen]
\tikzstyle{loopleft}=[in=150,out=210]
\tikzstyle{loopright}=[in=-30,out=30]
\tikzstyle{loopbelow}=[in=-120,out=-60]
\tikzstyle{loopabove}=[in=120,out=60]
\tikzstyle{OK}=[LimeGreen]
\tikzstyle{bad}=[Crimson]

\begin{scope}
\node[Nature] (v) at (0,0) {$v$};
\node (i1) at (-1,.5) {};
\node (i3) at (-1,-.5) {};
\path[->] (i1) edge[bend left] (v);
\path[->] (i3) edge[bend right] (v);

\node (v1) at (1.5,.7) {$v_1$};
\node (v2) at (1.5,0) {$v_2$};
\node (v3) at (1.5,-.7) {$v_3$};
\path[->] (v) edge[bend left] (v1);\path[->] (v) edge (v2);\path[->] (v) edge[bend right] (v3);

\node at (0.25,-2) {\textcolor{black}{Game $\game$}};
\end{scope}

\begin{scope}[xshift = 5cm]
\node[Eloise] (v) at (0,0) {$v$};
\node (i1) at (-1,.5) {};
\node (i3) at (-1,-.5) {};
\path[->] (i1) edge[bend left] (v);
\path[->] (i3) edge[bend right] (v);

\node[Abelard] (nv1) at (1.5,.7) {\xcancel{$v_1$}};
\node[Abelard] (nv2) at (1.5,0) {\xcancel{$v_2$}};
\node[Abelard] (nv3) at (1.5,-.7) {\xcancel{$v_3$}};
\path[->] (v) edge[bend left] (nv1);\path[->] (v) edge (nv2);\path[->] (v) edge[bend right] (nv3);

\node (v1) at (3,.7) {$v_1$};
\node (v2) at (3,0) {$v_2$};
\node (v3) at (3,-.7) {$v_3$};
\path[->,bad] (nv1) edge (v1);
\path[->,OK] (nv1) edge (v2);
\path[->,OK] (nv1) edge (v3);
\path[->,OK] (nv2) edge (v1);
\path[->,bad] (nv2) edge (v2);
\path[->,OK] (nv2) edge (v3);
\path[->,OK] (nv3) edge (v1);
\path[->,OK] (nv3) edge (v2);
\path[->,bad] (nv3) edge (v3);
\node at (1,-2) {\textcolor{black}{Game $\gameL$}};

\end{scope}
\end{tikzpicture}
 \caption{Example of the gadget used to defined $\gameL$, where $\xcancel{v_i}$ is a shorthand for $(v,v_i)$, the node where \Eloise indicates she would prefer avoiding $v_i$ from $v$}\label{figure:gadget-leaking-perfect}
\end{figure}

Formally one defines $\graphL = (\VL,\EL)$ where $\VL=\VL_\Ei\cup \VL_\Ai$, $\VL_\Ei = V_\Ei \cup \VN$, $\VL_\Ai = \VA \cup \{ (v,w) \mid v \in \VN \;\text{and}\; w \in E(v) \}$ and 
\begin{multline*}
\EL=    E \;\setminus\; (\VN \times V) \;  \cup \; \{ (v,(v,w)) \mid v \in \VN\;\text{and}\; w \in E(v) \} \\
 \cup \; \{ ((v,w),w') \mid v \in \VN\;\text{and}\; w,w' \in E(v)  \}. 
\end{multline*}

For ease of presentation, we view a partial play $\hat{\pi}$ in $\gameL$ as a partial play $\pi$ in $\game$ together with a mapping associating to every prefix of $\pi$ ending in $\VN$ (with the possible exception of $\pi$ itself) the successor that \Eloise wishes to avoid.

Formally for a partial play $\hat{\pi}$ in $\gameL$, we denote by $\encod{\hat{\pi}}$ the partial play of $\game$ obtained
by removing all occurrences of vertices in $\VN \times V$ from $\hat{\pi}$.
A partial play $\hat{\pi}$ in $\gameL$ is entirely characterised by the pair
$(\pi,\xi)$ where $\pi$ is the partial play $\encod{\hat{\pi}}$  and $\xi$ is the mapping such for all $\pi' \prefix \pi$, $\xi(\pi')=w$ if and only if 
there exists $\hat{\pi}' \prefix \hat{\pi}$ with $\encod{\hat{\pi}'}=\pi'$ and $\hat{\pi}'$ ends in a vertex of the form $(v,w)$ for some $v \in \VN$.
In the following, we do not distinguish between a pair $(\pi,\xi)$ satisfying these conditions and the unique corresponding partial play. We adopt the same point of view for (infinite) plays.

Finally, the winning condition $\WCL$ is defined by 
$$\WCL = \{ (\lambda,\xi) \mid \lambda \in \WC \}  \\
        \cup \{ (\lambda,\xi) \mid \exists^{<\infty} \pi v \prefixstrict \lambda, \pi \in \Dom(\xi) \;\text{and}\; v \neq \xi(\pi)  \}  
$$
\ie $\WCL$ contains those plays that project to a winning play in $\game$ as well as those plays where \Abelard does not obeys \Eloise infinitely often.

{
\begin{remark}
\label{remark:determincay-leaking-perfect}
As $\WC$ is assumed to be a Borel subset of plays in $\game$, $\WCL$ is a Borel subset of the set of plays in $\gameL$. Indeed, the second part of the condition (which does not involve $\WC$) is Borel. As the first part is the inverse image of $\WC$ under the \emph{continuous} mapping $\hat{\lambda} \mapsto \encod{\hat{\lambda}}$, it is also Borel. Using Borel determinacy \cite{Martin75} the game $\gameL$ is determined, \ie either \Eloise or \Abelard has a winning strategy in $\gameL$. {Furthermore, remark that if $\WC$ is $\omega$-regular then so is $\WCL$.}
\end{remark}}

The following theorem relates the games $\game$ and $\gameL$.

\begin{theorem}\label{theo:perfect:main-aleph}
{Let $\game$ be a game. The leaking value in $\game$ is at most $\aleph_0$ if and only if \Eloise has a winning strategy in $\gameL$.

More precisely, from a winning strategy (\resp positional winning strategy, \resp finite-memory  winning strategy) $\phiNR_{\Ei}$ of \Eloise in $\gameL$, we can define a strategy (\resp positional winning strategy, \resp finite-memory  winning strategy) $\phi_{\Ei}$ for \Eloise in $\game$ such that $\CL{\phi_{\Ei}} \leq
\aleph_{0}$. 

Moreover, from a winning strategy (\resp a positional winning strategy,\resp finite-memory  winning strategy) $\phiNR_{\Ai}$ for \Abelard, we can define 
  a strategy (\resp a positional strategy,\resp finite-memory  winning strategy)
$\phi_{\Ai}$ for \Abelard in $\game$ such that \emph{for any} strategy $\phi_{E}$ of \Eloise  $\cardinal{\outcomes{v_0}{\phi_\Ei}{\phi_\Ai}\setminus\WC} = 2^{\aleph_0}$.}
\end{theorem}

\begin{proof}
First assume that \Eloise has a winning strategy $\phiNR_\Ei$ in $\gameNR$. We define a strategy $\phi_\Ei$ for her in $\game$ as follows. 
For any partial play  $\pi$ in $\game$ ending in $\VE$,
 if there exists a partial play of the form $(\pi,\xi)$ in $\gameNR$ in which \Eloise respects $\phiNR_\Ei$ then this play is unique and we let $\phi_{\Ei}(\pi)=\phiNR_\Ei((\pi,\xi))$. Otherwise $\phi_{\Ei}(\pi)$ is undefined.

A straightforward induction shows that for each partial play $\pi$ ending in $\VE$ where \Eloise respects $\phi_\Ei$ the strategy $\phi_{\Ei}$ is defined.
{Furthermore remark that if $\phiNR_{\Ei}$ is positional (\resp uses finite-memory), $\phi_{\Ei}$
is also positional (\resp also uses finite-memory).}

Let us now prove that $\CL{\strat_\Ei}\leq\aleph_0$. For this, fix a strategy $\strat_\Ai$ of \Abelard in $\game$ and consider a play $\lambda$ in $\outcomes{v_0}{\phi_\Ei}{\phi_\Ai} \setminus \WC$ that is losing for \Eloise.
As \Eloise respects $\strat_\Ei$ in $\play$, there exists by construction of $\phi_{\Ei}$, a unique play 
 of the form $(\play,\xi_{\lambda})$ in $\gameL$ where \Eloise respects $\phiNR_\Ei$. As $\phiNR_\Ei$ is winning in $\gameL$, the corresponding play $(\lambda,\xi_{\lambda})$ is won by \Eloise and this can only be because \Abelard obeys \Eloise
only finitely often (indeed, recall that $\lambda$ is losing for her in $\game$). Let $\pi_{\lambda}$ be the longest prefix of $\lambda$
of the form $\pi v$ with $\pi \in \Dom(\xi_{\lambda})$ and $v \neq \xi_{\lambda}(\pi)$ (\ie $\pi_{\lambda}$ is the last time where \Abelard obeys
\Eloise).

We claim that $\lambda \in \outcomes{v_0}{\phi_\Ei}{\phi_\Ai} \setminus \WC$
is uniquely characterised by $\pi_{\lambda}$. In particular it will imply that $\outcomes{v_0}{\phi_\Ei}{\phi_\Ai} \setminus \WC$ is countable as it can be injectively mapped into the countable set $V^{*}$. 

Let $\lambda_{1} \neq \lambda_{2}  \in \outcomes{v_0}{\phi_\Ei}{\phi_\Ai} \setminus \WC$ and let $(\lambda_{1},\xi_{1})$ and $(\lambda_{2},\xi_{2})$ 
be the corresponding plays in $\gameL$. We will show that $\pi_{\lambda_{1}} \neq \pi_{\lambda_{2}}$.
Consider the greatest common prefix $\pi$  of $\lambda_{1}$ and $\lambda_{2}$. In particular there exists $v_{1} \neq v_{2} \in V$ such that $\pi v_{1} \prefixstrict \lambda_{1}$ and $\pi v_{2} \prefixstrict \lambda_{2}$. As $\lambda_{1}$ and $\lambda_{2}$
respects the same strategies for \Eloise and \Abelard, $\pi$ must end
in $\VN$. Moreover for all prefixes of $\pi$ (including $\pi$), $\xi_{\lambda_1}$ and $\xi_{\lambda_2}$ coincide. Let $w = \xi_{\lambda_{1}}(\pi)=\xi_{\lambda_{2}}(\pi)$ be the vertex \Eloise wants to avoid in $\pi$. Assume without loss of generality that $w \neq v_{1}$. \Ie
 \Abelard obeys \Eloise at $\pi$ in  $(\lambda_{1},\xi_{1})$. In particular,
$\pi v_{1} \prefix \pi_{\lambda_{1}}$: therefore $\pi_{\lambda_{1}} \not\prefix \pi_{\lambda_{2}}$ and thus $\pi_{\lambda_{1}} \neq \pi_{\lambda_{2}}$.

Conversely, assume that \Eloise has no winning strategy in $\gameNR$. By Remark~\ref{remark:determincay-leaking-perfect}, \Abelard has a winning strategy $\phiNR_\Ai$ in $\gameNR$. 

Using $\phiNR_\Ai$ we define a  strategy $\phi_\Ai$ of \Abelard in $\game$ that is only partially defined. {It can be turned into a full strategy by picking an arbitrary move for \Abelard for all partial plays where it is not defined. This transformation can only increase the set of losing plays for \Eloise and hence we can work with $\phi_{\Ai}$ as is.}

The strategy $\phi_\Ai$ uses as a memory a partial play in $\gameL$, \ie with any partial play $\pi$ in $\game$ where \Abelard respects $\phi_{\Ai}$ we associate a partial play $\tau(\pi)=(\pi,\xi)$ in $\gameNR$ where \Abelard respects $\phiNR_{\Ai}$. The definition of both $\phiNR_{\Ai}$ and $\tau$ are done by induction.

 Initially when $\pi = v_0$ one lets $\tau(\pi)=(v_0,\xi)$ where $\xi$ is defined nowhere. Now, 
assume that the current partial play is $\pi$ and that it ends in some vertex $v$ and assume that $\NRmap{\pi}=(\pi,\xi)$.
\begin{itemize}
\item If $v\in\VA$ then $\phi_\Ai(\pi) = \phiNR_\Ai((\pi,\xi)) = v'$ and $\NRmap{\pi\cdot v'} = (\pi\cdot v',\xi)$. 
\item If $v\in\VE$ and \Eloise moves to some $v'$ then $\NRmap{\pi\cdot v'} = (\pi\cdot v',\xi)$. 
\item If $v\in\VN$ and \Nature moves to some $v'$ then $\NRmap{\pi\cdot v'}$
is defined only if there exists at least one $w \in E(v)$ such that 
$\phiNR_{\Ai}(\pi,\xi[\pi \mapsto w])=v'$ where we denote by $\xi[\pi \mapsto w]$ the extension of $\xi$ where $\pi$ is mapped to $w$. In this case,
if $\phiNR_{\Ai}(\pi,\xi[ \pi \mapsto v'])=v'$ then we take 
$\NRmap{\pi\cdot v'}=(\pi \cdot v',\xi[ \pi \mapsto v'])$. Otherwise
we pick $w \in E(v)$ such that 
$\phiNR_{\Ai}(\pi,\xi[\pi \mapsto w])=v'$ and set $\NRmap{\pi\cdot v'}=(\pi \cdot v',\xi[ \pi \mapsto w])$.
\end{itemize}

In the last case, remark that $\NRmap{\pi\cdot v'}$ is always defined for 
at least one $v' \in E(v)$. 
Indeed, consider any node $w\in E(v)$ and set $v'=\phiNR_{\Ai}(\pi,\xi[\pi \mapsto w])$: then for this $v'$ $\NRmap{\pi\cdot v'}$ is defined.
Furthermore if it is defined for exactly one 
$v' \in E(v)$, then it is equal some to $(\pi \cdot v',\xi)$ with $\xi(\pi)=v'$: indeed, it means that $\phiNR_{\Ai}(\pi,\xi[\pi \mapsto w])=v'$ for every $w$, and in particular for $w=v'$ and therefore $\NRmap{\pi\cdot v'}=(\pi \cdot v',\xi[ \pi \mapsto v'])$. This means in particular that \Abelard disobeys \Eloise.

{Finally remark that if $\phiNR_{\Ai}$ is positional (\resp uses finite-memory) then $\phi_{\Ai}$
is also positional (\resp also uses finite-memory).}

Let $\phi_\Ei$ be a strategy for \Eloise in $\game$. In order to prove that $\CL{\strat_\Ei}=2^{\aleph_0}$ we will establish the following stronger result: $\cardinal{\outcomes{v_0}{\phi_\Ei}{\phi_\Ai}\setminus\WC} = 2^{\aleph_0}$.

First remark\footnote{This is no longer true for the full version of $\phi_{\Ai}$.} that $\outcomes{v_0}{\phi_\Ei}{\phi_\Ai} \cap \WC = \emptyset$. Indeed, consider a play $\lambda \in \outcomes{v_0}{\phi_\Ei}{\phi_\Ai}$. By construction of $\phi_{\Ai}$, there exists a play of the form $(\lambda,\xi)$ in $\gameL$ where \Abelard respects $\phiNR_{\Ai}$: in particular it implies that $\lambda\notin\WC$. 

It remains to show that $\cardinal{\outcomes{v_0}{\phi_\Ei}{\phi_\Ai}} \geq 2^{\aleph_{0}}$. Consider the tree $\treeGame{v_0}{\phi_\Ei}{\phi_\Ai}$ of
all partial plays respecting both $\phi_{\Ei}$ and $\phi_{\Ai}$. To show
that $\treeGame{v_0}{\phi_\Ei}{\phi_\Ai}$ has $2^{\aleph_{0}}$ branches, it is enough to show that every infinite branch in $\treeGame{v_0}{\phi_\Ei}{\phi_\Ai}$ goes through infinitely many nodes with at least $2$ successors.

Let $\lambda$ be a branch in $\treeGame{v_0}{\phi_\Ei}{\phi_\Ai}$ and
let $\tau(\lambda)=(\lambda,\xi)$ be the corresponding play in $\gameL$.
As $\tau(\lambda)$ is won by \Abelard, he obeys \Eloise infinitely often during this play. Hence there exists 
$\pi_{1} v_{1} \prefixstrict \pi_{2} v_{2} \prefixstrict \cdots \prefixstrict \lambda$ such that for all $i \geq 1$, $\pi_{i}$ ends in $V_{N}$ and 
$\xi(\pi_{i}) \neq v_{i}$. As remarked previously for all $i \geq 0$, $\pi_{i}$
has at least two successors in $\treeGame{v_0}{\phi_\Ei}{\phi_\Ai}$ (as otherwise it would imply that \Abelard disobeys \Eloise at $\pi_{i}$ in $\tau(\lambda)$).
\end{proof}

{
\begin{remark}
One should think of the last part of the statement of Theorem~\ref{theo:perfect:main-aleph} as a determinacy result in the spirit Borel determinacy \cite{Martin75}. Indeed, it states that if \Eloise does not have a strategy that is good against every strategy of \Abelard then he has one that is bad (for her) against any of her strategies. 
\end{remark}
}

\subsection{A Game to Decide If the Leaking Value Is Smaller Than Some $k$}\label{section:perfectLeakingFinite}

Our goal in this section is to design a technique to decide if the leaking value is smaller than some fixed $k$ in a given two-player game with Nature for {an arbitrary} Borel winning condition. 

Fix a graph $G=(V,E)$, an arena $\arena=(G,\VE,\VA,\VN)$ and a game $\game=(\arena,v_0,\WC)$ where $\WC$ is a \emph{Borel} winning condition. Fix a bound $k\geq 0$.
 We design a two-player perfect-information game \emph{without Nature} $\gameFk=(\arenaFk,s,\WCFk)$ such that \Eloise wins $\gameFk$ if and only if $\LVal{\game} \leq k$.

Intuitively the main vertices in the game $\gameFk$ are pairs formed by a vertex from $V$ together with an integer $i$ such that $0\leq i\leq k$ that indicates the maximum number of plays \Eloise is claiming that she may loose. A vertex $(v_0,i)$ is controlled by the same player that controls $v$ in $\game$. For technical reasons we also add an initial vertex $s$ that is controlled by \Eloise. 

Informally a play in $\gameFk$ proceeds as follows. In the initial move, \Eloise goes from $s$ to a vertex $(v,i)$ for some $0\leq i\leq k$. Then we have the following situation depending on the current vertex.
\begin{itemize}
	\item If the play is  in some vertex $(v,i)$ with $v\in \VE$, \Eloise can move to any $(w,i)$ with $w\in E(v)$. 
	\item If the play is in some vertex $(v,i)$ with $v\in \VA$, \Abelard can move to any $(w,i,?)$ that is a vertex controlled by \Eloise and from which she can decrease the integer value by going to any vertex $(w,j)$ with $0\leq j\leq i$.
	\item If the play is in some vertex $(v,i)$ with $v\in \VN$, \Eloise can move to a vertex $\mu$ that stands for a function (we overload $\mu$ here) with domain $E(v)$ and that takes its value in $\{0,\dots,i\}$ and is such that $\sum_{w\in E(v)}\mu(w)=i$, \ie she indicates for every possible successor of $v$ a new integer whose values sum to $i$. Then \Abelard can choose any $w\in E(v)$ and the play goes to $(w,\mu(w))$.
\end{itemize}

Formally one defines $\graphFk = (\VF,\EF)$ where $\VF=\VF_\Ei\cup \VF_\Ai$, 
\begin{align*}
\VF_\Ei\;=\; & 
	\{s\} 
	\cup \{(v,i)\mid v\in V_\Ei\cup V_\Ni,\ 0\leq i \leq k \}
	\cup \{(v,i,?)\mid v\in V_\Ai,\ 0\leq i \leq k \},
\end{align*}
\begin{align*}
\VF_\Ai \; =\; & 
	\{(v,i)\mid v\in V_\Ai,\ 0\leq i \leq k \}\\
	&\cup \{\mu \mid \exists v\in V_\Ni \text{ s.t. } \mu:E(v)\rightarrow \{0,\dots,k\} \text{ and }\sum_{w\in E(v)}\mu(w)\leq k\}
\end{align*}
and
\begin{align*}
\EF\; =\; &
	\{(s,(v_0,i))\mid 0\leq i\leq k\}\\
	&\cup  \{((v,i),(w,i)\mid v\in V_\Ei \text{ and } w\in E(v)\}
	\cup \{((v,i),(w,i,?)\mid v\in V_\Ai \text{ and } w\in E(v)\}\\
	& \cup \{((w,i,?),(w,j)\mid i\geq j\}\\
	& \cup \{((v,i),\mu)\mid v\in V_\Ni \text{ and } \mu:E(v)\rightarrow \{0,\dots,k\} \text{ s.t. }\sum_{w\in E(v)}\mu(w)=i\}\\
	& \cup \{(\mu,(w,\mu(w)))\mid w \text{ is in the domain of $\mu$}\}
\end{align*}
Finally we let $\arenaFk=(\graphFk ,\VF_\Ei,\VF_\Ai)$.

Let $\playF$ be a play in $\arenaFk$ and let us define $\rho(\playF)\in V^\omega$ to be the play obtained from $\playF$ by keeping only the vertices in $V\times \{0,\dots,k\}$ and then projecting them on the $V$ component. The play $\playF$ is winning for \Eloise if one of the following holds:
\begin{itemize}
	\item $\rho(\playFk)\in \WC$; or
	\item no vertex of the form $(v,0)$ is visited in $\playF$.
\end{itemize}

Call $\WCFk$ the corresponding winning condition, \ie 
$\WCFk=\{\playF\mid \rho(\playF)\in \WC\}\cup (\VF^\omega\setminus \VF^*(V\times\{0\})\VF^\omega) $ and let $\gameFk=(\arenaFk,s,\WCFk)$.

The following theorem relates both games $\game$ and $\gameFk$.

\begin{theorem}\label{theo:perfect:main-k}
Let $\game$ be a game and let $k\geq 0$ be an integer. The leaking value in $\game$ is smaller or equal than $k$ if and only if \Eloise has a winning strategy in $\gameFk$.

More precisely, from a winning strategy (\resp finite-memory winning strategy) $\phiF_{\Ei}$ of \Eloise in $\gameFk$, we can define a strategy (\resp finite-memory winning strategy) $\phi_{\Ei}$ for \Eloise in $\game$ such that $\CL{\phi_{\Ei}} \leq k$. 

\end{theorem}

\begin{proof}
	Assume first that $\LVal{\game}\leq k$ and let $\phi_\Ei$ be a strategy of \Eloise that witnesses it (note that as the leaking value is finite such an optimal strategy for \Eloise necessarily exists). 
	
	With any partial play $\play$ in $\game$ where \Eloise respects $\phi_\Ei$ we associate an integer $\tau(\play)\leq k$ by letting 
	$$\tau(\play) = \max\{\cardinal{\outcomes{\play}{\phi_\Ei}{\phi_\Ai}\setminus\WC}\mid \phi_\Ai \text{ \Abelard's strategy } \}$$ where $\outcomes{\play}{\phi_\Ei}{\phi_\Ai} = \outcomes{v_0}{\phi_\Ei}{\phi_\Ai}\cap \play V^\omega$ denotes the (possibly empty) set of infinite plays that starts by $\play$ and where \Eloise (\resp \Abelard) respects $\phi_\Ei$ (\resp $\phi_\Ai$). Remark that $\tau(v_0) = \LVal{\game}$. Moreover, when $\play$ increases (with respect to the prefix ordering) $\tau$ is easily seen to be decreasing and therefore it implies that $\tau(\play)$ is always smaller or equal than $k$.
	
	We now define a strategy $\phiFk_\Ei$ for \Eloise in $\gameFk$ as follows, where we let $\play=\rho(\playF)$:
	\begin{itemize}
		\item $\phiFk_\Ei(s)=(v_0,\LVal{\game})$;
		\item $\phiFk_\Ei(\playF)=(\phi(\lambda),i)$ if $\playF$ ends in some vertex $(v,i)$ with $v\in V_\Ei$;
		\item $\phiFk_\Ei(\playF)=(w,\tau(\play\cdot w))$ if $\playF$ ends in some vertex $(w,i,?)$;
		\item $\phiFk_\Ei(\playF)=\mu$ if $\playF$ ends in some vertex $(v,i)$ with $v\in V_\Ni$ where $\mu: E(v)\rightarrow \{0,\dots,k\}$ is defined by letting $\mu(w) = \tau(\play\cdot w)$.
	\end{itemize}
	
	It easily follows from the definition of $\phiFk_\Ei$ that, for any partial play $\playF$ starting from $s$ and where \Eloise respects $\phiFk_\Ei$ one has the following.
	\begin{enumerate}[(i)]
		\item The play $\rho(\playF)$ is a play in $\game$ that starts in $v_0$ and where \Eloise respects $\phi_\Ei$.
		\item If $\playF$ ends in some vertex $(v,i)$ then $i=\tau(\rho(\playF))$.
	\end{enumerate}
	
	Note that the two above properties implies the following. 
	\begin{enumerate}[(i)]\setcounter{enumi}{2}
		\item\label{item:csq_tape0} If some play $\playF$ where \Eloise respects $\phiFk_\Ei$ eventually visits a vertex of the form $(v,0)$ then $\rho(\playF)\in \WC$. 
	\end{enumerate}
	
	We now establish that the strategy $\phiFk_\Ei$ is winning for \Eloise in $\gameFk$. For that, assume toward a contradiction that there is a losing play $\playF$ for \Eloise where she respects $\phiFk_\Ei$. By definition of $\WCFk$, it means that $\playF$ contains a vertex of the form $(v,0)$ and that $\rho(\playFk)\notin \WC$, which contradicts property~(\ref{item:csq_tape0}) above. 
	
	We now turn to the converse implication. Hence, we assume that \Eloise has a winning strategy $\phiFk_\Ei$ in $\gameFk$. Thanks to $\phiFk_\Ei$ we define a strategy $\phi_\Ei$ for \Eloise in $\game$. This strategy will associate with any partial play $\play$ a partial play $\playFk$ in $\gameFk$ where \Eloise respects $\phiFk_\Ei$ and such that $\rho(\playFk)=\play$. Initially, when $\play=v_0$ we let $\playFk=s\cdot \phiFk(s)$. Now consider a partial play $\play$ where \Eloise respects $\phi$. Then, we do the following.
	\begin{itemize}
		\item If $\play$ ends in some vertex $v\in V_\Ei$ and if $\phiFk(\playFk)=(w,i)$ then we let $\phi(\play)=w$ and we associate with $\play\cdot w$ the partial play $\playFk\cdot (w,i)$.
		\item If $\play$ ends in some vertex $v\in V_\Ai$ (equivalently $\playFk$ ends in some $(v,i)$ with $v\in V_\Ai$) and if \Abelard moves to some vertex $w$ then we associate with $\play\cdot w$ the partial play $\playFk \cdot (w,i,?)\cdot (w,j)$ where $(w,j) = \phiFk(\playFk \cdot (w,i,?))$.
		\item If $\play$ ends in some vertex $v\in V_\Ni$ and if \Nature moves to some vertex $w$ then we associate with $\play\cdot w$ the partial play $\playFk \cdot \mu \cdot(w,\mu(w))$ where $\mu=\phiFk(\playFk)$.
	\end{itemize}
	Note that one easily verifies that $\playFk$ is a partial play in $\gameFk$ where \Eloise respects $\phiFk_\Ei$ and such that $\rho(\playFk)=\play$.
	
	In the remaining we will prove that $\CL{\phi_\Ei}\leq k$, which implies that $\LVal{\game}\leq k$. The proof goes by contradiction, assuming that \Abelard has a strategy $\phi_\Ai$ in $\game$ such that $\cardinal{\outcomes{v_0}{\phi_\Ei}{\phi_\Ai}\setminus\WC}>k$. Using, $\phi_\Ai$ we define a strategy for \Abelard in $\gameFk$ as follows, where we let $\play=\rho(\playFk)$:
	\begin{itemize}
		\item $\phiFk_\Ai(\playFk)=(w,i,?)$ if $\playFk$ ends in some vertex $(v,i)$ with $v\in V_\Ai$ and $\phi_\Ai(\play)=w$;
		\item $\phiFk_\Ai(\playFk) = (w,\mu(w))$ if $\playFk$ ends in some vertex $\mu$ and $w$ is such that $\mu(w)<\cardinal{\outcomes{\play}{\phi_\Ei}{\phi_\Ai}\setminus\WC}$; in case several such $w$ exist $\phiFk_\Ai$ chooses one minimising $\mu(w)$. Moreover, in case $\mu(w)=0$ the strategy $\phiFk_\Ai$ simply exhibit from that point a losing play extending the current one, which exists as $\cardinal{\outcomes{\play}{\phi_\Ei}{\phi_\Ai}\setminus\WC}\geq 1$.
	\end{itemize}
	Remark that in the above definition (second item), the existence of such a $w$ is verified by induction and is ensured by the initial assumption that $\cardinal{\outcomes{v_0}{\phi_\Ei}{\phi_\Ai}\setminus\WC}>k$ together with the definition of $\phiFk_\Ai$ when the play ends in a $\mu$-vertex. More precisely in a partial play $\playFk$ ending in a vertex $(v,i)$ we always have that $i<\cardinal{\outcomes{\play}{\phi_\Ei}{\phi_\Ai}\setminus\WC}$. 
	
	We will show that the play $\playFk$ obtained when \Eloise respects $\phiFk_\Ei$ and  \Abelard respects $\phiFk_\Ai$ is won by \Abelard, hence leading a contradiction. First note that if $\playFk$ eventually visits a vertex of the form $(v,0)$ then by the definition of $\phiFk_\Ai$ one has $\play=\rho(\playFk)\notin \WC$. Hence, it suffices to prove that $\playFk$ eventually visits a vertex of the form $(v,0)$. Assume this is not the case. Hence, after some point all (main) vertices in $\playFk$ are of the form $(v,i)$ with the same integer $i$: call $\playFk'$ the prefix of $\playFk$ ending in the first such vertex and consider the set $\outcomes{\rho(\playFk')}{\phi_\Ei}{\phi_\Ai}$. This set contains at least $i+1$ plays not in $\WC$. But as after $\playFk'$ the integer stays equal to $i$ it implies that $\cardinal{\outcomes{\rho(\playFk')}{\phi_\Ei}{\phi_\Ai}\setminus\WC}\leq 1$ (the only possibly losing play being $\play$ as all other possible move of \Nature leads to a situation where no more play is losing for \Eloise as otherwise $\strat_\Ai$ would indicate to mimic it). But as $i\geq 1$ and $\cardinal{\outcomes{\rho(\playFk')}{\phi_\Ei}{\phi_\Ai}\setminus\WC}\geq i$ it leads a contradiction. Hence, $\playFk$ eventually visits a vertex of the form $(v,0)$ and therefore, as already noted, it implies that $\playFk\notin\WCFk$. Hence, this contradicts the fact that $\phiFk_\Ei$ is losing and concludes the proof of the converse implication.
	
	The fact that if $\phiFk_\Ei$ has finite-memory then so does $\phi_\Ei$ is by definition.
\end{proof}

\subsection{A Game to Decide the Existence of a topologically-Good Strategy}\label{section:perfectTopo}

Our goal in this section is to design a technique to decide whether \Eloise has a topological good strategy in a perfect-information game with Nature. {Unfortunately, we do not know how to obtain results in the general case, and therefore we focus on games where \Abelard is not playing (\ie one-player game with Nature). However, remember that, as already mentioned in Remark~\ref{rk:Parys}, if the underlying arena has an MSO decidable theory and if the winning condition is $\omega$-regular, Corollary~\ref{cor:Parys-parity} implies decidability of the existence of a topologically-good strategy in the setting where both \Eloise and \Abelard play.}

\subsubsection{Large Sets of Branches and Dense Set of Nodes}

We start by recalling simple results from \cite{CarayolS17} that provide a useful characterisation of large sets of branches in a tree. For this fix a $D$-tree $t$ for some set of directions $D$.
Call a set of nodes $W\subseteq t$ \defin{dense} if  $\forall u\in t$, $\exists v\in W$ such that $u\prefix v$. Given a dense set of nodes $W$, the set $\mathcal{B}(W)$ of {branches supported by $W$} is defined as the set of branches $\pi$ that have infinitely many prefixes in $W$. Formally, 
$$\mathcal{B}(W)=\{\alpha\in D^\omega \mid \exists (u_i)_{i\geq 0}\in W^\omega  \text{such that $\alpha$ is the limit of } (u_i)_{i\geq 0}\}$$
 Using the existence of decomposition-invariant winning strategies in Banach-Mazur games, the following lemma from \cite[Lemma~5]{CarayolS17} characterises large sets of branches (we repeat the proof for sake of completeness).

\begin{lemma}\label{lemma:largeVSdense}
Let $t$ be a $D$-tree for some $D$ and  $B$ be a {Borel} set of branches in $t$. Then $B$ is large if and only if there exists a dense set of nodes  $W\subseteq t$ such that $\mathcal{B}(W) \subseteq B$.
\end{lemma}

\begin{proof}
Assume that $B$ is large and let $f$ be a decomposition-invariant strategy for \Eloise in the associated Banach-Mazur game {(recall that we assumed $B$ to be Borel)}. Consider the set:
\[
 W = \{ vf(v) \mid v \in \{0,1\}^* \}.
\]
The set $W$ is dense (as for all $v \in \{0,1\}^*$, $v \prefixstrict vf(v) \in W$). We claim that $\mathcal{B}(W)$ is included in $B$. Let $\pi$ be a branch 
in $\mathcal{B}(W)$. As $\pi$ has infinitely many prefixes in $W$, there exists
a sequence of words $u_{1},u_{2},\cdots$ such that $u_{1}f(u_{1}) \prefixstrict u_{2}f(u_{2}) \prefixstrict \cdots \prefixstrict \pi$. As the lengths of the $u_{i}$ are strictly increasing, there exists a sub-sequence $(v_i)_{i\geq 1}$ of $(u_i)_{i\geq 1}$ such that for all
$i \geq 1$, $v_{i}f(v_{i}) \prefixstrict v_{i+1}$. Now, consider the play in the Banach-Mazur game where \Abelard first moves to $v_{1}$ and then \Eloise responds by going to $v_{1}f(v_{1})$. Then \Abelard moves to $v_{2}$ (which is possible as $v_{1}f(v_{1}) \prefixstrict v_{2}$) and \Eloise moves to $v_{2}f(v_{2})$. And so on. In this play \Eloise respects the strategy $f$ and therefore wins. Hence, the branch $\pi$ associated with this play belongs to $B$.

Conversely let $W$ be a dense set of nodes such that $\mathcal{B}(W) \subseteq B$. To show that $B$ is large, we define a decomposition-invariant strategy $f$ for \Eloise in the associated Banach-Mazur game.
 For all nodes $u$ we pick  $v$ of $W$ such that $u$ is a strict prefix of $v$ (since $W$ is dense there must always exist such a $v$). Let $v=uu'$ and fix $f(u)=u'$. A play where \Eloise respects $f$ goes through infinitely many 
nodes in $W$ (as $f$ always points to an element in $W$). Hence, the branch associated with the play belongs to $\mathcal{B}(W)\subseteq B$ which shows that $f$ is winning for \Eloise.\end{proof}

In order to describe a dense set of nodes, we mark a path to this set in the tree as follows. Let $t$ be a tree. A \defin{direction mapping} is a mapping $d:t \rightarrow D$, and given a set of nodes $W$, we say that $d$ \emph{points to $W$} if for every node $u$ there exists $d_1,\ldots,d_k\in D$ such that $u d_1\cdots d_k \in W$ and for all $1\leq j\leq k$, $d_j = d(u d_1 \cdots d_{j-1})$. We have the following result from \cite[Lemma~6]{CarayolS17} (we repeat the proof for sake of completeness).

\begin{lemma}\label{directions}
A  set of nodes $W$ is dense if and only if there exists a direction mapping that points to $W$.
\end{lemma}

\begin{proof}
Assume that $W$ is dense. We define $d(v)$ by induction on $v$ as follows. Let $v$ such that $d(v)$ is not yet defined, we pick a node $v d_1\cdots d_k \in U$ (there must exists one since $W$ is dense), and for all $j\leq k$ we define 
\[d(v i_1 \cdots d_{j-1})= d_j.\] The mapping is defined on every node and satisfies the requirement by definition. The other implication is straightforward (for all nodes $v$, there exists $v d_1\cdots d_k \in W$).
\end{proof}

\subsubsection{Simulation Game}

Fix a graph $G=(V,E)$, an {arena} $\arena=(G,\VE,\VA,\VN)$ where we have $\VA=\emptyset$  (\ie \Abelard is not part of the game) and a game $\game=(\arena,v_0,\WC)$ where we assume that $\WC$ is \emph{Borel}. 
{We assume that $v_0\in\VE$ and that the game is turn based, meaning that along a play the pebble alternatively visits $\VE$ and $\VN$, which formally means that $E\subseteq \VE\times \VN\cup \VN\times \VE$. This restriction is not essential but highly simplifies the presentation.}

We design a two-player perfect-information game without Nature such that \Eloise wins in $\gameT=(\arenaT,v_0,\WCT)$ if and only if she has a topologically-good strategy in $\game$.

The arena $\arenaT$ of the game $\gameT$ is quite similar to $\arena$ and the main intuition is that \Eloise mimics a play against \Nature in $\game$ and additionally describes a dense set of nodes $W$ (thanks to a direction mapping and an explicit annotation of nodes in $W$) in the tree of possible outcomes.  \Abelard simulates the moves of \Nature and he tries either to prove that $W$ is not dense or that there is a losing play in $\mathcal{B}(W)$. 
Formally one defines $\graphT = (\VT,\ET)$ where $\VT=\VT_\Ei\cup \VT_\Ai$, $\VT_\Ei = V_\Ei$, $\VT_\Ai = V_\Ni\times V_{\Ei} \times\{\top,\bot\}$ and 
$$
\ET=    \{(v,(v',w,b)) \mid v'\in E(v),\ w \in E(v')\text{ and } 
 b\in\{\top,\bot\}\} 
 \cup\{((v,w,b),w') \mid w'\in E(v)\}	
$$

Intuitively in a partial play $\lambda$, by choosing an edge from $v$ to $(v',w,b)$ \Eloise indicates that the direction mapping in $\lambda\cdot v'$ is to go to $w$; moreover if $b=\top$ she indicates that $\lambda\cdot v'$ is in the dense set $W$ (remark that, {due to the turn based nature of the game}, one can safely assume that the element in $W$ are always partial plays ending in a vertex in $\VN$).
{A play is winning for \Eloise if either it satisfies the winning condition while visiting infinitely many nodes marked as belonging to the dense set or if at some point no more position in $W$ is reached while \Abelard infinitely often selects a direction that is not the one given by the direction mapping (\ie he does not let \Eloise a chance to get to a position in $W$).}
Formally, the winning condition $\WCT$ is defined by
\begin{multline*}
\WCT = \{v_0(v_0',w_0,b_0)v_1(v_1',w_1,b_1)v_2\cdots\mid v_0v'_0v_1v'_1v_2v'_2\cdots \in \WC 
   \textrm{ and } \exists^{\infty} j \text{ s.t. } b_j=\top \}\\
 \cup  \{v_0(v_0',w_0,b_0)v_1(v_1',w_1,b_1)v_2\cdots\mid \exists^{<\infty} j \text{ s.t. }b_j=\top 
\text{ and }\exists^{\infty} j \text{ s.t. }v_{j+1}\neq w_j)\}.	
\end{multline*}

The following result connects the  games $\game$ and $\gameT$.

\begin{theorem}\label{thm:topo-perfect}
\Eloise has a topologically-good strategy in $\game$ if and only if she has a winning strategy in $\gameT$.

More precisely, from a winning strategy (\resp positional strategy, \resp finite-memory strategy) $\phiT_{\Ei}$ of \Eloise in $\gameT$, we can define a topologically-good strategy (\resp positional strategy, \resp finite-memory strategy) $\phi_{\Ei}$ for \Eloise in $\game$.
\end{theorem}

\begin{proof}
Assume that \Eloise has a topologically-good strategy in $\game$. Call $\phi$ this strategy and let $t_\phi$ be the set of all partial plays starting from $v_0$ where \Eloise respects $\phi$. By definition $t_\phi$ is a tree and its branches are those plays in $\game$ where \Eloise respects $\phi$. As $\phi$ is topologically-good the set of branches in $t_\phi$ that belongs to $\Omega$ is large and therefore thanks to Lemma~\ref{lemma:largeVSdense} it 
contains a dense set of nodes {$W$} that, using  Lemma~\ref{directions}, can be described by a direction mapping $d$.

{Define a strategy $\phiT$ in $\gameT$ for \Eloise by letting $\phiT(v_0(v_0',w_0,b_0)v_1(v_1',w_1,b_1)\dots v_k)=(v_{k}',w_{k},b_{k})$ where $v_{k}' = \phi(v_0v_0'v_1v_1'\dots v_k)$, $w_{k} = d(v_0v_0v_1v_1'\dots v_kv_k')$ and $b_{k} = \top$ if $v_0v_0'v_1v_1'\dots v_kv_k'\in W$ and $b_{k} = \bot$ otherwise. }

Now consider a play {$\playT=v_0(v_0',w_0,b_0)v_1(v_1',w_1,b_1)\dots$} in $\gameT$ where \Eloise respects $\phiT$: if it goes infinitely often through vertices in $\VN\times\VE\times\{\top\}$ then $v_0v_0'v_1v_1'\dots$ is an infinite branch in $t_\phi$ that goes through infinitely many nodes in $W$ hence, belongs to $\WC$ and so $\playT\in\WCT$; otherwise, thanks to the direction mapping and the definition of $\phiT$ it follows that if eventually \Abelard always chooses to go from $(v',w,b)$ to $w$ then one eventually reaches a vertex in $\VN\times\VE\times\{\top\}$ and therefore $\playT\in\WCT$.

Conversely, assume that \Eloise has a winning strategy $\phiT$ in $\gameT$.
We define a strategy $\phi$ for \Eloise in $\game$ as follows. The strategy $\phi$ is defined so that with a partial play $\play$ in $\game$ (where she respects $\phi$) is associated a partial play $\playT$ in $\gameT$ (where she respects $\phiT$). Initially $\play=\playT=v_0$. Let $\play=v_0v'_0v_1v_1'\cdots v_k$ be a partial play where she respects $\phi$ and let $\playT=v_0(v'_0,w_0,b_0)v_1(v_1',w_1,b_1)\cdots v_k$; then call $\phiT(\playT) = (v'_k,w'_k,b_k)$;  define $\phi(\play) = v'_k$ and $\widetilde{\play v'_k} = \playT  (v'_k,w'_k,b_k)$. Now let $t_\phi$ be the set of all partial plays starting from $v_0$ where \Eloise respects $\phi$. Define the set of nodes $W$ in $t_\phi$ as those partial plays that ends in $\VN$ and such that $\playT$ ends in a vertex in $\VN\times\VE\times\{\top\}$ and define a direction mapping $d$ in $t_\phi$ by letting, for any $\lambda$ ending in $\VN$, $d(\lambda) = w$ where $w$ is such that $\playT$ ends in a vertex in $\VN\times\{w\}\times\{\bot,\top\}$ (in other nodes there is a single son so there is only one way to define $d$). As $\phiT$ is winning one easily deduces that $d$ is a direction mapping that points to $W$ and that $\mathcal{B}(W)\subseteq \Omega$. Therefore, the subset of branches of $t_\phi$ that satisfies $\Omega$ is large, meaning that $\phi$ is topologically-good.
\end{proof}

\section{Perfect-Information Games with Nature: Some Consequences}
\label{section:consequences-perfect}

We now discuss several consequences of our results in the perfect-information setting considered so far.

\subsection{The Special Case of Parity Games Played on Finite Arenas}

In the following we argue that concepts of almost-surely winning strategies and of topologically-good strategies coincide in the special case where the arenas are \emph{finite} and where one considers an $\omega$-regular winning condition.

First recall that in the setting of stochastic games played on finite graphs and equipped with an $\omega$-regular winning condition, it is well-known (see \eg \cite{ChatterjeePHD}) that finite-memory strategies suffices for both players. Formally \Eloise has an almost-surely strategy if and only if she has a finite-memory strategy $\strat_\Ei$ such that for every finite-memory strategy $\strat_\Ai$ of \Abelard one has $\mesure{v_0}{\strat_\Ei}{\strat_\Ai}(\WC)=1$.

A similar property actually holds for the topological setting.

\begin{lemma}
Let $\game=(\arena,v_0,\WC)$ be a game with an $\omega$-regular winning condition played on a finite arena. Then \Eloise has a topologically-good strategy if and only if she has a finite-memory strategy $\strat_\Ei$ such that for every finite-memory strategy $\strat_\Ai$ of \Abelard the set  $\outcomes{v_0}{\phi_\Ei}{\phi_\Ai}\setminus \WC$ of losing plays for \Eloise is meager in the tree $\treeGame{v_0}{\phi_\Ei}{\phi_\Ai}$.
\end{lemma}

\begin{proof}
	In the case where \Abelard is not part of the game it is a direct consequence of Theorem~\ref{thm:topo-perfect} together with the fact that two-player games on finite graphs with an $\omega$-regular enjoy finite-memory strategies. 
	In the general setting, it is a consequence of Corollary~\ref{cor:Parys}. Indeed, topologically-good strategies are MSO definable and therefore, when the graph is finite, they can be chosen to be regular (\ie implemented by a finite transducer) hence, be finite-memory. Now it remains to prove that one can without loss of generality restrict \Abelard's strategies to be finite-memory. For that, we use the same argument. Fix a \emph{finite-memory} strategy $\strat_\Ei$ of \Eloise: the set of \Abelard's strategies $\strat_\Ai$ such that the set  $\outcomes{v_0}{\phi_\Ei}{\phi_\Ai}\setminus \WC$ of losing plays for \Eloise is large in the tree $\treeGame{v_0}{\phi_\Ei}{\phi_\Ai}$ is MSO definable in a synchronised product of the arena together with a transducer implementing strategy $\strat_\Ei$, \ie it is MSO definable on a fixed finite graph. Hence, if this set is non-empty it contains a finite-memory strategy.
\end{proof}

The following relates the stochastic and the topological settings in the special case where the arenas are \emph{finite} and where one considers an $\omega$-regular winning condition. 
Note that here we make no assumption on the probability distribution put on the transitions.

\begin{theorem}\label{thm:topo-finite-regular}
Let $\game=(\arena,v_0,\WC)$ be a game with an $\omega$-regular winning condition played on a finite arena. Then \Eloise almost-surely wins if and only she wins in the topological sense.
\end{theorem}

\begin{proof}
First recall that, as pointed in Section~\ref{subsection:topological}, topological and probabilistic largeness coincide for $\omega$-regular properties of regular trees.
Moreover as established above one can safely (in both setting) restrict to finite-memory strategies (for both \Eloise and \Abelard). 

Assume that \Eloise has a finite-memory almost-surely winning strategy $\strat_\Ei$. We claim that it is topologically-good. Indeed, consider a finite-memory strategy $\strat_\Ai$ of \Abelard. As the arena is finite and as both strategies have finite-memory, the tree $\treeGame{v_0}{\phi_\Ei}{\phi_\Ai}$ is regular and the set $\outcomes{v_0}{\phi_\Ei}{\phi_\Ai}\setminus \WC$ of losing plays has measure $0$ hence is meager. 

Conversely assume that \Eloise has a finite-memory topologically-good  strategy $\strat_\Ei$ and let us prove that it is almost-surely winning. Indeed, consider a finite-memory strategy $\strat_\Ai$ of \Abelard. As the arena is finite and as both strategies have finite-memory, the tree $\treeGame{v_0}{\phi_\Ei}{\phi_\Ai}$ is regular and the set $\outcomes{v_0}{\phi_\Ei}{\phi_\Ai}\setminus \WC$ of losing plays is meager hence has measure $0$. 
\end{proof}

We believe that Theorem~\ref{thm:topo-finite-regular} is an important result because it essentially means that in most of the situations (namely when restricting to both $\omega$-regular winning conditions and finite arenas) for which one can decide the existence of almost-surely winning strategies, then the concept is the same as being topologically-good. Moreover, as one can decide existence of topologically-good strategies for largest classes of games (as explained below in Section~\ref{subsection:Games-on-Infinite-Arenas-perfect}) it strengthen our belief that topologically-good strategies are a very valuable notion.

We now explain why Theorem~\ref{thm:topo-finite-regular} generalises previous work from \cite{VolzerV12} and \cite{AsarinCV10} (also see \cite{BrihayeHM15} for related questions). In \cite{VolzerV12} —~rephrased in our setting~— Varacca and Völzer considered (among many other things) games where \Nature plays alone and whose winning condition is $\omega$-regular and in particular they showed that if the arena is finite then the set of outcomes (\ie the set of all plays as \Nature plays alone) is large if and only if it has probability $1$, \ie topological and probabilistic largeness coincide for $\omega$-regular properties of finite Markov chains.

 A natural question, addressed by Asarin \emph{et al.} in~\cite{AsarinCV10}, is whether this is still true for Markov decision processes {(\ie a game with Eloise and Nature in the probabilistic setting)}. For this they introduced a notion of three player games\footnote{We change here the name of the players to stick to the presentation of this paper and use EBM-game instead of the original name, ABM-game.} (EBM-games) where \Eloise plays against \Abelard who is split into two sub-players —~Banach who is good and Mazur who is evil. Banach starts playing for \Abelard and after some time he decides to let Mazur play for a while and then Mazur let him play again and so on. \Eloise \emph{does not observe} who —~Banach or Mazur~— is acting for \Abelard. Say that \Eloise wins the game if she has a strategy such that Banach also has a strategy such that whatever Mazur does the winning condition is satisfied. The main result of \cite{AsarinCV10} is that for an EBM-game on a finite arena with an $\omega$-regular objective \Eloise has a winning strategy iff she has an almost-surely winning strategy in the \Eloise-\Nature game obtained by seing the “Banach/Mazur” player as the single stochastic player \Nature (for arbitrary probability distributions). 

This result is a corollary of Theorem~\ref{thm:topo-finite-regular} as it is easily seen that in the \Eloise-\Nature game obtained by merging the “Banach/Mazur” players as the single player \Nature, \Eloise has a topologically-good strategy if and only if \Eloise wins the EBM-game. Indeed, she has a topologically-good strategy if and only if she has a strategy so that in the induced Banach-Mazur game she has a strategy that wins against any strategy of \Abelard: hence, it suffices to see \Eloise in the Banach-Mazur game as Banach and \Abelard as Mazur. 

Remark that our approach differs from \cite{AsarinCV10} by the fact that we reason by reduction instead of providing an ad-hoc algorithm; moreover topologically-good strategies make sense also for two-player games with \Nature while EBM-games do not extend naturally to capture a second antagonistic player.

\subsection{Variant of Tree Automata}\label{section:tree-automata}

We now discuss consequences of our results in the cardinality setting for classes of automata on infinite trees.

A \defin{parity tree automaton} $\mathcal{A}$ is a tuple $\langle A, Q , q_{ini},\Delta,\col \rangle$ where $A$ is a finite {input alphabet}, $Q$ is a finite {set of states}, $\qini\in Q$ is the {initial state}, $\Delta \subseteq Q\times A \times Q \times Q$ is a {transition relation} and $\col:Q\rightarrow \colors$ is a colouring function. 

Given an $A$-labelled complete binary tree $t$, a \defin{run} of $\mathcal{A}$ over $t$ is a $Q$-labelled complete binary tree $\run$ such that 
\begin{enumerate}[(\itshape i\upshape)]
\item
 the root is labelled by the initial state, \emph{i.e.} $\run(\epsilon)=\qini $;
\item
 for every node $u\in\{0,1\}^*$, $(\run(u),t(u),\run(u\cdot0),\run(u\cdot1))\in\Delta$.
\end{enumerate}
A branch $\pi=\pi_0\pi_1\pi_2\cdots$ is \defin{accepting} in the run $\run$ if its labels satisfies the parity condition, \ie $\liminf (\col(\run(\pi_1\cdots \pi_i)))_{i\geq 0}$ is even; otherwise it is rejecting.

Classically, one declares that a tree $t$ is accepted by $\mathcal{A}$ if there exists a run of $\mathcal{A}$ on $t$ such that \emph{all} branches in it are accepting. One denotes by $L(\mathcal{A})$ the set of accepted trees and such a language is called \defin{regular}.

Several relaxations of this criterion have been investigated in~\cite{BN95,CHS14a,CarayolS17}. 
\begin{itemize}
\item \defin{Automata with cardinality constraints}. Among others one can consider the language $\LAccUnc{\mathcal{A}}$ of those trees for which there is a run with at least uncountably many accepting branches~\cite{BN95}, and the language $\LRejAtMostCount{\mathcal{A}}$ of those trees for which there is a run with at most countably many rejecting branches~\cite{CarayolS17}.
\item \defin{Automata with topological bigness constraints}: a tree belongs to $\LLarge{\mathcal{A}}$ if and only if there is a run whose set of accepting branches is large~\cite{CarayolS17}.
\item\defin{Qualitative tree automata}: a tree belongs to $\Lqual{\mathcal{A}}$ if and only if there is a run whose set of accepting branches has measure $1$~\cite{CHS14a}.
\end{itemize}

{Our results implies the following theorem~\cite{BN95,CarayolS17}, where by \emph{effectively regular} we mean that the language is regular and that one can effectively construct an accepting automaton (in the statement below, starting from $\mathcal{A}$).}

\begin{theorem}
For any parity tree automaton $\mathcal{A}$, $\LAccUnc{\mathcal{A}}$, $\LRejAtMostCount{\mathcal{A}}$ are effectively regular.
\end{theorem}

\begin{proof}
Start with the case  $\LRejAtMostCount{\mathcal{A}}$. One can think of the acceptance of a tree $t$ as a game $\game$ where \Eloise labels the input by transitions and \Nature chooses which branch to follow: $t\in \LRejAtMostCount{\mathcal{A}}$ iff the leaking value of this game is at most $\aleph_0$. Consider game $\gameL$ as in Theorem~\ref{theo:perfect:main-aleph}. This game (up to some small changes) is essentially the following: the play starts at the root of the tree; in a node $u$ \Eloise chooses a valid transition  of the automaton and indicates a direction she wants to avoid and then \Abelard chooses the next son; the winning condition is that either the parity condition is satisfied or finitely often \Abelard obeys \Eloise. It is then easy to see this latter game as the “usual” acceptance game for {some tree automaton with an $\omega$-regular acceptance condition}.

Now consider the case $\LAccUnc{\mathcal{A}}$. One can think of the acceptance of a tree $t$ as a game $\game$ where \Eloise does nothing, \Abelard labels the input by transitions and \Nature chooses which branch to follows; the winning condition is the complement of the parity condition: $t\in \LAccUnc{\mathcal{A}}$ iff the leaking value of this game is $2^{\aleph_0}$. Again, one can consider game $\gameL$ as in Theorem~\ref{theo:perfect:main-aleph} in which we know that \Abelard has a winning strategy. Then switch the names of the players, complement the winning condition and obtain an acceptance game for $\LAccUnc{\mathcal{A}}$ where in a node $u$ \Eloise chooses a valid transition  of the automaton, then \Abelard indicates a direction he wants to avoid and then \Eloise chooses the next son; the winning condition is that the parity condition is satisfied and infinitely often \Eloise obeys \Abelard. Then one can easily prove that this game is equivalent to the following game: in a node $u$ \Eloise chooses a valid transition  of the automaton and may indicate a direction to follow, then \Abelard chooses the next son (and if \Eloise indicated a direction to follow he must respect it); the winning condition is that the parity condition is satisfied and infinitely often \Eloise does not indicate a direction. This latter game can easily be seen as the “usual” acceptance game for {some tree automaton with an $\omega$-regular acceptance condition}.
\end{proof}

\begin{remark} One can wonder whether a similar statement can be obtained for the languages $\LLarge{\mathcal{A}}$. In \cite{CarayolS17} such languages are indeed shown to be effectively regular. One could used Theorem~\ref{thm:topo-perfect} to derive an alternative proof but we omit it here as the construction is far less elegant than for automata with cardinality constraints (and therefore the gain compared with the direct approach in \cite{CarayolS17} is unclear).
\end{remark}

\subsection{Games Played on Infinite Arenas}\label{subsection:Games-on-Infinite-Arenas-perfect}

We claim that, in many contexts where the probabilistic approach leads to undecidability, the two approaches (cardinality and topological) that we proposed permit to obtain \emph{positive} results for the main problem usually addressed: decide if \Eloise has a “good” strategy and if so compute it. 

As this is not the core topic of the present paper we only briefly mention some of these contexts and, for each of them, point out the undecidability result in the probabilistic setting and the decidability result in the two-player game (without nature) setting that combined with our main results (Theorem~\ref{theo:perfect:main-aleph}~/~Theorem~\ref{thm:topo-perfect}) leads to decidability in the cardinality/topological setting.
\begin{itemize}
\item Games played on pushdown graphs. These are games played on infinite graphs that can be presented as the transition graph of a pushdown automaton, as the one we considered in Example~\ref{example:game-running1}. They are of special interest because in particular they permit to capture programs with recursion and they are also the very first class of two-player games on infinite graphs that where shown to be decidable~\cite{Walukiewicz01}. But when moving to the probabilistic setting, and already for \Eloise-\Nature reachability games, they were shown (except under a quite strong restriction) to lead undecidability~\cite{EtessamiY05}. In contrast, \Eloise-\Abelard-\Nature (\resp \Eloise-\Nature) parity games are decidable in the cardinality (\resp topological) setting as a consequence of the decidability for the \Eloise-\Abelard setting from~\cite{Walukiewicz01}. 
\item Games played on higher-order and collapsible pushdown graphs. Handling higher-order recursion, a programming paradigm that has been widely adopted in the last decade,
as all mainstream languages have added support for
higher-order procedures\footnote{For example, they were the major novelty in Java 8, they are central to Scala, and  they are also at the core of JQuery, the most popular JavaScript library widely used in client-side web programming.}, is a crucial question in program verification (see \eg \cite{Kobayashi13} for a survey on that topic). One possible approach consists in finding an automata model capturing the behaviours of such programs (as pushdown automata do for order-1 recursion), and collapsible pushdown automata (as well as higher-order pushdown automata for a restricted class of program) form such a class \cite{HMOS08,CS12,HagueMOS17}. As \Eloise-\Abelard parity games played on transition graphs of collapsible pushdown automata are decidable \cite{HMOS08}, it turns out that \Eloise-\Abelard-\Nature (\resp \Eloise-\Nature) parity games played on collapsible pushdown graphs are decidable in the cardinality (\resp topological) setting one. We believe this is an interesting starting point to study decidability of verification problems for programs with both higher-order recursion and uncontrollable and unpredictable behaviours. As an example, think of a jQuery program relying on a call to an external web service to complete a task: higher-order comes from using a call-back function to treat the answer of the web service while unpredictability comes from the fact that the web service may time out.
\item A popular non regular winning condition in pushdown game is the boundedness/unbounded\-ness condition that imposes a restriction on how the stack height evolves during a play. For stochastic games with \Nature only (\ie probabilistic pushdown automata) there are some positive results \cite{EsparzaKM05} but they break (because of \cite{EtessamiY05}) whenever \Eloise comes in. {In the cardinality (\resp topological) setting we have decidability in the general case of \Eloise-\Abelard-\Nature (\resp \Eloise-\Nature) thanks to Theorem~\ref{theo:perfect:main-aleph} (\resp Theorem~\ref{thm:topo-perfect}) combined with the results in \cite{BouquetSW03,Gimbert04}.}
\end{itemize}

\section{Imperfect-Information Games with Nature}
\label{section:imperfect}
We now move to a richer setting where \Eloise has imperfect-information. The vertices of the game are partitioned by an equivalence relation and \Eloise does not observe exactly the current vertex but only its equivalence class. In full generality, \Abelard should also have imperfect-information but {we are not able to handle this general case and therefore we assume here that he is perfectly informed}. Of course, as \Eloise has imperfect-information we have to slightly change the definition of the game (she now plays actions) and to restrict the strategies she can use. We also change how \Nature interacts with the players, but one can easily check that this setting captures the one we gave in the perfect-information case. 

One could wonder why we did not directly treat the imperfect-information case. There are two main reasons for that. Firstly, in the imperfect-information setting we only have results for the parity condition and not for any Borel condition. Secondly, the proof of Theorem~\ref{theo:main-partial} crucially uses the results obtained in the perfect-information setting.

\subsection{Definitions}

An \defin{imperfect-information arena} is a tuple $\arena=(\VE,\VA,\Gamma,\Delta_\Ei,\Delta_\Ai,\sim)$ where $\VE$ is a {countable} set of \Eloise's vertices, $\VA$ is a countable set of \Abelard's vertices (we let $V=\VE\uplus\VA$), $\Gamma$ is a {possibly uncountable} set of \Eloise's actions, $\Delta_\Ei:\VE\times \Gamma\rightarrow 2^V$ is \Eloise's transition function and $\Delta_\Ai:\VA\rightarrow 2^V$ is \Abelard's transition function and $\sim$ is an equivalence relation on $V$. We additionally require that the image by $\Delta_\Ei$ (\resp $\Delta_\Ai$) is never the empty set.
We also require that there is no two vertices $v_1\in \VE$ and $v_2 \in \VA$ such that $v_1\sim v_2$ (\ie the $\sim$ relation always distinguishes between vertices owned by different players).

As in the perfect-information setting, a play involves two antagonistic players —~\Eloise and \Abelard~— together with an unpredictable and uncontrollable player called \Nature. It starts in some initial vertex $v_0$ and when in some vertex $v$ the following happens:
\begin{itemize}
\item if $v\in\VE$, \Eloise chooses an action $\gamma$ and then \Nature chooses the next vertex among those $v'\in\Delta_\Ei(v,\gamma)$;
\item if $v\in\VA$, \Abelard chooses the next vertex $v'\in\Delta_\Ai(v)$.
\end{itemize}
Then, the play goes on from $v'$ and so on forever. 

Hence, a \defin{play} can be seen as an element in $(\VE\cdot \Gamma\cup \VA)^\omega$ compatible with $\Delta_{\Ei}$ and $\Delta_{\Ai}$. More formally, it is a sequence $\play=x_0x_1x_2\cdots$ such that for all $i\geq 0$ if $x_ix_{i+1}\in \VE\cdot \Gamma$ then one has $x_{i+2}\in\Delta_\Ei(x_i,x_{i+1})$; and if $x_i\in \VA$ then one has $x_{i+1}\in\Delta_\Ai(x_i)$. A \defin{partial play} is a prefix of a play that belongs to $(\VE\cdot \Gamma\cup \VA)^*$.

Two $\sim$-equivalent vertices are supposed to be indistinguishable by \Eloise and we extend $\sim$ as an equivalence relation on $V^*$: $v_0\dots v_h\sim v'_0\dots v'_k$ if and only if $h=k$ and $v_i\sim v'_i$ for all $0\leq i\leq k$; we denote by $[\lambda]_{/_\sim}$ the equivalence class of $\lambda\in V^*$.
 An \defin{observation-based strategy} for \Eloise is a map $\strat:(V^*{\VE})_{/_\sim} \rightarrow \Gamma$. We say that \Eloise~\defin{respects} $\phi$ during a  play $\lambda=v_0\gamma_0v_1\gamma_1v_2\gamma_2\cdots$ (where $\gamma_i$ is the empty word when $v_i\in \VA$ and an action in $\Gamma$ when $v_i\in\VE$) if and only if $\gamma_{i+1}=\phi([v_0\cdots v_{i}]_{/_\sim)}$ for all $i\geq 0$ such that $v_i\in\VE$. 

\begin{remark}
One may expect a strategy for \Eloise to also depend on the actions she has played so far, \ie to be a map $\phi_\Evei:((\VE\cdot \Gamma\cup \VA)^*\cdot \VE)_{/_\sim}\rightarrow \Gamma$ where $\sim$ is extended on $\Gamma$ by letting $\gamma\sim x$ iff $\gamma=x$ when $\gamma\in\Gamma$. But such a strategy can be mimicked by a strategy (in our sense) $\phi_\Evei':V^*\rightarrow \Gamma$ by letting $\phi_\Evei'([v_0\cdots v_k]_{/_\sim}) = \phi_\Evei([v_0\gamma_0\cdots \gamma_{k-1}v_k]_{/_\sim})$ with $\gamma_i = \phi_\Evei([v_0\gamma_0\cdots \gamma_{i-1}v_i]_{/_\sim})$ when $v_i\in\VE$ and $\gamma_i=\epsilon$ otherwise. Note that requiring to be observation-based does not interfere with the previous trick.
\end{remark}

A \defin{strategy} for \Abelard is a map $\strat:(\VE\cdot \Gamma\cup \VA)^*\cdot \VA \rightarrow V$. We say that \Abelard~\defin{respects} $\phi$ in the play $\lambda=v_0\gamma_0v_1\gamma_1v_2\gamma_2\cdots$ (again, $\gamma_i$ is the empty word when $v_i\in \VA$ and an action in $\Gamma$ when $v_i\in\VE$) if and only if $v_{i+1}=\phi(v_0\gamma_0v_1\cdots v_{i})$ for all $i\geq 0$ such that $v_i\in\VA$.

With an initial vertex $v_0$, a strategy $\phi_\Ei$ of \Eloise and a strategy $\phi_\Ai$ of \Abelard, we associate the set $\outcomes{v_0}{\phi_\Ei}{\phi_\Ai}$ of all possible plays starting from $v_0$ and where \Eloise (\resp \Abelard) respects $\phi_\Ei$ (\resp $\phi_\Ai$).

In this part, we only have positive results for parity winning conditions, hence we focus on this setting (but generalising the various notions to any Borel winning condition is straightforward).
A parity winning condition is defined thanks to a colouring function $\col:V\rightarrow \colors$ with a \emph{finite} set of colours $\colors\subset\mathbb{N}$. {We require that colouring function stays constant on the equivalence classes of the relation $\sim$ (\ie $\col(v)=\col(v')$ for all $v \sim v'$.)}.

Again, a play $\lambda=v_0\gamma_0v_1\gamma_1v_2\gamma_2\cdots$ (where $\gamma_i=\epsilon$ when $v_i\in\VA$) satisfies the parity condition if $\liminf(\col(v_i))_{i\geq 0}$ is even; we denote by $\WC_\col$ the set of plays satisfying the parity condition defined by the colouring function $\col$.

A \defin{imperfect-information parity game with nature} is a tuple $\game=(\arena,\col,v_0)$ consisting of an imperfect-information arena $\arena$, a colouring function $\col$ and an initial vertex $v_0$.

\begin{remark}
 A more symmetric notion of imperfect-information game would let
\Abelard also play actions (\ie $\Delta_{\Ai} : \VA \times \Gamma \mapsto 2^{V}$) while Nature would choose the successor
as it does for \Eloise. Consider such a game $\arena=(\VE,\VA,\Gamma,\Delta_\Ei,\Delta_\Ai,\sim)$ where \Abelard
plays actions. We can simulate it by a game $\arena=(\VE',\VA,\Gamma,\Delta_\Ei',\Delta_\Ai',\sim')$ in our setting.
 For every vertex $v$ of \Abelard and every action 
$\gamma \in \Gamma$, we introduce a new vertex $(v,\gamma)$ for \Eloise (\ie $\VE'= \VE \cup \VA \times \Gamma$). Furthermore we set $\Delta_{\Ai}'(v) = \{ (v,\gamma) \mid \gamma \in \Gamma \}$ and for everty vertex of \Eloise of the form $(v,\gamma)$, we take $\Delta_{\Ei}'((v,\gamma),\gamma') = \Delta_{\Ai}(v,\gamma)$ for all action $\gamma' \in \Gamma$. 
For the original vertices $v \in \VE$ and for $\gamma \in \Gamma$, we take $\Delta_{\Ei}'(v,\gamma)= \Delta_{\Ei}(v,\gamma)$. Finally the equivalence {relation} $\sim'$ coincides with $\sim$, and equates all new vertices. It is then easy to check that both games are equivalent.
\end{remark}

In order to evaluate how good an \Eloise's strategy is, we can take exactly the same definitions and notations as we did in the perfect-information setting (this is why we do not repeat them here). Hence, we have the notions of \defin{cardinality leaking} of a strategy (thanks to Definition~\ref{def:cardleaking}),  \defin{leaking value} of a game (thanks to Definition~\ref{def:cardleakGame}), and \defin{topologically-good} strategy (thanks to Definition~\ref{def:topoGood}).

\begin{remark}\label{remark:supNotMaxCardleakingImperfect}
For the same reason as in the prefect-information setting we have that for any strategy $\phi_\Ei$  one has $\CL{\phi}\in\mathbb{N}\cup\{\aleph_0,2^{\aleph_0}\}$ and as a consequence that $\LVal{\game}\in\mathbb{N}\cup\{\aleph_0,2^{\aleph_0}\}$.	
\end{remark}

{
\begin{example}\label{ex:ex-imperfect}
	Consider the Büchi game where \Eloise and \Abelard choose simultaneously and independently a bit in $\{0,1\}$: if the bits are the same the game goes to a special vertex coloured by $0$, otherwise goes to a special vertex coloured by $1$, and then, in both cases, another round starts and so on forever. Hence, \Eloise wins if she infinitely often guesses correctly choice of \Abelard. To simulate the concurrent aspect of the choices of the player we will use imperfect information: \Abelard chooses first but \Eloise does not observe his choice, and she chooses second. Moreover, for her choice, \Eloise has a third option which is to let \Nature choose for her: technically, once \Abelard made his choice, \Nature is also making a choice (hidden to \Eloise) and then \Eloise has three options: choose bit $0$, choose bit $1$ or pick the bit chosen by \Nature.
	
	Formally (see Figure~\ref{fig:ex-imperfect} for an illustration) one defines the imperfect-information arena $\arena=(\VE,\VA,\Gamma,\Delta_\Ei,\Delta_\Ai,\sim)$ where $\VE=\{l,w,v_0,v_1,v_{0,0},v_{0,1},v_{1,0},v_{1,0}\}$, $\VA=\{v\}$, $\Gamma=\{\sharp,0,1,N\}$, $\Delta_\Ai$ and $\Delta_\Ei$ (we omit meaningless actions but could add a dummy state to handle them) are given by
	\begin{itemize}
		\item $\Delta_\Ai(v)=\{v_0,v_1\}$: \ie \Abelard encodes the choice of his bit by going either to $v_0$ or $v_1$;
		\item $\Delta_\Ei(v_0,\sharp)=\{v_{0,0},v_{0,1}\}$ and $\Delta_\Ei(v_1,\sharp)=\{v_{1,0},v_{1,1}\}$: this corresponds to the step where \Nature is choosing its bit;
		\item $\Delta_\Ei(v_{0,0},0)=\Delta_\Ei(v_{0,0},N)=\Delta_\Ei(v_{0,1},0)=\Delta_\Ei(v_{1,0},1)=\Delta_\Ei(v_{1,1},1)=\Delta_\Ei(v_{1,1},N)=\{f\}$: this corresponds to either \Eloise choosing the same bit as \Abelard or mimicking luckily the choice of \Nature;
		\item $\Delta_\Ei(v_{0,0},1)=\Delta_\Ei(v_{0,1},1)=\Delta_\Ei(v_{0,1},1)=\Delta_\Ei(v_{1,0},0)=\Delta_\Ei(v_{1,0},N)=\Delta_\Ei(v_{1,1},0)=\{l\}$: this corresponds to either \Eloise choosing a bit different from \Abelard or mimicking unluckily the choice of \Nature;
		\item $\Delta_\Ei(f,\sharp)=\Delta_\Ei(l,\sharp)=\{v\}$: this corresponds to start a new round.
	\end{itemize}
	and, $v_0\sim v_1$ and $v_{0,0}\sim v_{0,1}\sim v_{1,0}\sim v_{1,1}$. The colouring function $\col$ equals $1$ everywhere except on $f$ where it equals $0$. Finally we let $\game=(\arena,\col,v)$.

	Consider an \Eloise's strategy  $\strat_\Ei$ that finitely often plays action $N$. Then, it is easily seen that $\CL{\strat}=2^{\aleph_0}$. Indeed, consider the strategy $\strat_\Ai$ of \Abelard that chooses the bit opposite to that prescribed by $\strat_\Ei$ (and any bit when $\strat_\Ei$ plays action $N$): then there are no winning play for \Eloise in $\outcomes{v}{\strat_\Ei}{\strat_\Ai}$, hence $\CL{\strat}=2^{\aleph_0}$.
	
	Now, consider the strategy $\psi_\Ei$ of \Eloise that always plays action $N$ from vertices in $\{v_{0,0},v_{0,1},v_{1,0},v_{1,1}\}$. Then one has $\CL{\strat}=\aleph_0$. Indeed, consider any strategy $\psi_\Ai$ of \Abelard. Then, for every $k\geq 0$ there are only finitely many plays in $\outcomes{v}{\psi_\Ei}{\psi_\Ai}$ that never visits $f$ after the $n$-th round (namely the ones where \Nature only makes incorrect choices after the $k$-th round). As the set of loosing plays is the countable union of the previous plays when $k$ ranges over $\mathbb{N}$, there are only countably many loosing plays in $\outcomes{v}{\strat_\Ei}{\strat_\Ai}$, hence $\CL{\strat}={\aleph_0}$. As the set of loosing plays is clearly a countable union of nowhere dense sets, it follows that $\psi_\Ei$ is also topologically good.
\end{example}
}

\begin{figure}
\begin{center}
\begin{tikzpicture}[>=stealth',thick,scale=1,transform shape]
\tikzstyle{Abelard}=[draw]
\tikzstyle{Eloise}=[draw,circle]
\tikzstyle{Nature}=[draw,diamond,scale = .65,font=\Large]
\tikzstyle{Buchi}=[fill=MediumSpringGreen]
\tikzset{every loop/.style={min distance=15mm,looseness=10}}
\tikzstyle{loopleft}=[in=150,out=210]
\tikzstyle{loopright}=[in=-30,out=30]
\tikzstyle{loopbelow}=[in=-120,out=-60]
\tikzstyle{loopabove}=[in=120,out=60]
\node[Eloise] (l) at (0,0) {$l$};
\node[Abelard,minimum height=.7cm] (v) at (1.5,0) {$v$};
\node[Eloise] (v0) at (3,1) {$v_{0}$};
\node[Eloise] (v1) at (3,-1) {$v_{1}$};
\node[Eloise] (v00) at (5.5,2.25) {$v_{0,0}$};
\node[Eloise] (v01) at (5.5,.75) {$v_{0,1}$};
\node[Eloise] (v10) at (5.5,-.75) {$v_{1,0}$};
\node[Eloise] (v11) at (5.5,-2.25) {$v_{1,1}$};
\node[Eloise,Buchi] (f) at (8,0) {$f$};

\draw[->] (l) -- (v) node[midway,above] {$\sharp$};
\path[->] (v) edge (v0);\path[->] (v) edge (v1);
\draw (v0) -- (4,1) node[midway,above] {$\sharp$};
\draw[->] (4,1) -- (v00);\draw[->] (4,1) -- (v01);
\draw (v1) -- (4,-1) node[midway,above] {$\sharp$};
\draw[->] (4,-1) -- (v10);\draw[->] (4,-1) -- (v11);

\path[->,bend left] (v00) edge node[right] {$0,N$} (f);
\path[->] (v01) edge node[above] {$0$} (f);
\path[->] (v10) edge node[below] {$1$} (f);
\path[->,bend right] (v11) edge node[right] {$1,N$} (f);

\path[->,bend right=40] (v00) edge node[above] {$1$} (l);
\path[->,bend right=40] (v01) edge node[above] {$1,N$} (l);
\path[->,bend left=40] (v10) edge node[below] {$0,N$} (l);
\path[->,bend left=40] (v11) edge node[below] {$0$} (l);

\draw[->] (f) -- (v) node[midway,above] {$\sharp$};
 
\end{tikzpicture}
\end{center}
\caption{Arena of Example~\ref{ex:ex-imperfect}: $\sim$-equivalent vertices are depicted in the same column, actions are written on edges, and $f$ is the only vertex coloured by $0$.}\label{fig:ex-imperfect}
\end{figure}
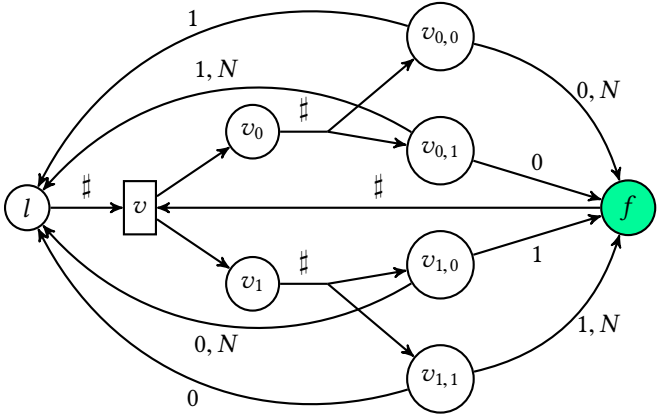

\subsection{Imperfect-Information Two-Player Games}

We now introduce another version of games with imperfect-information where there are only two antagonist players — \Eloise and \Abelard. The only difference with the previous model with \Nature is that  now the non-determinism induced by a choice of an action of \Eloise is resolved by \Abelard. This concept was first considered in~\cite{CDHR07} for finite arenas.

Let $\arena=(\VE,\VA,\Gamma,\Delta_\Ei,\Delta_\Ai,\sim)$ be an imperfect-information arena. Then a play involves two players \Eloise and \Abelard: it starts in some initial vertex $v_0$ and when in some vertex $v$ the following happens:
\begin{itemize}
\item if $v\in\VE$, \Eloise chooses an action $\gamma$ and then \Abelard chooses the next vertex among those $v'\in\Delta_\Ei(v,\gamma)$;
\item if $v\in\VA$, \Abelard chooses the next vertex $v'\in\Delta_\Ai(v)$.
\end{itemize}
Then, the play goes on from $v'$ and so on forever. Again, a \defin{play} is  an element in $(\VE\cdot \Gamma\cup \VA)^\omega$ and a \defin{partial play} is a finite prefix of a play in $(\VE\cdot \Gamma\cup \VA)^*$.

Observation-based strategies for \Eloise are defined as for imperfect-information games with Nature. We shall later consider winning conditions that are slightly more general than parity conditions hence, we allow any Borel subset $\WC$ of $(\VE\cdot \Gamma\cup \VA)^\omega$. 

An \defin{imperfect-information two-player game} is a tuple $\game=(\arena,\WC,v_0)$ consisting of an arena of imperfect-information, a winning condition $\WC$ and an initial vertex $v_0$. A strategy $\phi_\Ei$ of \Eloise is \defin{winning} in $\game$ if any play starting from $v_0$ where \Eloise respects $\phi_\Ei$ belongs to $\WC$.

\begin{remark}
Note that even for reachability conditions —~\ie when $\WC=V^*FV^\omega$ for some non-empty $F\subseteq V$~— and finite arena, imperfect-information two-player games are not determined, \ie it can happen that none of the two players has a winning strategy. See \cite[Example~2.3]{CDHR07}.
\end{remark}

\subsection{Deciding Whether the Leaking Value Is at Most $\aleph_0$}

Our goal in this section is to design a technique to decide whether \Eloise
has a strategy with a {cardinality leaking of} at most $\aleph_{0}$ in an imperfect-information parity game with nature. 

For the rest of this section we fix an imperfect-information parity game with nature $\game=(\arena,\col,v_0)$ where $\arena=(\VE,\VA,\Gamma,\Delta_\Ei,\Delta_\Ai,\sim)$ and we aim at deciding whether $\LVal{\game}\leq \aleph_0$.

The approach has the same flavour as the one for the perfect-information case: we define an imperfect-information  game without \Nature where \Abelard is now in charge of simulating choices of \Nature while \Eloise will indicate together with her action, a successor that she wants to avoid;  
 moreover \Abelard will be forced (thanks to the winning condition) to respect her choices infinitely often. 

In order to express the choice of \Nature she wants to avoid while 
preserving the fact that she is imperfectly informed about the actual vertex, \Eloise will provide with her action $\gamma\in \Gamma$, a map $\theta:V\rightarrow V$ such that for all $v\in V$ one has $\theta(v)\in \Delta_\Ei(v,\gamma)$; we denote by $\Theta_\gamma$ the set of such maps (for a given $\gamma\in \Gamma$). Intuitively, the meaning of \Eloise playing $(\gamma,\theta)$ is that she plays action $\gamma$ and would prefer, for each vertex $v$, if the play turns to be in $v$, that the next vertex is not $\theta(v)$. 
{
\begin{remark}
The map $\theta$ may be partial: what is important is that, if at some point the play can be in $v$ then $\theta(v)$ should be defined. In particular if there are two bounds, one on the size of the  $\sim$-equivalence classes of $V$ and one on the out-degree of the vertices in $G$, then $\Theta_\gamma$ can be chosen to be finite (up to coding). This will be the case for pushdown games when discussing consequences in Section~\ref{section:consequences-imperfect}.
\end{remark}}

We define a \emph{two-player} imperfect-information arena 
$\arenaNR=(\VE,\VA,\widehat{\Gamma},\widehat{\Delta}_\Ei,\Delta_\Ai,\sim)$
where $\widehat{\Gamma}=\bigcup_{\gamma\in \Gamma}\{\gamma\}\times\Theta_\gamma$ and  $\widehat{\Delta}_\Ei(v,(\gamma,\theta)) = \Delta_\Ei(v,\gamma)$.
A play in $\arenaNR$ is of the form $v_0\nu_{0}v_1\nu_{1}v_2\cdots$ where for all $i \geq 0$, $\nu_{i}$ is equal to some $(\gamma_i,\theta_i)$ if $v_{i} \in \VE$
and is empty otherwise. For some $i \geq 0$, if $v_{i} \in \VE$ and $v_{i+1} \neq 
\theta_{i}(v_{i})$, we say that \Abelard \emph{obeys} \Eloise at this point.
 
We let $\WCNR$ consists of those plays $v_0\nu_{0}v_1\nu_{1}v_2\cdots$ 
 such that either $v_0v_1v_2\cdots\in \WC_\col$ or there are only finitely many $i$ such that  $v_i\in V_\Ei$ and $v_{i+1} \neq \theta_i(v_i)$, \ie either the play satisfies the parity condition or eventually \Abelard never obeys \Eloise. Finally, we denote by $\gameNR$ the two-player imperfect-information game $(\arenaNR,\WCNR,v_0)$. The next result relates $\game$ and $\gameNR$.

\begin{theorem}\label{theo:main-partial}
The leaking value of $\game$ is at most $\aleph_0$ if and only if \Eloise has a winning strategy in $\gameNR$.

More precisely, from a winning strategy (\resp positional winning strategy, \resp finite-memory winning strategy) $\phiNR_{\Ei}$ of \Eloise in $\gameL$, we can define a strategy (\resp positional winning strategy, \resp finite-memory winning strategy) $\phi_{\Ei}$ for \Eloise in $\game$ such that $\CL{\phi_{\Ei}} \leq
\aleph_{0}$. 
\end{theorem}

\begin{proof}
Let $\lambda$ be a partial play in $\game$ (resp. $\gameNR$), we denote by $\encod{\lambda}$ the sequence of vertices in $V^{*}$ obtained by removing the actions from $\lambda$. 

First assume that \Eloise has a winning strategy $\phiNR_\Ei$ in $\gameNR$.
This direction is very similar to the perfect-information case. We define a strategy $\phi_\Ei$ for \Eloise in $\game$ by letting $\phi_\Ei(\play) = \gamma$ whenever $(\gamma,\theta)=\phiNR_\Ei(\play)$. In particular, note that $\phi_\Ei$ uses the same memory as $\phiNR_\Ei$. 
For any partial play $\lambda$ in $\game$ in which \Eloise respects $\phi_{\Ei}$, there exists a unique play, denoted $\hat{\lambda}$, in which \Eloise respects $\phiNR_{\Ei}$ and such that
$\encod{\lambda}=\encod{\hat{\lambda}}$. By taking the limit, we extend this notation 
from partial plays to plays.

Let us now prove that $\CL{\strat_\Ei}\leq\aleph_0$. For this, fix a strategy $\strat_\Ai$ of \Abelard in $\game$ and consider a play $\lambda$ in $\outcomes{v_0}{\phi_\Ei}{\phi_\Ai} \setminus \WC$.

As $\phiNR_\Ei$ is winning in $\gameL$,  $\hat{\lambda}$ (in which \Eloise respects $\phiNR_{\Ei}$) is won by \Eloise. As $\hat{\lambda}$ does not satisfy the parity condition (because $\lambda$ does not), \Eloise wins because \Abelard obeys her
only finitely often. Let $\pi_{\lambda}$ be the longest prefix $\pi$ of $\lambda$ such that $\hat{\pi}$ is of the form $\pi' v (\gamma,\theta) v'$ with $v' \neq \theta(v)$  (\ie it is the last time where \Abelard obeys
\Eloise). By convention, if \Abelard never obeys \Eloise in $\lambda$ we let $\pi_\lambda=\epsilon$.

We claim that $\lambda \in \outcomes{v_0}{\phi_\Ei}{\phi_\Ai} \setminus \WC$
is uniquely characterised by $\pi_{\lambda}$. In particular $\outcomes{v_0}{\phi_\Ei}{\phi_\Ai} \setminus \WC$ is countable as it can be injectively mapped into the countable set of partial plays in $\game$. 

Let $\lambda_{1} \neq \lambda_{2}  \in \outcomes{v_0}{\phi_\Ei}{\phi_\Ai} \setminus \WC$. We will show that $\pi_{\lambda_{1}} \neq \pi_{\lambda_{2}}$.
Consider the greatest common prefix $\pi$  of $\lambda_{1}$ and $\lambda_{2}$.  As $\lambda_{1}$ and $\lambda_{2}$
respects the same strategies for \Eloise and \Abelard, $\pi$ must end by some $v \gamma$ with $v$ of \Eloise. In particular there exists $v_{1} \neq v_{2} \in \Delta_{\Ei}(v,\gamma)$ such that $\pi v_{1} \prefixstrict \lambda_{1}$ and $\pi v_{2} \prefixstrict \lambda_{2}$. The partial play $\hat{\pi}$ ends in
$v (\gamma,\theta)$ for some $\theta \in \Theta_{\gamma}$. Assume w.l.o.g. that $\theta(v) \neq v_{1}$.
\Ie \Abelard obeys \Eloise at $\widehat{\pi v_{1}}$ in $\lambda_{1}$. In particular,
$\pi v_{1} \prefix \pi_{\lambda_{1}}$: therefore $\pi_{\lambda_{1}} \not\prefix \pi_{\lambda_{2}}$ and thus $\pi_{\lambda_{1}} \neq \pi_{\lambda_{2}}$.

For the converse implication, as the game $\gameNR$ may not be determined, we cannot proceed as in the perfect-information case\footnote{Recall that in the perfect-information case, for the converse implication of the proof of Theorem~\ref{theo:perfect:main-aleph} we were, thanks to determinacy, considering a winning strategy for \Abelard in $\gameL$ and built out of it a winning strategy for him in $\game$.}. Hence, assume that the leaking value of $\game$ is at most $\aleph_0$ and let $\strat_\Ei$ be a strategy of \Eloise such that $\LVal{\strat_\Ei}\leq\aleph_0$ (thanks to Remark~\ref{remark:supNotMaxCardleakingImperfect} it exists).

In order to define a winning strategy in $\gameL$ for \Eloise,
we consider a \emph{perfect-information} parity game with Nature that we denote $\gameP$.
In this game each vertex belongs either to \Abelard or \Nature. 

To define $\gameP$, let $S\subseteq V^*$ be the set of all $\encod{\pi}$ for $\pi$ a partial play respecting $\phi_{\Ei}$ and let $\equiv$ be the equivalence relation  on $S$ defined for all $\encod{\pi}$,$\encod{\pi'} \in S$ by
$\encod{\pi} \equiv \encod{\pi'}$ if $\pi$ and $\pi'$ end in the same vertex and 
$\pi \sim \pi'$.
In the rest of the proof we will use letter $\eta$,$\eta'$,… to denote elements in $S$. In particular if $\eta=v_0\cdots v_k$ and $\eta'=v'_0\cdots v'_{k'}$ we have $\eta\equiv\eta'$ if and only if $k=k'$, $v_k=v'_{k'}$ and $v_i\sim v'_i$ for all $0\leq i<k$.

The vertices $V_{\gameP}$ are the equivalence classes  $\equiv$.
 A vertex $[\eta]_{/\equiv} \in V_{\gameP}$ belongs to \Abelard if $\eta$ ends in a vertex of \Abelard and it belongs to \Nature otherwise. There is an edge from $[\eta]_{/\equiv}$ to $[\eta']_{/\equiv}$ if $\eta'$ extends $\eta$ by one vertex. The initial vertex
is $[v_{0}]_{/\equiv}$. Lastly, the parity condition is given by the mapping associating to $[\eta]_{/\equiv} \in V_{\gameP}$ the colour
$\col(v)$ of the  the last vertex $v$ of $\eta$.

A partial play $\xi$ in $\gameP$ is of the form $\xi=[\eta_{0}]_{/\equiv} [\eta_{1}]_{/\equiv} \cdots [\eta_{n}]_{/\equiv}$
where $\eta_{0}=v_{0}$ and for all $i <n$, $\eta_{i+1}$ extends $\eta_{i}$ by one vertex. With such a $\xi$, we naturally associate the play $\tau(\xi)$ in $\game$ defined by
$v_{0} \nu_{0} v_{1} \nu_{1} \cdots v_{n}$ where for all $i \geq 0$,
$v_{i}$ is defined as the last vertex in $\eta_{i}$  and $\nu_{i}$ is equal to $\phi_{\Ei}([\eta_{i}]_{/_\sim})$ if $v_{i}$ belongs to \Eloise and $\nu_{i}$ is empty otherwise. It is easy to
show that for all $i \leq n$, $v_{0} v_{1} \cdots v_{i} \equiv \eta_{i}$.
Hence, as $\phi_{\Ei}$ is observation-based, $\tau(\xi)$ respects $\phi_{\Ei}$.
In fact, the continuous mapping $\tau$ establishes 
a one to one correspondance between the partial plays in $\gameP$ and
the partial plays in $\game$ where \Eloise respect $\phi_{\Ei}$. By continuity, this mapping extends to plays.

The leaking value of $\gameP$ is at most $\aleph_0$: indeed, any strategy $\phi_{\Ai}^{\gameP}$ for \Abelard in $\gameP$ can be lifted to a strategy $\phi_{\Ai}$ in $\game$
such that $\{ \tau(\xi) \mid \xi \;\text{a play in $\gameP$ which
respects $\phi_{\Ai}^{\gameP}$}\}$ is equal to $\outcomes{v_{0},\game}{\phi_{\Ei}}{\phi_\Ai}$, and therefore (as $\LVal{\strat_\Ei}\leq\aleph_0$) the set of losing plays for \Eloise in $\gameP$ when \Abelard uses strategy $\phi_{\Ai}^{\gameP}$ has cardinality at most $\aleph_0$.

Therefore one can use Theorem~\ref{theo:perfect:main-aleph} for the (perfect-information) game $\gameP$ and gets that \Eloise has a winning strategy in 
the game $\gamePNR$ (defined as in Theorem~\ref{theo:perfect:main-aleph}). As the winning condition of $\gamePNR$ is a disjunction of two parity conditions, the winning condition of $\gamePNR$ is a so-called Rabin condition\footnote{We refer the reader not familiar with Rabin conditions to \cite{Thomas97} for a formal definition. Let us also stress that the Rabin condition is in fact on the sequence of edges taken during the play and not on sequence of vertices visited. By a slight modification of $\gamePNR$, it can be transformed into a Rabin condition on the sequence of vertices visited.}.
Therefore \Eloise has a \emph{positional} winning strategy $\phi_{\Ei}^{{\gamePNR}}$ in $\gamePNR$ \cite{Klarlund94}.
For $\eta \in S$ ending with a vertex $v$ of \Eloise,   $\phi_{\Ei}^{{\gamePNR}}$ associates to $[\eta]_{/\equiv}$ a pair $([\eta]_{/\equiv},[\eta v']_{/\equiv})$  with $v'
\in \Delta_{\Ei}(v,\phi_{E}([\eta]_{/_\sim}))$. This strategy is equivalently described by the mapping $\phi_{B}$ associating to $[\eta]_{/\equiv}$  the vertex $v'$ in $\Delta_{\Ei}(v,\phi_{E}([\eta]_{/_\sim}))$.

The key property of this strategy is that any play $\lambda$ in
$\game$ which respects $\phi_{\Ei}$ and such that $\lambda$ has
infinitely many prefixes of the form $\pi v \gamma v'$ with $v \in \VE$
and $v' \neq \phi_{B}([\encod{\pi v}]_{/\equiv})$, satisfies the parity
condition. Indeed, toward  a contradiction assume that $\lambda$ does not
satisfy the parity condition. Let $\lambda'=\tau^{-1}(\lambda)$ be the
corresponding play in $\gameP$ and let $(\lambda',\phi_{B})$ be the
corresponding play in $\gamePNR$. None of these plays satisfies the
parity condition. However as $(\lambda',\phi_{B})$ respects the
positional winning strategy for \Eloise described by $\phi_{B}$, it is
won by \Eloise. This implies that \Abelard only obeys \Eloise finitely
often which brings the contradiction.

In order to define a strategy for \Eloise in $\gameNR$ we will mimic $\strat_\Ei$ to choose the $\Gamma$-component (call $\gamma$ the action) and use $\strat_B$ to choose the $\Theta_\gamma$-component.

For this we let $\phiNR_\Ei([\pi]_{/_\sim}) = (\gamma,\theta)$ where $\gamma=\strat_\Ei([\pi]_{/_\sim})$ and $\theta$ is defined as follows. Let $v\in V$: if there exists ${\pi' \sim \pi}$ ending with $v$ we take
$\theta(v)=\phi_{B}([\encod{\pi'}]_{/_\equiv})$; otherwise we define $\theta(v)=w$ for some arbitrary $w\in \Delta(v,\gamma)$ (the value actually does not matter).

Now consider a play $\playNR=v_0\nu_0v_1\nu_1v_2\cdots$ in $\gameNR$ where \Eloise respects $\phiNR_{\Ei}$, denote $\nu_i=(\gamma_i,\theta_i)$ when $\nu_i\neq \epsilon$ (\ie when $v_i\in V_\Ei$) and define $\gamma_i=\epsilon$ when $\nu_i=\epsilon$. 
By contradiction assume that $\playNR$ is losing for \Eloise. 
Consider the play $\play=v_0 \gamma_0 v_1 \gamma_1 v_2\cdots$ : it is a play in $\game$ where \Eloise respects $\phi_{\Ei}$ and as $\playNR\notin \WCNR$ one also has $\play\notin\WC_\col$. But as $\playNR$ is losing for \Eloise it means that for infinitely many $i$ one has $v_{i+1}\neq \theta_i(v_i)$, which implies that for infinitely many $i$ one has $v_{i+1} \neq \phi_B({[v_0\cdots v_i]_{/_\equiv}})$. Therefore as remarked previously, it implies that $\lambda\in\WC_\col$ hence, leading a contradiction.
\end{proof}

Note that, due to page limit constraints, contrarily to what we did in Section~\ref{section:perfectLeakingFinite} for the perfect-information setting we do not tackle here the problem of deciding whether the leaking value is smaller than some given $k$. However, we hope that it is clear that the approach we just developed, for the problem of deciding whether the leaking value is at most $\aleph_0$, to shift from perfect to imperfect-information, also leads to treat the other question as well.

\subsection{Deciding the Existence of a topologically-Good Strategy}

Our goal in this section is to design a technique to decide whether \Eloise has a topological good strategy in an imperfect-information parity game with Nature. We only have results in the case of games where \Abelard is not playing (\ie one-player game with Nature) hence, we implicitly assume this from now.

We start by giving a useful characterisation of large sets of branches in a tree when the set of branches is defined by a parity condition. For this fix a $D$-tree $t$ for some set of directions $D$. Assume that we have a colouring function $\col:t\rightarrow \colors$ for a finite set $\colors$ of colours. 

Call a \defin{local-strategy} for \Eloise a pair $(\phi_f,\phi_n)$ of two maps from $t$ into $D\times \{\top,\bot\}$. For all node $u\in t$, we let $d_f(u)$ (\resp $d_n(u)$) be the unique element such that $\phi_f(u) \in\{d_f(u)\}\times\{\top,\bot\}$ (\resp $\phi_n(u)\in\{d_n(u)\}\times\{\top,\bot\}$).

A local-strategy is \emph{valid} if the following holds.
\begin{enumerate}
\item For every $u\in t$ both $u\cdot d_f(u)$  and $u\cdot d_n(u)$ are nodes in $t$; \ie $\phi_f$ and $\phi_u$ indicates an existing son.
\item For every $u\in t$ there is a node $v=ud_1\cdots d_\ell$ such that $\phi_f(v)\in D\times \{\top\}$ and $d_i = d_f(ud_1\cdots d_{i-1})$ for all $i<\ell$; \ie following $\phi_f$ leads to a node where the second component is $\top$.
\item For every $u\in t$ there is a node $v=ud_1\cdots d_\ell$ such that $\phi_n(v)\in D\times \{\top\}$ and $d_i = d_n(ud_1\cdots d_{i-1})$ for all $i<\ell$; \ie following $\phi_n$ leads to a node where the second component is $\top$.
\end{enumerate}

Take a valid local-strategy $(\phi_f,\phi_n)$. A $(\phi_f,\phi_n)$-compatible branch is any branch in $t$ that can be obtained as follows: one selects any node $u_0$ in $t$ and then one lets $v_0$ be the shortest node satisfying property (2) above (w.r.t. node $u_0$), then one selects any node $u_1$ such that $v_0\prefixstrict u_1$ and one lets $v_1$ be be the shortest node satisfying property (3) above (w.r.t. node $u_1$), then one selects any node $u_2$ such that $v_1\prefixstrict u_2$ and one lets $v_2$ be the shortest node satisfying property (3) above (w.r.t. node $u_2$), and so on forever (\ie we use property (2) only in the first round and then we use property (3) forever). 

We have the following lemma (whose proof follows the one of \cite[Proposition~13]{Graedel08}).

\begin{lemma}\label{lemma:local-strategy}
The set of branches satisfying the parity condition in $t$ is large if and only if there is a valid local-strategy $(\phi_f,\phi_n)$ such that any $(\phi_f,\phi_n)$-compatible branch satisfies the parity condition.
Moreover one can choose $(\phi_f,\phi_n)$ such that $\phi_f(u_1) = \phi_f(u_2)$ and $\phi_n(u_1) = \phi_n(u_2)$ whenever $t[u_1] = t[u_2]$.
\end{lemma} 

\begin{proof}
We rely on the characterisation of large sets by means of Banach-Mazur games. 

Obviously if there is a valid local-strategy $(\phi_f,\phi_n)$ such that any $(\phi_f,\phi_n)$-compatible branch satisfies the parity condition, then it leads a winning strategy for \Eloise in the Banach-Mazur game. Indeed, for her first move \Eloise goes down in the tree using $\phi_f$ until she ends up in a node whose father's second component by $\phi_f$ was $\top$ and in the next rounds she does similarly but using $\phi_n$. The resulting play is a $(\phi_f,\phi_n)$-compatible branch hence, satisfies the parity condition.

We now prove the other implication, \ie we assume that the set of branches satisfying the parity condition in $t$ is large or equivalently that \Eloise wins the Banach-Mazur game. The beginning of the proof is very similar to the one that Banach-Mazur games with Muller winning condition admit positional strategies \cite[Proposition~13]{Graedel08}. Let $u$ be a node in $t$ then one denotes by $C(u)=\{\col(v)\mid u\prefix v\}$ the set of colours of nodes reachable from $u$ in $t$. Obviously one has $C(w)\subseteq C(u)$ for all $u\prefix w$. In case one has  $C(w)= C(u)$ for all $u\prefix w$ we say that $u$ is a \emph{stable} node (and so does its descendants). 
As the set of colours is finite, for all node $u$ there is a stable node $v$ such that $u\prefix v$.

We claim that for all stable node $u$, $\min C(u)$ is even. Indeed, assume that there is some stable $u$ such that $\min C(u) = m$  is odd: then a winning strategy (leading a contradiction) of \Abelard in the Banach-Mazur game would consist to go to $u$ in its first move and then whenever he has to play to go to a node with colour $m$ (which he can always do by stability).

Now we define a valid local-strategy $(\phi_f,\phi_n)$ as follows. First, fix a total ordering on $D$. For every $u\in t$, call $u_s$ the unique minimal (for the length lexicographic ordering) stable node such that $u\prefixstrict u_s$: define $\phi_f(u) = (d,x)$ where $u_s=u\cdot d\cdot w$ with $d\in D$ and $x=\top$ if $w=\epsilon$ and $x=\bot$ otherwise. For every $u\in t$ that is stable, call $u'$ the unique minimal (for the length lexicographic ordering)  node with colour $\min C(u)$ and such that $u\prefixstrict u'$: define $\phi_n(u) = (d,x)$ where $u'=u\cdot d\cdot w$ with $d\in D$ and $x=\top$ if $w=\epsilon$ and $x=\bot$ otherwise. For every $u\in t$ that is not stable define $\phi_n(u) = (d,\bot)$ where $d$ is the minimal direction such that $ud\in t$ (the value of $\phi_n$ does not matter but we want it to be the same in all isomorphic subtrees so we have to define it in a systematic way). From the definition one directly gets that $\phi_f(u_1) = \phi_f(u_2)$ and $\phi_n(u_1) = \phi_n(u_2)$ whenever $t[u_1] = t[u_2]$.

The fact that $(\phi_f,\phi_n)$ is valid is by definition and the fact that any $(\phi_f,\phi_n)$-compatible branch satisfies the parity condition is a direct consequence of the fact that for all stable node $u$ one has $\min C(u)$ even.
\end{proof}

Recall that we assume that \Abelard is not part of the game. Hence, we omit him in notations when considering the original game (\ie we do not write $\VA$ nor $\Delta_\Ai$).

For the rest of this section we fix an imperfect-information one-player parity game with nature $\game=(\arena,\col,v_0)$ where $\arena=(V,\Gamma,\Delta,\sim)$ and we aim at deciding whether \Eloise has a topologically-good strategy.

The main idea is to define an imperfect-information game without \Nature but with \Abelard. In this game \Eloise simulates a play in $\game$ and also describes a local-strategy for a Banach-Mazur game played on the outcomes; \Abelard  is in charge of simulating the Banach-Mazur game: sometimes he chooses the directions and sometimes he plays what the local-strategy of \Eloise is indicating. Moreover \Eloise does not observe who is currently playing in the Banach-Mazur game. The winning condition checks the parity condition as well as correctness of the simulation of the Banach-Mazur game (in particular that no player plays eventually forever).

In order to describe the local-strategy, \Eloise will provide with any action $\gamma\in \Gamma$ a partial map $\theta:V\rightarrow (V\times\{\top,\bot\})\times (V\times\{\top,\bot\})$ such that for all $v\in V$ one has $\theta(v)\in \Delta(v,\gamma)\times\{\top,\bot\}\times\Delta(v,\gamma)\times\{\top,\bot\}$; we denote by $\Theta_\gamma$ the set of such maps (for a given $\gamma\in \Gamma$).

We define a two-player imperfect-information arena (all vertices belong to \Eloise so we omit vertices and the transition relation of \Abelard)
$\arenaT=(\VT,\GammaT,\DeltaT,\simT)$
where $\VT = V\times\{E,A\}\times\{f,n\}$ (the second component is used to remember who plays in the simulation of the Banach-Mazur game; the third component is $f$ if the first move of \Eloise in the Banach-Mazur game has not yet been fully played), $(v,X,x)\simT(v',Y,y)$ if and only if $v\sim v'$ (\Eloise does not observe the second and third components),  $\GammaT=\bigcup_{\gamma\in \Gamma}\{\gamma\}\times\Theta_\gamma$ and  $\DeltaT((v,X,x),(\gamma,\theta))$ is as follows.
\begin{itemize}
\item If $X=A$ then it equals $\Delta(v,\gamma) \times\{E,A\}\times\{x\}$: \Abelard can choose any successor and can decide to finish/continue his move in the Banach-Mazur component.
\item If $X=E$ then it is the singleton consisting of node $(v_x,Y,y)$ defined by letting\footnote{In case $\theta(v)$ is undefined \Eloise looses the play. We assume this never happens but it can easily be captured in the winning condition by adding an extra vertex.} $\theta(v) = (v_f,y_f,v_n,y_n)$ and letting $Y=A$ and $y=n$ if $y_x=\top$ (we switch the player in the Banach-Mazur game) and $Y=E$ and $y=x$ if $y_x=\bot$ (she keeps playing). 
\end{itemize}

We let $\WCT$ consists of those plays $(v_0,X_0,x_0)(v_1,X_1,x_1)(v_2,X_2,x_2)\cdots$ such that either \begin{inparaenum}[(i)]\item $v_0v_1v_2\cdots$ satisfies the winning condition and one has $X_j=A$ for infinitely many $j$ (\ie \Eloise does not eventually play forever in the Banach-Mazur game) or \item there is some $N\geq 0$ such that one has $X_j=A$ for all $j\geq N$  (\ie \Abelard eventually plays forever in the Banach-Mazur game) 
\end{inparaenum}. In particular $\WCT$ is a (positive) Boolean combination of $\WC$ and a parity condition.

Finally we denote by $\gameT$ the imperfect-information game $(\arenaT,\WCT,(v_0,A,f))$.
The following relates the  games $\game$ and $\gameT$.

\begin{theorem}\label{thm:topo-imperfect}
\Eloise has a topologically-good strategy in $\game$ if and only if she has a winning strategy in $\gameT$. 

More precisely, from a winning strategy (\resp positional strategy, \resp finite-memory strategy) $\phiT_{\Ei}$ of \Eloise in $\gameT$, we can define a topologically-good strategy (\resp positional strategy, \resp finite-memory strategy) $\phi_{\Ei}$ for \Eloise in $\game$.
\end{theorem}

\begin{proof}
Strategies $\phiT$ for \Eloise in $\gameT$ are in bijections with pairs made of a strategy $\phi$ in $\game$ together with a local-strategy $(\phi_f,\phi_n)$ in the tree of the outcomes of $\phi$ in $\game$. Now if $\phiT$ is winning in $\gameT$ we have thanks to the second part of $\WCT$ that the local-strategy $(\phi_f,\phi_n)$ is valid, and thanks to the first part of $\WCT$ that any compatible play is winning for \Eloise in the Banach-Mazur game. Hence, it implies that $\phi$ is topologically good (the set of winning plays in $T_{v_0}^\phi$ is large). Obviously $\phi$ does not require more memory than $\phiT$.

Conversely if \Eloise has a topologically good strategy $\phi$ in $\game$ we can associate with $\phi$ a local-strategy $(\phi_f,\phi_n)$ as in Lemma~\ref{lemma:local-strategy} (applied to $T_{v_0}^\phi$). Using $\phi$, $\phi_f$ and $\phi_n$ we define a winning strategy $\phiT$ for \Eloise in $\gameT$ as follows. 
We let 
$\phiT([(v_0,X_0,x_0)(v_1,X_1,x_1)\cdots(v_k,X_k,x_k)]_{/_\simT})=(\gamma,\theta)$ where $\gamma=\phi([v_0v_1\cdots v_k]_{/_\sim})$ and $\theta$ is defined as follows. Let $v\in V$: if there is no 
$v'_0\cdots v'_k\in T_{v_0}^\phi\cap [v_0\cdots v_k]_{/_\sim}$ with $v_k=v$ we let $\theta(v)$ undefined; otherwise choose such a $v'_0\cdots v'_k$ (the representative actually does not matter thanks to the fact that $(\phi_f,\phi_n)$ is the same in isomorphic subtrees) and define $\theta(v) = (\phi_f(v'_0\cdots v'_k),\phi_n(v'_0\cdots v'_k))$.

Now consider a play $\playT=(v_0,X_0,x_0)(v_1,X_1,x_1)(v_2,X_2,x_2)\cdots$ in $\gameT$ where \Eloise respects $\phiT$. If there are infinitely many $i$ such that $X_i=E$ then there are infinitely many $j$ such that $X_j=A$ (this is because $(\phi_f,\phi_n)$ is valid). Moreover if there are infinitely many $i$ such that $X_i=E$ then the play $v_0v_1v_2\cdots$ is a branch in $T_{v_0}^\phi$ that is $(\phi_f,\phi_n)$-compatible and therefore it satisfies the parity condition by Lemma~\ref{lemma:local-strategy} and definition of $(\phi_f,\phi_n)$. Therefore, the strategy $\phiT$ is winning for \Eloise in $\gameT$ and it concludes the proof.\end{proof}

\section{Imperfect-Information Games with Nature: Some Consequences}\label{section:consequences-imperfect}

\subsection{Imperfect-Information Parity Games on Finite Graphs}

For imperfect-information, in the case of finite arena, as soon as one considers co-Büchi conditions almost-sure winning is undecidable even for \Eloise-\Nature game where \Eloise is totally blind (all vertices are equivalent)~\cite{BGB12}. Therefore, both the cardinality setting and the topological one are interesting alternative to retrieve decidability: indeed, thanks to Theorem~\ref{theo:main-partial} (\resp Theorem~\ref{thm:topo-imperfect}) combined with the results in \cite{CDHR07} we get decidability for finite arena for any parity condition. 

\subsection{Imperfect-Information Parity Games on Infinite Finite Graphs}

There is very few work in the probabilistic setting about games with imperfect-information played on \emph{infinite} arenas. The notable exception is the case of concurrent reachability games played on single-exit  recursive state machines%
\footnote{Concurrency is a special instance of imperfect-information where \Abelard is perfectly informed: he chooses an action which is stored on the state and cannot be observed by \Eloise who next chooses an action that together with the one by \Abelard leads to the next state (chosen by \Nature). Recursive state machines are equivalent with pushdown automata; however the single exit case quite strongly restricts the model.} 
for which impressive results where obtained in~\cite{EtessamiY08}. In the non-stochastic setting, it is easy to derive decidability results for \emph{parity} game played on pushdown graphs when \Eloise perfectly observes the stack content but not the exact control state and \Abelard is perfectly informed (see \eg \cite{AminofLMSV13}); this result can easily be extended for more general classes of graphs as collapsible pushdown graphs as defined in \cite{HMOS08}. Hence, thanks to Theorem~\ref{theo:main-partial} and \ref{thm:topo-imperfect} we obtain decidability results for games with \Nature played on those classes of infinite arenas. Note that in the cardinality setting, even if we require that \Abelard has perfect-information our model captures concurrent games.

\subsection{Probabilistic Automata}

A temptation would be to consider cardinality/topological variants of probabilistic automata on infinite words~\cite{BGB12} as such a machine can be though as an \Eloise-\Nature game where \Eloise is totally blind: \eg declare that an $\omega$-word is accepted by an automaton if all but a countable number of runs on it are accepting (\resp the set of accepting runs is large). However, a simple consequence of our results is that the languages defined in this way are always $\omega$-regular.

\section{Conclusion and Perspectives}\label{section:conclusion}

In this paper we provided several reductions that can later be used to obtain decidability results depending on the properties of both the arena and the winning condition. 
More precisely, we give transformations that associate with any game with \Nature $\game=(\arena,\WC,v_0)$ a game \emph{without} \Nature $\game'=(\arena',\WC',v'_0)$ on which the question (Is the leaking value is countable? Is the leaking value smaller than some threshold $k$? Is there a topologically good strategy? ) on the original game $\game$ is restated as whether Eve has  a winning strategy in $\game'$. In all cases the new arena $\arena'$ is obtained from $\arena$ by adding some gadgets, and the new winning condition $\WC'$ is a Boolean combination of $\WC$ with an $\omega$-regular condition. Moreover, for some cases, we need extra hypothesis that we recall in the table below.

\bigskip

\renewcommand{\arraystretch}{1.3}
\noindent\begin{tabularx}{\textwidth}{m{3.7cm}m{5cm}m{5cm}}
\toprule
&  $\LVal{\game} \leq\aleph_0$?\newline $\LVal{\game} \leq k$? & Topologically good?\\ \midrule
\emph{Perfect-information} & & \\  & No extra hypothesis on $\game$\newline No extra hypothesis on $\WC$ & Eve + \Nature only\newline No extra hypothesis on $\WC$\\ \midrule
\emph{Imperfect-information} & & \\ & Adam perfect\newline $\WC$: parity & Eve + \Nature only\newline $\WC$: parity\\ \bottomrule
\end{tabularx}

\medskip

Regarding perspectives the most natural question is whether we can drop the restriction on \Abelard not being part of the game for questions regarding the topological setting. 

Another exciting problem is whether one can decide if the leaking value of a game is finite (without knowing \emph{a priori} the bound). We believe that this problem should be decidable for parity games on finite graph but using different techniques than the one developed in this paper.

\begin{acks}
The authors are indebted to Damian Niwi\'nski for suggesting to investigate the notion of \emph{goodness} of a strategy. 

The authors also would like to thank Thomas Colcombet who simplified the proof of Theorem~\ref{theo:main-partial} as well as Pawe\l{l} Parys for pointing that existence of $\mathcal{L}$-good strategies can be expressed in MSO logic as soon as $\mathcal{L}$ is MSO-definable.
\end{acks}

\bibliographystyle{ACM-Reference-Format}
\bibliography{abbrevs,CS19}


\end{document}